\documentclass[reqno]{amsart}

\usepackage{amssymb}
\usepackage{mathabx}
\usepackage[utf8]{inputenc} 
\usepackage[english]{babel}
\usepackage{amstext}
\usepackage{euscript}
\usepackage{bbm}
\usepackage{mathrsfs}
\usepackage{enumerate}
\usepackage{graphicx}
\usepackage{caption}
\usepackage{amscd}
\usepackage{verbatim}
\usepackage{amssymb}
\usepackage{amsthm}
\usepackage{amsmath}
\usepackage{amstext}
\usepackage{amsfonts}
\usepackage{latexsym}
\usepackage{mathtools} 
\usepackage{enumitem} 
\usepackage{accents} 
\usepackage[foot]{amsaddr} 
\usepackage{textcomp} 
\usepackage{cite}

\usepackage[dvipsnames]{xcolor}

\definecolor{dkblue}{RGB}{1,31,91} 
\definecolor{dpblue}{HTML}{000f3a} 
\usepackage[colorlinks=true, pdfstartview=FitV, linkcolor=dkblue, citecolor=dkblue, urlcolor=dkblue]{hyperref}

\usepackage{orcidlink}

\newcommand{\eqdef }{\overset{\mbox{\tiny{def}}}{=}}
\newcommand{\pv}{p}
\newcommand{\pZ}{|\pv|}
\newcommand{\qv}{q}
\newcommand{\qZ}{|\qv|}

\newcommand{\pM}{p^{\mu}}
\newcommand{\qM}{q^{\mu}}

\newcommand{\pZprime}{\pv^{\prime 0}}
\newcommand{\qZprime}{\pv^{\prime 0}}
\newcommand{\pZzero}{\pv^{0}}
\newcommand{\qZzero}{\qv^{0}}

\newcommand{\rth}{{\mathbb{R}^3}}

\newcommand{\secref}[1]{Section~\ref{#1}}

\def\norm#1{\big\Vert#1\big\Vert}
\def\R{\mathbb R}

\newcommand{\vnrm}[1]{\left\vert#1\right\vert}
\newcommand{\irnp}[2]{\left\langle {#1},{#2} \right\rangle}

\newcommand{\threed}{{\mathbb R}^3}

\newcommand{\qq}{\mathfrak{q}}
\newcommand{\tq}{t^{\qq}}
\newcommand{\tmq}{t^{-\qq}}
\newcommand{\tqq}{t^{2\qq}}
\newcommand{\tTq}{t^{3\qq}}

\newcommand{\tmQq}{t^{-6\qq}}
\newcommand{\tQQq}{t^{6\qq}}
\newcommand{\tmTq}{t^{-3\qq}}

\newcommand{\compK}{\mathcal{K}}
\newcommand{\linL}{\mathcal{L}}

\newcommand{\macro}{{\bf P}}
\newcommand{\micro}{{\bf \{I- P\}}}

\newcommand{\Emptyset}{\text{\o}}

\newcommand{\kerU}{\mathcal{U}}

\newcommand{\smep}{\varepsilon}

\newcommand{\ind}{{\bf 1}}

\newcommand{\highG}{\Theta}

\newcommand{\mrel}{\varrho}

\newcommand{\CN}{\mathrm{nul}}

\newcommand{\Ga}{\Gamma}

\newcommand{\de}{\delta}

\newcommand{\T}{\mathbb{T}}
\newcommand{\Z}{\mathbb{Z}}

\newcommand{\Si}{\Sigma}

\newcommand{\real}[1]{\mathfrak{Re}\left(#1\right)}
\newcommand{\imag}[1]{\mathfrak{Im}\left(#1\right)}

\newcommand{\spr}{\zeta}

\newcommand{\tc}{\vartheta}
\newcommand{\tw}{\mathcal{T}_{\qq}}

\DeclareMathOperator*{\esssup}{ess\,sup}

\newcommand{\wb}{\mathfrak{b}}

\newcommand{\tiq}{\mathcal{S}_{\qq}}

\newcommand{\sourceS}{\mathcal{H}}

\newcommand{\hh}{\mathbbm{W}}

\newcommand{\dr}{d}

\theoremstyle{definition}
\newtheorem{theorem}{Theorem}
\newtheorem{conjecture}[theorem]{Conjecture}
\newtheorem{corollary}[theorem]{Corollary}
\newtheorem{lemma}[theorem]{Lemma}
\newtheorem{proposition}[theorem]{Proposition}
\newtheorem{remark}[theorem]{Remark}

\numberwithin{equation}{section}
\numberwithin{theorem}{section}

\allowdisplaybreaks

\begin{document}


\title[Global stability of the massless Boltzmann equation on FLRW]{Future global stability of Maxwell--J{\"u}ttner equilibria and vacuum for the massless Boltzmann equation on FLRW spacetimes} 

\author[R. M. Strain]{Robert M. Strain$^{\ddagger}$}
\address{$^\ddagger$Department of Mathematics, University of Pennsylvania, Philadelphia, PA 19104, USA. \href{mailto:strain@math.upenn.edu}{strain@math.upenn.edu} (\orcidlink{0000-0002-1107-8570} \href{https://orcid.org/0000-0002-1107-8570}{https://orcid.org/0000-0002-1107-8570})}
\thanks{$^\ddagger$Partially supported by the NSF grant DMS-2408264 of the USA}

\author[M. Taylor]{Martin Taylor$^{\dagger}$}
\author[R. Velozo Ruiz]{Renato Velozo Ruiz$^{\dagger}$}
\address{$^\dagger$Imperial College London,
Department of Mathematics,
South~Kensington~Campus,~London~SW7~2AZ,~United~Kingdom. \href{mailto:martin.taylor@imperial.ac.uk}{martin.taylor@imperial.ac.uk} (\orcidlink{0000-0001-5747-5267} \href{https://orcid.org/0000-0001-5747-5267}{https://orcid.org/0000-0001-5747-5267}), \href{mailto:r.velozoruiz@imperial.ac.uk}{r.velozoruiz@imperial.ac.uk} (\orcidlink{0009-0002-5337-2118} \href{https://orcid.org/0009-0002-5337-2118}{https://orcid.org/0009-0002-5337-2118}).}
\thanks{$^\dagger$Partially supported through Royal Society Tata University Research Fellowship URF\textbackslash R1\textbackslash191409.}

\begin{abstract}
In this work we study the general relativistic massless Boltzmann equation on Friedmann--Lema\^itre--Robertson--Walker spacetimes with spatial topology \(\mathbb{T}^3\) in the linear and decelerated expanding regimes, where the scale factor is \(t^{\mathfrak{q}}\) with \(\mathfrak{q}\in [0,1]\). The massless Boltzmann equation on these backgrounds admits non-stationary Maxwell--J\"uttner equilibria of the form $\exp(- |t^{2\mathfrak{q}}p|)$.   For $0 \leq \mathfrak{q} \leq 1$, we prove future global-in-time existence and uniqueness of small perturbations of these equilibria in the case of hard ball interaction without symmetry assumptions. For $0\leq  \mathfrak{q} < 1/3$, we prove that the perturbation --- measured in a suitable \(L^2_p\) based energy norm --- decays at the superpolynomial time-decay rate of $t^{-3\mathfrak{q}}\exp(-t^{1-3\mathfrak{q}})$, whereas for $1/3< \mathfrak{q} \leq  1$ we obtain the polynomial time-decay rate of $t^{-3\mathfrak{q}}$. In the borderline case \(\mathfrak{q}=1/3\), we show the time-decay of \(t^{-3\mathfrak{q} -c}\) with a uniform constant \(c>0\). Finally, for \(\frac{1}{3}< \mathfrak{q}\leq 1\), we prove future global-in-time existence and uniqueness of small perturbations of the vacuum solution on \(\mathbb{T}^3\).
\end{abstract}

\thispagestyle{empty}

\maketitle
\setcounter{tocdepth}{2}
\tableofcontents

\section{Introduction}

In this article, 
we will study the long time dynamics of the general relativistic massless Boltzmann equation on homogeneous and isotropic cosmological models of the universe described by the Friedmann--Lema\^itre--Robertson--Walker (FLRW) spacetimes
\begin{equation}\notag
\mathcal{M}=I\times \Sigma,\qquad g=-dt\otimes dt+a^2(t)g_{\Sigma},  
\end{equation}
where \(I\subset \R\) is an open interval, \((\Sigma,g_{\Sigma})\) is a constant curvature manifold, and \(a\colon I\to (0,+\infty)\) is called the scale factor.  This article studies the future global non-linear stability of Maxwell--J\"uttner equilibria on decelerated and linearly expanding FLRW spacetimes in which the constant curvature manifold is the 3-dimensional flat torus \(\mathbb{T}^3\). To the best of our knowledge, this is the first future global-in-time stability result for the Boltzmann equation in the presence of spatial dependence and a non-trivial gravitational field.

\subsection{The Boltzmann equation in general relativity}

Consider a \(4\)-dimensional Lorentzian manifold \((\mathcal{M},g)\), and let \(\mathcal{P}\subset T\mathcal{M}\) be the \emph{mass-shell} of spacetime defined by 
\[\mathcal{P}:=\big\{(x,p)\in T\mathcal{M}: g_x(p,p)=-m^2c^2\,\,\, \text{with $p$ future-directed} \big\},\] where \(c>0\) is the speed of light, and \(m\geq 0\) is a fixed parameter describing the mass of the particles under consideration.  The general relativistic Boltzmann equation on \((\mathcal{M},g)\) concerns functions $F\colon \mathcal{P} \to [0,\infty)$.

Consider a local coordinate system \((t=x^0, x^1, x^2 , x^3)\) on \(\mathcal{M}\), and let \((x^{\mu},p^{i})\) be the corresponding conjugate coordinate system on \(\mathcal{P} \subset T\mathcal{M}\), so that \((x^{\mu},p^{i})\) corresponds to the point \(p^{\mu}\partial_{x^{\mu}}|_{x}\in \mathcal{P} \subset T\mathcal{M}\), where $p^0$ is defined in terms of \((x^{\mu},p^{i})\) by the \emph{mass shell relation}
\begin{equation}\label{mass-shellChristoff}
    g_{\mu \nu} p^{\mu} p^{\nu} = -m^2c^2.
\end{equation}
We use here the Einstein summation convention over repeated indices and adopt the convention that Greek indices range over 0, 1, 2, 3, and Latin indices range over 1, 2, 3.  Spacetime points are typically denoted $(t,x) \in \mathcal{M}$, though in some expressions, which will be clear from the context, $x\in \mathcal{M}$ is used for brevity.  By a slight abuse, we will use the notation $p$ both for points in $\mathcal{P}_x$, for fixed $x\in \mathcal{M}$, and for $p=(p^1,p^2,p^3)$ to denote elements of $\mathbb{R}^3$.  Note that the latter parameterizes the former, with $p^0>0$ defined by \eqref{mass-shellChristoff}.  We will also write $p^{\mu}$ to denote the $\mu$ component of  $(p^0, p^1, p^2, p^3)$.

The Boltzmann equation on \((\mathcal{M},g)\) then takes the form 
\begin{equation}\label{genrelatbolzt}
p^0\partial_t F+p^i\partial_{x^i}F-p^{\mu} p^{\nu}\Gamma_{\mu\nu}^i\partial_{p^i}F=C(F,F),    
\end{equation}
where
\[
	\Gamma_{\mu \nu}^{\alpha} = \frac{g^{\alpha \beta}}{2} \big( \partial_{\mu}g_{\beta \nu} + \partial_{\nu}g_{\mu \beta} -  \partial_{\beta}g_{\mu \nu} \big),
\]
denote the Christoffel symbols of the metric \(g\).

The collision operator \(C(F,F)\) in \eqref{genrelatbolzt} is defined by
\begin{equation}\notag 
C(F,G)\eqdef \int_{\mathcal{P}_{x}}d\mu_{\mathcal{P}_{x}}(q)\int_{\mathcal{P}_{x}}d\mu_{\mathcal{P}_{x}}(q')\int_{\mathcal{P}_{x}}d\mu_{\mathcal{P}_{x}}(p')W\big(F(q')G(p')-F(q)G(p)\big),
\end{equation}
in terms of the volume form \(d\mu_{\mathcal{P}_{x}}\) in the fibers of the mass-shell \(\mathcal{P}_{x}\) given by \[d\mu_{\mathcal{P}_{x}}(p)=\frac{\sqrt{-\mathrm{det}\,g}}{-p_0}dp^1dp^2dp^3.\]
We recall that indices are raised and lowered with respect to the metric \(g\). Moreover, the transition rate \(W=W(p,q|p',q')\) is defined by 
\begin{equation}\label{transition_kernel}
W(p,q|p',q')=\frac{1}{2}s\sigma(\varrho, \theta)\delta^{(4)}(p^{\mu}+q^{\mu}-p'^{\mu}-q'^{\mu})(-\mathrm{det}\,g)^{-\frac{1}{2}},
\end{equation}
where \(\sigma(\varrho, \theta)\) is a scattering kernel measuring the interactions between particles. The conservation of energy and momentum due to the elastic collisions among particles is expressed by the fact that \(W\) is supported on \(p\), \(q\), \(p'\), \(q'\), satisfying 
\begin{equation}
p^\mu+q^\mu={\pv}^{\prime\mu}+{\qv}^{\prime\mu},
\label{collisional.inv.widetilde}
\end{equation}
where the notation \(p^{\mu}\), \(q^{\mu}\), and \(p'^{\mu}\), \(q'^{\mu}\), indicates that these are pre and post-collisional momenta, respectively. Here the \emph{relative momentum function} \(\mrel\,\colon \mathcal{P}_x\times \mathcal{P}_x\to [0,+\infty) \) is defined by
\begin{align}\notag 
\mrel(p,q) 
\eqdef
\sqrt{(\pM-\qM) (p_\mu-q_\mu)}
 =\sqrt{-2(\pv^\kappa \qv_\kappa + m^2c^2) },
\end{align}
where \(\varrho(p,q)\) is well-defined since \(p-q\) is either a null or a spacelike vector. Define also the function \(s\,\colon \mathcal{P}_x\times \mathcal{P}_x\to [0,+\infty) \) by 
\begin{align}\notag 
    s(p,q)
\eqdef
-(\pM+\qM) (p_\mu+q_\mu)
=2\left( -\pv^\kappa \qv_\kappa + m^2c^2 \right),
\end{align}
where \(s\geq 0\) since \(p+q\) is  either a null or a timelike vector. Note that \(s=\varrho^2+4m^2c^2.\) We observe that \(\varrho\) and \(s\) are \emph{collision invariants} in the sense that \(\varrho(p,q)=\varrho(p',q')\) for all \(p\), \(q\), \(p'\), \(q'\) satisfying \eqref{collisional.inv.widetilde}, and similarly for \(s\). We define also the \emph{scattering angle function} \(\theta \,\colon \mathcal{P}_x\times \mathcal{P}_x\times \mathcal{P}_x\times \mathcal{P}_x\to \R\) as
\begin{equation}\notag 
\cos(\theta (p,q,p',q'))
\eqdef
(\pM - \qM) ({\pv}_\mu^\prime -{\qv}_\mu^\prime)/\mrel^2.
\end{equation}
By the conservation of energy and momentum \eqref{collisional.inv.widetilde}, it can be shown
that \(\theta\) and \(\cos \theta\) are well-defined.
We now remark that the transition rate \eqref{transition_kernel} is \emph{coordinate invariant}, and so, also is the collision operator.

\subsection{The Einstein--Boltzmann system and Maxwell--J\"uttner FLRW }\label{subsecMJFLRW}

The Boltzmann equation in general relativity \eqref{mass-shellChristoff}--\eqref{genrelatbolzt} is a fundamental kinetic model for the study of gravitational collisional systems. In the theory of general relativity, these systems are modelled by the \emph{Einstein--Boltzmann system}, which consists of equations \eqref{mass-shellChristoff}--\eqref{genrelatbolzt} coupled to the Einstein equations 
\begin{equation}\label{einstein}
    Ric(g)_{\mu\nu}-\frac{1}{2}R(g)g_{\mu\nu}
    =
    \frac{8\pi G}{c^4} T_{\mu\nu}
\end{equation}
where the energy-momentum tensor \(T_{\mu\nu}\) for the Boltzmann equation takes the form 
\begin{equation}\notag
    T_{\mu\nu}[F]=\int_{\mathcal{P}_x}F(t,x,p)p_{\mu}p_{\nu}\frac{\sqrt{-\mathrm{det}\,g}}{-p_0}dp^1dp^2dp^3.
\end{equation}
The Einstein--Boltzmann system in the case of massless particles admits an explicit FLRW solution with spatial topology \(\mathbb{T}^3\), the scale factor \(a(t)=(2t)^{\frac{1}{2}}\), and with \(F(t,x,p)=\exp (-|2tp|)\) being a Maxwell--J\"uttner equilibrium.  More generally, for any $\qq \geq \frac{1}{2}$, the \emph{Einstein--massless Boltzmann--scalar field} system admits an FLRW solution with scale factor $a(t) \sim t^{\qq}$ for large $t$, and $F(t,x,p) = \exp(-\vert a(t)^2 p \vert)$ a Maxwell--J\"uttner equilibrium.

Further, for any \(\qq\in (0,1]\), there are other matter models for which the Einstein equations \eqref{einstein} admit as solutions the FLRW spacetime with spatial topology \(\T^3\), and scale factor \(a(t) = t^{\qq}\). These are the geometric backgrounds we consider in this article. We refer to the appendix of \cite{TaylorVelozo2025} for more information about how these spacetimes arise as solutions to the Einstein equations coupled to suitable matter models.

\subsection{The massless Boltzmann equation on FLRW spacetimes}

Let us now focus on the massless Boltzmann equation \eqref{mass-shellChristoff}--\eqref{genrelatbolzt} on the FLRW spacetimes with \(\mathbb{T}^3\) spatial topology of the form
\begin{equation}
\mathcal{M}=(0,+\infty)_t\times \mathbb{T}^3_x,\label{spacetimemanifold}
\end{equation}
with the Lorentzian metric
\begin{equation}
g_{\qq}=-dt\otimes dt+t^{2\qq}(dx^1\otimes dx^1+dx^2\otimes dx^2+dx^3\otimes dx^3),   \label{metricFLRW} 
\end{equation}
where \(\qq\in (0,1]\), and the speed of light is normalised to unity $c=1$. These geometric backgrounds for \(\qq>0\) model an expanding universe with a big bang singularity at \(t= 0\). For \(\qq\in (0,1)\), these spacetimes are said to undergo decelerated expansion in the sense that the scale factor \(a(t)=t^{\qq}\) satisfies \(\frac{d^2}{dt^2}a(t)<0\).  The case \(\qq=1\) is linearly expanding in the sense that \(\frac{d^2}{dt^2}a(t)=0\). We refer to the spacetime \eqref{spacetimemanifold}--\eqref{metricFLRW} in the non-expanding case \(\qq=0\), as the non-expanding FLRW spacetime. We will also consider this latter spacetime.

In this paper, we investigate the long-time dynamics of the \emph{massless Boltzmann equation on the FLRW spacetime} \((\mathcal{M},g_{\qq})\) which takes the form
\begin{equation}
    p^0\partial_t F+\pv^i \partial_{x^i} F
    - \frac{2 \qq}{t} p^ip^0 \partial_{p^i} F = C(F,F),  \qquad F(t=1,x,p)=F_1(x,p),
\label{rBoltz.eqn0}
\end{equation}
where \(F_1(x,p)\) is a regular initial datum, and \(p^0\eqdef |t^\qq\pv|\) by the mass-shell relation with \(|p|^2=(p^1)^2+(p^2)^2+(p^3)^2\). Here, we have used the explicit form of the non-zero Christoffel symbols \(\Gamma_{\mu\nu}^{\alpha}\) of \(g_{\qq}\) given by
\begin{align*}
        \Gamma^i_{j0}
    =    \Gamma^i_{0j}
    =
    \frac{\qq}{t} \delta_{ij},
    \quad
    \Gamma^0_{ij}
    =
    \frac{\qq}{t} t^{2 \qq} \delta_{ij},
\end{align*}
where \(\delta_{ij}\) is the standard Kronecker delta which is 1 when the indices are equal, and 0 otherwise. Note that for any smooth function $h(\tqq \pv)$, the linear operator in the left hand side of \eqref{rBoltz.eqn0} satisfies
\begin{align} \label{nullRHStp}
       \Big(  \partial_t  +\frac{\pv^i}{|\tq \pv|}\partial_{x^i}
    -\frac{2\qq }{t} p^i \partial_{p^i} \Big) \big(h(\tqq \pv)\big) = 0.
\end{align}

On the FLRW background \eqref{spacetimemanifold}--\eqref{metricFLRW}, the collision operator \(C(F,F)\) is 
\begin{equation}\label{collisionop}
C(F,G) =  t^{9\qq}\int_{\threed} \frac{dq}{\qZzero}
\int_{\threed}\frac{dq^\prime}{\qZprime}
\int_{\threed}\frac{dp^\prime}{\pZprime}
W(p, q \,| p^\prime, q^\prime) [F(p^{\prime})G(q^{\prime})-F(p)G(q)], 
\end{equation}
in terms of the explicit volume form \(d\mu_{\mathcal{P}_{x}}(p)=t^{3\qq}(p^0)^{-1} dp^1dp^2dp^3\) on the fibers of the mass-shell \(\mathcal{P}\) of FLRW. 
Here, the transition rate \(W(p,q|p',q')\) is given by \eqref{transition_kernel}.

\subsubsection{Scattering kernel of the collision operator}

The scattering collision kernel \(\sigma (\varrho,\theta)\) we will consider in the transition rate \(W(p,q|p',q')\) is
\begin{equation}
\sigma (\varrho,\theta)\eqdef \text{constant},
\label{hardSPHERE}
\end{equation}
with value normalised to be 1. This kernel corresponds to the case of short range interactions \cite{MR1402248,MR1958975}. We refer to the kernel \eqref{hardSPHERE} as the \emph{hard-ball cross section} since it is the relativistic analogue of the hard-sphere kernel for the Newtonian Boltzmann equation.\footnote{The Newtonian limit of the Boltzmann equation on \(\mathbb{R}^3_x\) in the massive case with \eqref{hardSPHERE} is the usual hard-sphere Boltzmann equation as shown in \cite{MR2679588}. }

More general commonly studied examples of collision kernels include $\sigma (\mrel,\theta )$ of the form
\begin{equation}\label{kernelgeneral}
    \sigma (\mrel,\theta ) = \mrel^a \sigma_0(\theta), \quad  \quad -4<a\leq  2,
\end{equation}
for a suitable non-negative function \(\sigma_0(\theta)\). We expect our techniques to generalize to these kernels with a suitable cutoff assumption on $\sigma_0(\theta)$. One example that appears frequently in the physics literature is the Israel kernel \cite{MR165921}:
\begin{equation}\label{israelkernel}
\sigma(\mrel,\theta)=\frac{\sigma_0(\theta)}{\mrel(1+\mrel^2)} .    
\end{equation}
We refer to \cite{MR3186493} for more information. Other collision kernels are discussed for example in \cite[Appendix B]{MR2679588}, and \cite{MR635279, MR471665}. See Section \ref{generalkernels} for a further discussion about kernels of the form \eqref{kernelgeneral}.

\subsubsection{Conservation laws and the H-theorem}\label{subsconserved}

Key quantities for the solutions of the massless Boltzmann equation on FLRW spacetime \eqref{rBoltz.eqn0} are their \emph{mass}, \emph{momentum}, and \emph{energy} defined respectively as
\begin{equation}\label{mass_mom_entrop}
t^{6\qq}\int_{\mathbb{T}^3_x}\int_{ \mathbb{R}^3_p} Fdpdx, \qquad t^{6\qq}\int_{\mathbb{T}^3_x}\int_{ \mathbb{R}^3_p} t^{2\qq}p^iFdpdx,\qquad t^{6\qq}\int_{\mathbb{T}^3_x}\int_{ \mathbb{R}^3_p} |t^{2\qq}p|Fdpdx,
\end{equation}
where \(i\in \{1,2,3\}\). Here the weight \(t^{6\qq}\) comes from the volume form of phase space. A basic property of the solutions to the massless Boltzmann equation \eqref{rBoltz.eqn0} is the conservation in time of these quantities,
\begin{align}\notag
\frac{d}{dt}\bigg(t^{6\mathfrak{q}}\int_{\mathbb{T}^3_x}\int_{ \mathbb{R}^3_p}  
\begin{pmatrix}
         1 \\ t^{2\mathfrak{q}}p^i\\ |t^{2\mathfrak{q}}p|
    \end{pmatrix} 
    F(t,x,p) dpdx\bigg)&=0.
\end{align}
See Corollary \ref{corconservationlaws} for the proof of these properties. These global conservation laws arise from the conservation of energy and momentum of collisions among particles.

We also observe that Boltzmann's H-theorem holds in this setting. Let \(H[F]\) be the \emph{entropy} of a distribution function \(F(t,x,p)\) at time \(t\) defined by 
\begin{equation} \label{eq:defofentropy}
    H[F]\eqdef \int_{\mathbb{T}^3_x}\int_{ \mathbb{R}^3_p} - t^{6\qq} F\log F(t,x,p) dpdx.
\end{equation}
where \(-F \log F\) is the standard entropy density. Similarly as in \eqref{mass_mom_entrop}, the term \(t^{6\qq}\) comes from the volume form of phase space. In this context, the H-theorem states that the entropy of the system \(H[F(t)]\) is a non-decreasing function of time since
\(\frac{d}{dt}H[F(t)]\geq 0.\) See Theorem \ref{prophthm} for its proof.

\subsubsection{Maxwell--J\"uttner equilibria for massless Boltzmann on FLRW}

In the massless case when \(m=0\), the Boltzmann equation on FLRW spacetime \((\mathcal{M},g_{\qq})\) has an explicit non-stationary \emph{Maxwell--J\"uttner equilibrium}
 \begin{equation}
J ( \tqq p)
=
\frac{e^{-  \tq \pv^0}}{8\pi}=\frac{e^{- |\tqq\pv|}}{8\pi},
\label{juttner}
\end{equation}
where the normalising constant is chosen so that \(\int_{\R^3} J(p)dp=1\) at \(t=1\). For $\qq>0$ the Maxwell--J\"uttner solution \eqref{juttner} is decaying as $t \to +\infty$ and thus referred to as a ``non-stationary equilibrium''. In the massive case $m>0$, there are no non-trivial global Maxwell--J\"uttner equilibria, see \cite[Section 12.5]{CercignaniKremer2002} and \cite[Theorem 5.1]{Ellis_Maartens_MacCallum_2012}.

More generally, the massless Boltzmann equation on FLRW spacetime \eqref{rBoltz.eqn0} admits an explicit \(5\)-dimensional family of non-stationary Maxwell--J\"uttner equilibria, 
\begin{equation}\label{equilibriageneral}
J_{a,b,c} (t^{{2\qq}}p)
=\frac{1}{8\pi}e^{a+b \cdot  t^{2\qq}p-c|t^{2\qq}p|}  
\end{equation}
where \(a\in \R\), \(b\in \R^3\), and \(c>|b|\) are five constant parameters. Note that \(J_{0,0,1}=J\) is the normalised Maxwell--J\"uttner equilibrium in \eqref{juttner}.

The Maxwell--J\"uttner equilibria \(J_{a,b,c}\) are characterised as
suitable maximisers of the entropy \eqref{eq:defofentropy}, subject to constant mass, momentum, and energy \eqref{mass_mom_entrop}.

For $\qq=0$ it is well known that the Boltzmann equation is dissipative around the stationary solution \eqref{juttner}.\footnote{This fact is well-known in the case of massive particles. It is also true in the massless case as it is shown in Theorem \ref{Torus Existence} in the hard-ball case.} For $\qq>0$, since \eqref{juttner} is decaying in time, this dissipative effect weakens as $\qq$ increases.  This fact is captured in the weakening of the decay rates of the perturbations in Theorem \ref{Torus Existence} as $\qq$ increases.

\subsubsection{The vacuum solution for massless Boltzmann on FLRW}

For all \(\qq\geq 0\), the massless Boltzmann equation on FLRW spacetime \((\mathcal{M},g_{\qq})\) admits the trivial stationary solution \(F\equiv 0\), also known as the \emph{vacuum solution}. This distribution function is referred as vacuum since it models a system with no particles. Theorem \ref{thmvacuum} below concerns the stability of the vacuum for the massless Boltzmann equation when \(\qq>\frac{1}{3}\).

\subsection{The stability problem for the Maxwell--J\"uttner equilibrium}

In this work, we study the problem of convergence to equilibrium for the solutions of the massless Boltzmann equation on FLRW \eqref{rBoltz.eqn0}. We will begin with a standard reformulation of the initial value problem around the Maxwell--J\"uttner equilibrium \eqref{juttner}. In what follows, we will use the notation
\(Q(F,G) = \frac{1}{p^0}C(F,G).\)

\subsubsection{Reformulation of the problem}

Given a solution \(F(t,x,p)\) of \eqref{rBoltz.eqn0}, we define \(f(t,x,p)\) by the relation
\begin{align}\label{standard.perturbation}
    F(t,x,p) \eqdef J(t^{2\qq}p) +\sqrt{J(t^{2\qq}p) } f(t,x,p).
\end{align}
The unknown \(f(t,x,p)\) satisfies the equation
\begin{align}
        \partial_t f+ \frac{\pv^i}{|\tq \pv|}\partial_{x^i} f
    -\frac{2\qq }{t} p^i \partial_{p^i} f
    +
    \linL f
=
\Gamma (f,f),
\label{rBoltz00}
\end{align}
with initial datum 
\begin{equation}\label{idff}
f(t=1,x,\pv)=f_1(x,\pv)\eqdef (F_1(x,\pv)-J(\pv))/\sqrt{J(\pv)}.
\end{equation}
Here, the linearised operator \(\linL\) and the nonlinear operator \(\Gamma\) are given by
\begin{align}
     \linL f
 & \eqdef 
 -J^{-1/2}Q(J ,\sqrt{J} f)- J^{-1/2}Q(\sqrt{J} f,J) ,
 \label{L0}
\end{align}
and
\begin{align}
\Gamma (h_1,h_2)
&\eqdef
 J^{-1/2}Q(\sqrt{J} h_1,\sqrt{J} h_2).
\label{gamma0}
\end{align}

We will suppose in the following that the perturbation \(f(t,x,p)\) satisfies for all \(t\geq 1\) that
\begin{multline}\label{zeroconserv}
           t^{6\qq} \int_{\mathbb{T}^3_x}\int_{ \mathbb{R}^3_p}   J^{\frac{1}{2}}(\tqq \pv)f(t,x,\pv)d\pv dx=    t^{6\qq} \int_{\mathbb{T}^3_x}\int_{ \mathbb{R}^3_p} \tqq p^i  J^{\frac{1}{2}}(\tqq \pv)f(t,x,\pv)d\pv dx
       \\ 
       =   t^{6\qq} \int_{\mathbb{T}^3_x}\int_{ \mathbb{R}^3_p}  |\tqq \pv| J^{\frac{1}{2}}(\tqq \pv)f(t,x,\pv)d\pv dx=0,
\end{multline}
where \(i\in \{1,2,3\}\). Due to the conservation of the mass, momentum, and energy, if these relations are satisfied initially then the equation \eqref{rBoltz00} implies that they will continue to be satisfied for all \(t \geq 1\).

For the stability problem of the more general Maxwell--J\"uttner equilibria \eqref{equilibriageneral}, one can reformulate the initial value problem similarly as in \eqref{rBoltz00}--\eqref{idff}, and assume a condition on the perturbation \(f(t,x,p)\) similar to \eqref{zeroconserv}. The results we will obtain for the problem \eqref{rBoltz00}--\eqref{idff} can be extended to these other equilibria. 

We could also relax the assumption \eqref{zeroconserv} by identifying the Maxwell--J\"uttner equilibrium \(J_{a,b,c}\) one would converge at large-times from the values of the conserved quantities \eqref{mass_mom_entrop} at \(\{t=1\}\). See Appendix \ref{MJparamet}. Then one can suitably modify the assumption \eqref{zeroconserv} by the difference with the corresponding equilibrium \(J_{a,b,c}\).

\subsection{The collision operator on FLRW}\label{sec:FLRWcoll}

By the  conservation of energy and momentum \eqref{collisional.inv.widetilde}, we can parameterize ${\pv}^{\prime\mu}$ and ${\qv}^{\prime\mu}$ by ${\pv}, {\qv} \in \mathbb{R}^3$ and an angular variable $\omega\in \mathbb{S}^2$. Due to the diagonal structure of the FLRW metric, we can follow the arguments in the special relativistic case, for example in \cite{MR2765751}, to reduce the collision integrals in \eqref{collisionop}. We present in this section the center of momentum and Glassey--Strauss reductions of the collision operator. 

\subsubsection{Collisional invariants for massless particles on FLRW}

Let us write the functions \(s(p,q)\) and \(\varrho^2(p,q)\) using the explicit form of the FLRW metric \(g_{\qq}\) in the case of massless particles \(m=0\),
\begin{align}
\mrel(p,q) 
 =\tq \sqrt{2\big(\pZ \qZ - p\cdot q\big)},\qquad s(p,q)
= \tqq 2\big(\pZ\qZ - p\cdot q\big).\label{sDEFINITIONFLRW}
\end{align} 
Note that \(s=\varrho^2\geq 0.\) We further show some elementary bounds for \(s=\varrho^2\).

\begin{lemma} 
For all \((t=x^0,x)\in \mathcal{M}\) and all \((p,q)\in \mathcal{P}_x\times \mathcal{P}_x\), there holds
\begin{equation}\label{basicrho}
      \tq \frac{|p\times q|}{\pZ^{\frac{1}{2}}\qZ^{\frac{1}{2}} }\leq \mrel(p,q) \leq \tq |\pv-\qv|, \qquad \tqq \frac{|p\times q|^2}{\pZ\qZ }\leq s(p,q)  \leq   \tqq |\pv-\qv|^2.
\end{equation}
\end{lemma}

\begin{proof}
    The estimates from above follow from  
    \begin{align*}
  |\pv| | \qv| \leq   \frac{1}{2}\left(|\pv|^2+ |\qv|^2\right)
    = \frac{1}{2}|\pv-\qv|^2 + \pv \cdot \qv.
\end{align*}
On the other hand, the estimates from below follow from 
\begin{equation}\label{mrel.upper.est}
    s = 2\tqq \frac{\pZ^2\qZ^2 - (p\cdot q)^2}{\pZ\qZ + p\cdot q}
    = 2\tqq \frac{|p\times q|^2}{\pZ\qZ + p\cdot q},
\end{equation}
and the Cauchy--Schwarz inequality.
\end{proof}

The estimates in \eqref{basicrho} are similar to those used for the massive case \(m=1\) in the non-expanding case \(\qq=0\) in \cite[Lemma 3.1]{MR1211782}.

\subsubsection{Center of momentum reduction of the collision operator}

Let us reduce the collision operator \(Q(F,G)\) in the center-of-momentum representation. For this, we will use the parametrisation of the post-collisional momenta in the \emph{center-of-momentum expression} written, for \(\omega\in \mathbb{S}^2\), as
\begin{align}
p'
&=\frac{p+q}{2}+\frac{\sqrt{2(|\pv| |\qv| - \pv \cdot \qv)}}{{2}}\bigg(\omega+(\gamma-1)(p+q)\frac{(p+q)\cdot \omega}{|p+q|^2}\bigg), \label{pprimedefcm}\\
q'
&=\frac{p+q}{2}-\frac{\sqrt{2(|\pv| |\qv| - \pv \cdot \qv)}}{{2}}\bigg(\omega+(\gamma-1)(p+q)\frac{(p+q)\cdot \omega}{|p+q|^2}\bigg),\label{qprimedefcm}
\end{align}
where
\begin{align}\notag 
    \gamma =\gamma (p,q)=\frac{|p|+|q|}{\sqrt{2 (\pZ \qZ - p\cdot q )}}.
\end{align}
 With this parametrisation of the post-collisional momenta
\(p'\) and \(q'\), we have
\begin{align*}
|p'|&=\frac{|p|+|q|}{2}+\frac{(p+q)\cdot \omega}{2},\qquad
|q'|=\frac{|p|+|q|}{2}-\frac{(p+q)\cdot \omega}{2}.
\end{align*}
Then, the conservation of energy and momentum \eqref{collisional.inv.widetilde} are satisfied in the form
 \begin{equation}
 |\tq p|+|\tq q|  =|\tq p'|+|\tq q'|,\qquad  
  p+ q  =p^{\prime}+q^{\prime}.
 \label{collisionalCONSERVATION}
 \end{equation}
Furthermore the scattering angle $\theta$ satisfies
\begin{align}\notag 
\cos\theta
= &~
  \omega \cdot k(p,q),
\end{align}
where 
\begin{align*}
     k^i(p,q)
= &~ 
      \frac{1 }{|\pv|+|\qv|+{\sqrt{2  \left(|\pv| |\qv| - \pv \cdot \qv \right)}}}\left(\frac{2\left(|\qv|\pv^{i} -  |\pv|\qv^{i}\right)}{{\sqrt{2  \left(|\pv| |\qv| - \pv \cdot \qv \right)}}}     
     + \left(\pv^{i} - \qv^{i} \right) \right).
\end{align*}
It can be calculated that $|k| =1$.

Then, following the proof of \cite[Theorem 2 and Corollary 3]{MR2765751}, the collision operator \(Q(F,G)\) can be reduced in its \emph{center-of-momentum representation} to 
\begin{equation}
Q(F,G)
=
\tTq\int_{\mathbb{R}^3}  {dq}  ~
\int_{\mathbb{S}^{2}} d\omega 
~ v_{\Emptyset}  \sigma (\mrel,\theta ) ~ \big[F(p^{\prime })G(q^{\prime})-F(p)G(q)\big].
\label{collisionCMc}
\end{equation}
where \(v_{\Emptyset}=v_{\Emptyset}(p,q)\) is the Møller velocity given by
\begin{equation}\label{moller}
v_{\Emptyset}(p,q)\eqdef \frac{\mrel(p,q)s(p,q)^{\frac{1}{2}}}{4p^0q^0}
=\frac{1}{2} \frac{|\pv| |\qv| - \pv \cdot \qv}{|\pv| |\qv|}
=
\frac{1}{4}\left| \frac{\pv}{\pZ} -\frac{\qv}{\qZ}\right|^2. 
\end{equation}
We used here that for massless particles in FLRW, we have \( \pv^{0}=\tq |\pv|\), \(\qv^{0}=\tq |\qv|\), and \(\mrel^2 
 =s=2 \tqq \left(|\pv| |\qv| - \pv \cdot \qv \right).\)

\subsubsection{Glassey--Strauss reduction of the collision operator}

On the other hand, we can reduce the collision operator \(Q(F,G)\) in the Glassey--Strauss representation. For this, we will use the parametrisation of the post-collisional momenta in the Glassey--Strauss expression, for \(\omega\in \mathbb{S}^2\), as
\begin{equation}
\begin{split}
\pv^{\prime } &= 
\pv^{} 
+
\frac{2\omega \cdot (|\pv|  \qv- |\qv|  \pv)(|\pv| + |\qv|) }{(|\pv| + |\qv|)^{2}-(\omega \cdot [ \pv + \qv])^{2}}\omega^{},
 \\
\qv^{\prime } &= 
\qv^{} 
-
\frac{2\omega \cdot (|\pv|  \qv- |\qv|  \pv)(|\pv| + |\qv|) }{(|\pv| + |\qv|)^{2}-(\omega \cdot [ \pv + \qv])^{2}} \omega^{}.
\end{split}
\label{postCOLLvelGS.FLRW}
\end{equation}

Then, following the proof of \cite[Theorem 4]{MR2765751}, the collision operator \(Q(F,G)\) can be reduced in the \emph{Glassey--Strauss representation} as 
\begin{equation}
Q(F,G)
=
\tTq\int_{\mathbb{R}^3}  {dq}  ~
\int_{\mathbb{S}^{2}} d\omega 
~ v_{\Emptyset}  \sigma (\mrel,\theta ) \mathcal{B}({\pv},{\qv},\omega ) \big[F(p^{\prime })G(q^{\prime})-F(p)G(q)\big],
\label{collision.GS.flrw}
\end{equation}
where 
\begin{equation}\notag
\mathcal{B}({\pv},{\qv}, {\omega} )
=
\frac{4 (|{\pv}| + |{\qv}|)^2\left|{\omega} \cdot (|{\pv}|  {\qv}- |{\qv}|  {\pv})\right| }{\left[(|{\pv}| + |{\qv}|)^{2}-({\omega} \cdot ( {\pv} + {\qv}))^{2}\right]^{2}}.
\end{equation}
Recall the definition of the Møller velocity $v_{\Emptyset}$ in \eqref{moller}.

By using the parametrisation \eqref{postCOLLvelGS.FLRW} of the post-collisional momenta, following directly the argument in \cite{MR1105532}, one can calculate the following Jacobian change of variables.

\begin{lemma}\label{Lem change of variables}
Let $(p',q')$ be the post-collisional
momenta parametrised as in \eqref{postCOLLvelGS.FLRW}. For the collision map
$(p,q)\mapsto(p',q')$, the Jacobian change of variables is
\begin{align}\label{Jacobian.COV}
    \frac{\partial(p',q')}{\partial(p,q)}=-\frac{p^{0\prime}q^{0\prime}}{p^0q^0}
    =-\frac{|\pv'| |\qv'|}{|\pv| |\qv|}.
\end{align}
\end{lemma}

We use the parametrisation \eqref{postCOLLvelGS.FLRW} to perform the standard pre-post collision change of variables $(p,q)\mapsto(p',q')$ using the Jacobian \eqref{Jacobian.COV}.

\subsubsection{Decomposition of the linearised operator}

In the following, we will study the solutions \(F(t,x,p)\) of the massless Boltzmann equation on the FLRW spacetime 
\begin{equation}\label{massless.rBoltz}
    \partial_t F+ \frac{\pv^i}{|\tq \pv|}\partial_{x^i} F
    -\frac{2\qq }{t} p^i \partial_{p^i} F = Q(F,F), \qquad 
F(t=1, x, p) = F_1(x,p),
\end{equation}
in the case of hard-ball cross section \(\sigma=1\). Typically the center of momentum reduction of the collision operator will be used so that
\begin{equation}\label{collision.CM.massless}
Q(F,G)
=
\frac{t^{3\qq}}{2}
\int_{\mathbb{R}^3}  dq ~
\int_{\mathbb{S}^{2}} d\omega 
~  \left(1 - \frac{p\cdot q}{\pZ\qZ}\right) ~ \big(F(p^{\prime })G(q^{\prime})-F(p)G(q)\big),
\end{equation}
where \(p'=p'(p,q,\omega)\) and \(q'=q'(p,q,\omega)\) are defined by \eqref{pprimedefcm}--\eqref{qprimedefcm}.

We recall now the reformulation for the problem \eqref{massless.rBoltz}--\eqref{collision.CM.massless} around the Maxwell--J\"uttner equilibrium as in \eqref{rBoltz00}. We now decompose the linearised operator $\linL f$ defined in \eqref{L0} as
\begin{align}
     \linL h
 &  =
 \nu(t,p) h-\compK (h),\label{decomlinop}
\end{align}
in terms of the multiplication operator
\begin{equation}
    \nu(t,p) \eqdef
t^{3\qq}\int_{\mathbb{R}^3} ~  dq 
\int_{\mathbb{S}^{2}} ~ d\omega
~ v_{\Emptyset} ~  \sigma(\mrel,\theta )~ J(\tqq\qv),
\label{nuDEF}
\end{equation}
and the momentum integral operator \(\mathcal{K}\) defined by 
\begin{align}
\compK h 
&\eqdef
t^{3\qq}\int_{\mathbb{R}^3} ~  dq 
\int_{\mathbb{S}^{2}} ~ d\omega
~ v_{\Emptyset} ~ \sigma(\mrel,\theta )~ 
\sqrt{J(\tqq \qv)}\left\{ \sqrt{J(\tqq q^{\prime})} ~ h(p^{\prime })+\sqrt{J(\tqq p^{\prime})} ~ h( q^{\prime})\right\}
\notag
\\
&\qquad \qquad -t^{3\qq}\int_{\mathbb{R}^3} ~  dq 
\int_{\mathbb{S}^{2}} ~ d\omega
~ v_{\Emptyset} ~\sigma(\mrel,\theta )  
~ \sqrt{J(\tqq\qv) J(\tqq\pv)} ~ h(q)
\label{compactK}
\\
&=
\compK_2 h - \compK_1 h,
\notag
\end{align}
where the last line defines the integral operators \(\compK_2 \) and \(\compK_1 \). We now compute the function \(\nu(t,p)\), also known as collision frequency, for the hard-ball kernel.

\begin{proposition}\label{LEM:nuASYMP}
    With the hard-ball cross section \eqref{hardSPHERE}, the collision frequency is \(\nu(t,p)=\nu_0 t^{-3\qq}\), where \(\nu_0>0\) is a fixed constant.  
\end{proposition}

\begin{proof}
By the explicit form \eqref{moller} of the Møller velocity, we compute in the hard-ball cross-section case that
\begin{align}\label{nuDEFnext}
    \nu(t,p) =&~
\frac{t^{3\qq}}{16\pi}\int_{\mathbb{R}^3} ~  dq 
\int_{\mathbb{S}^{2}} ~ d\omega
~  \left(1 - \frac{p\cdot q}{\pZ\qZ}\right) ~
    e^{-{| \tqq q|}}
    \\ \notag
    =&~
\frac{t^{-3\qq}}{16\pi}\int_{\mathbb{R}^3} ~  dq 
\int_{\mathbb{S}^{2}} ~ d\omega
~  \left(1 - \frac{p\cdot q}{\pZ\qZ}\right) ~
    e^{-{|  q|}} = \nu_0t^{-3\qq},
\end{align}
where $\nu_0$ is constant since the value of the integral does not depend upon $\frac{p}{\pZ}$.
\end{proof}

\subsection{Notations and function spaces}\label{sec.notation}

Let us set the notations and function spaces we will use in the paper.

\subsubsection{Notations}

Through the paper, we will write \(A \lesssim B\) if there is a uniform constant \(C > 0\) such
that \(A \leq C B\). If \(A \lesssim B\) and \(B \lesssim A\), then we write \(A \approx  B\). We denote by \(B_r = B(0,r)\) the standard ball centered at zero and radius \(r > 0\). We let $\real{z}$ and $\imag{z}$ denote the real and complex parts of a complex number $z$, respectively. In the Fourier variables $k\in \Z^3$, we also define the Japanese bracket
\begin{align}\notag 
   \langle k \rangle \eqdef \left(1+|k|^2 \right)^{\frac12}. 
\end{align}
The momentum of a particle is denoted by a 4-vector \(p^\mu\). We raise and lower indices of the momenta with the FLRW metric, i.e. \(p_{\mu} = g_{\mu\nu} p^{\nu}\). The components of the inverse of the metric \(g\) are denoted with indices upstairs as \(g^{\mu\nu}\). Throughout the paper, Greek indices range over 0, 1, 2, 3, and lower case Latin indices range over 1, 2, 3.

We parameterize the momentum of a massless particle in FLRW spacetime \((\mathcal{M},g_{\qq})\) by \(p=(p^1,p^2,p^3)\in \R^3\) with particle energy $\pZzero = \tq |\pv|$ where \(|p|^2= p \cdot p\). We write here and onwards the Euclidean dot product as \(p \cdot q=\sum_{i=1}^3p^iq^i\). Abusively, we write \(p^{\mu}\) to denote both the 4-momentum \(p^{\mu}=(p^0, p^i)\), and the component \(p^{\mu}\).

\subsubsection{Weights and Lebesgue space norms}

For $m\ge 0$, we define the weight function
\begin{align} \label{weight.fcn.def}
     w_{m}^2(\tqq \pv) \eqdef 1+|\tqq \pv|^m.
\end{align}

For \(r\), \(s\geq 1\), we define the norms 
\begin{align*}
    \norm{h}_{L^r_{\pv}}
=\norm{h}_{L^r(\threed_{\pv})}
\eqdef &~
\left(
\int_{\R^3_p}
\left| h(\pv)\right|^r
d\pv
\right)^{1/r},
\\
    \norm{h}_{L^s_{T}}
=\norm{h}_{L^s([1,T])}
\eqdef &~
\left(
\int_{1}^T
\left| h(t)\right|^s
dt
\right)^{1/s},
\\
\|h\|_{L^s_k}
\eqdef &~
\left( \int_{\Z^3_{k}}
\left| \hat{h}(k)\right|^s
d\Sigma(k)
\right)^{1/s},
\end{align*}
where \(\hat{h}(k)\) is the Fourier transform of \(h(x)\), that is, \(\hat{h}(k)=\int_{\mathbb{T}^3} e^{-ik\cdot x}h(x)dx\) with \(k\in \Z^3\). We used above the notation $d\Sigma(k)$ for the discrete measure in $\Z^3$ such that
$$\int_{\Z^3}h(k)\,d\Sigma(k)=\sum_{k\in \Z^3}h(k).$$ In the following, we will use \(\langle \cdot ,\cdot \rangle\) to denote the standard \(L^2(\mathbb{R}^3_p)\) inner product.

\subsubsection{Weighted and mixed Lebesgue spaces}

For $1\le u$, $r< +\infty$, we consider the $L^1_kL^u_TL^r_p$ mixed Lebesgue space norms 
\begin{equation*}
\|f\|_{L^1_kL^u_TL^r_{\pv}}
\eqdef
 \int_{\Z^3_{k}}
\bigg(
\int_{1}^{T}
\bigg(
\int_{\R^3}
\left| \hat{f}(t,k,\pv)\right|^r
d\pv
\bigg)^{\frac{u}{r}} dt
\bigg)^{\frac{1}{u}}
d\Sigma(k),
\end{equation*}
the weighted mixed norms  
\begin{equation*}
\|w_{m}f\|_{L^1_kL^u_TL^r_{\pv}}
\eqdef
 \int_{\Z^3_{k}}
\bigg(
\int_{1}^{T}
\bigg(
\int_{\R^3} w_{m}( \tqq \pv)^r
\left| \hat{f}(t,k,\pv)\right|^r
d\pv
\bigg)^{\frac{u}{r}} dt
\bigg)^{\frac{1}{u}}
d\Sigma(k),
\end{equation*}
and, for \(l\geq 0\), the weighted norms
\begin{equation*}
\|w_{m}f\|_{L^1_{k,l}L^u_TL^r_{\pv}}
\eqdef
 \int_{\Z^3_{k}} \langle k \rangle^l
\bigg(
\int_{1}^{T}
\bigg(
\int_{\R^3} w_{m}( \tqq \pv)^r
\left| \hat{f}(t,k,\pv)\right|^r
d\pv
\bigg)^{\frac{u}{r}} dt
\bigg)^{\frac{1}{u}}
d\Sigma(k).
\end{equation*}
We use the standard modification when $u$ or $r$ are equal to $+\infty$. Above, \(\hat{f}(t,k,p)\) denotes the Fourier transform of $f(t,x,\pv)$ with respect to $x\in \T^3$.

\subsubsection{Chemin--Lerner function space}

In the main results we consider the low regularity function space \(L^1_kL^{\infty}_TL^2_p\).  For \(1\leq T<+\infty\), define the space of functions \(L^1_kL^{\infty}_TL^2_p\) with the norm 
\begin{equation}\notag 
\|  f\|_{L^1_kL^{\infty}_TL^2_p} \eqdef\int_{\Z^3}\esssup_{1\leq t \leq T}  \|\hat{f}(t,k,\cdot)\|_{L^2_{\pv}} \,d\Sigma (k) .
\end{equation}
For \(\qq\in [0,1]\), our analysis will rely on  estimates of the norm \(\|\tTq f\|_{L^1_kL^{\infty}_TL^2_p} \) for solutions of \eqref{rBoltz00}--\eqref{idff}.  The function space \(L^1_kL^{\infty}_TL^2_p\) was first considered for the study of collisional kinetic equations in \cite{1904.12086}. In the rest of the paper when we write $\sup_{1\leq t \leq T}$ we are referring to  $\esssup_{1\leq t \leq T}$.

\subsubsection{Initial data norms}

The initial data norms we will use are
\begin{equation*}
\|f\|_{L^1_{k}L^r_{\pv}}(1)= \|f_1\|_{L^1_{k}L^r_{\pv}}
\eqdef
 \int_{\Z^3_{k}} 
\bigg(
\int_{\R^3} 
\left| \hat{f}_1(k,\pv)\right|^r
d\pv
\bigg)^{\frac{1}{r}} 
d\Sigma(k),
\end{equation*}
and
\begin{equation*}
\|w_{m}f\|_{L^1_{k,l}L^r_{\pv}}(1)
\eqdef
 \int_{\Z^3_{k}} \langle k \rangle^l
\bigg(
\int_{\R^3} w_{m}( \pv)^r
\left| \hat{f}_1(k,\pv)\right|^r
d\pv
\bigg)^{\frac{1}{r}} 
d\Sigma(k).
\end{equation*}
When $l=0$ the weighted initial data norm satisfies $\|w_{m}f\|_{L^1_{k,0}L^r_{\pv}}(1) =\|w_{m}f\|_{L^1_{k}L^r_{\pv}}(1)$.

\subsection{The main results}

We now present the main results of the article. 

\subsubsection{ Global stability of the Maxwell--J\"uttner equilibrium}

We now state our main result concerning the future global stability of the Maxwell--J\"uttner equilibrium \eqref{juttner}.

\begin{theorem}[Future global stability of the Maxwell--J\"uttner equilibrium]\label{Torus Existence}
Let \(\qq\in [0,1]\). The Maxwell--J\"uttner equilibrium is nonlinearly stable as a solution to the massless Boltzmann equation on FLRW spacetime \eqref{rBoltz.eqn0} in the case of hard-ball interaction \eqref{hardSPHERE}.

More precisely, let \(l\geq 0\) and let \(f_1(x,p)\) be an initial datum satisfying \eqref{zeroconserv}. There is \(\epsilon_0>0\) such that if \(F_1(x,p)=J(p)+J^{\frac{1}{2}}(p)f_1\geq 0\) and \(\|\langle k \rangle^lf_1\|_{L^{1}_{k}L^2_p}\leq \epsilon_0\), then there is a unique global mild solution \(f(t,x,p)\) to \eqref{rBoltz00}--\eqref{idff} for the massless Boltzmann equation with hard-ball interaction \eqref{hardSPHERE} on FLRW spacetime \((\mathcal{M},g_{\qq})\) such that \(F(t,x,p)=J(t^{2\qq}p)+J^{\frac{1}{2}}(t^{2\qq}p)f(t,x,p)\geq 0\) and 
\begin{align}
\forall \,T\geq 1,\qquad    {\Vert{\tTq \langle k \rangle^lf}\Vert_{L^1_kL^\infty_TL^2_{\pv}}}
+
{\norm{t^{3\qq/2}\langle k \rangle^l f}_{L^1_kL^2_TL^2_{\pv}}}
\lesssim 
\norm{ \langle k \rangle^l{f}_1}_{L^1_k L^2_{\pv}}.\label{globalestim}
\end{align}

Let also \(m\geq 0\), and $0< \tc < 1$. There is \(\epsilon_1>0\) such that if \(\|w_{m}f_1\|_{L^{1}_{k}L^2_p}\leq \epsilon_1\), then the unique global mild solution \(f(t,x,p)\) satisfies the uniform decay estimate 
\begin{align}
\forall \,t\geq 1,\qquad    {\Vert{ |t^{2\qq}p|^{m/2}f(t)}\Vert_{L^1_kL^2_{\pv}}}
\lesssim t^{-3\qq}(\tw(t))^{-1}\norm{ w_{m}{f}_1}_{L^1_k L^2_{\pv}},\label{decayestimate} 
\end{align}
where 
\begin{equation}\label{weight.decay}
    \tw(t)\eqdef
    \begin{cases}
    \exp(\lambda t),\quad &\text{if} \,\,\qq=0,\\
        \exp (\tc\frac{\nu_0}{1-3\qq}t^{1-3\qq}),\quad &\text{if} \,\,0<\qq<\frac{1}{3},\\
        t^{\tc\nu_0},\qquad &\text{if} \,\,\qq=\frac{1}{3},\\
        1,\qquad &\text{if} \,\,\frac{1}{3}<\qq\leq 1,\\
    \end{cases}
\end{equation}
with \(\lambda>0\) a small constant, and \(\nu_0>0\) the uniform constant in Proposition \ref{LEM:nuASYMP}.
\end{theorem}

\begin{remark}[Functional space]\label{rmfcnspc}
Alternatively, we could have obtained global stability in other function spaces such as
$L^\infty_T H^N_x L^2_p \cap L^2_T H^N_x L^2_p$ for $N> \frac{3}{2}$ without using the Fourier transform, or other \(L^{\infty}_{}\) based functional spaces. In the former case we expect the following uniform estimate
\begin{align}\notag
  {\Vert{\tTq f}\Vert_{L^\infty_T H^N_xL^2_{\pv}}}
+
{\norm{t^{3\qq/2} f}_{L^2_TH^N_xL^2_{\pv}}}
\lesssim 
\norm{ {f}_1}_{H^N_x L^2_{\pv}},
\qquad  
\forall \,T\geq 1,
\end{align}
where the implicit constant is independent of $T \geq 1$.  
\end{remark}

\begin{remark}[Asymptotic vs orbital stability] By the explicit form of the J\"uttner equilibrium, for any $m\ge 0$, one can show that \({\Vert{ |t^{2\qq}p|^mJ(t^{2\qq}p)}\Vert_{L^1_kL^2_{\pv}}} \sim \tmTq\) for large times. In contrast, Theorem \ref{Torus Existence} establishes that for the norm of the perturbation \(\sqrt{J(t^{2\qq}\pv)}f(t, x, \pv)\) we have  \({\Vert{ |t^{2\qq}p|^m\sqrt{J(t^{2\qq}p)}f}\Vert_{L^1_kL^2_{\pv}}} \lesssim \tmTq (\tw(t))^{-1}\) for large times.  From \eqref{weight.decay}, $\tmTq (\tw(t))^{-1} = \tmTq$  for \(1/3<\qq\leq 1\), and $(\tw(t))^{-1}(t)$ decays faster in time for \(0\leq \qq\leq 1/3\). In this sense, Theorem \ref{Torus Existence} can be viewed as  orbital stability for \(1/3<\qq\leq 1\), and asymptotic stability for \(0\leq \qq\leq 1/3\).
\end{remark}

\begin{remark}[Dissipation threshold at \(\qq=1/3\)]\label{rmkdissthres}
The \(t^{3\qq}\) decay rate in Theorem \ref{Torus Existence} should be viewed as a result of cosmological expansion. The additional decay factor in \eqref{weight.decay} for \(\qq\leq 1/3\) should be viewed as a result of dissipation due to particle collisions. Thus \(\qq=1/3\) is a threshold below which the Boltzmann equation is dissipative, and above which is not.

We expect different values of \(\qq\) arise as thresholds for different choices of collision kernels. See Section \ref{generalkernels} below for more information.
\end{remark}

\begin{remark}[Cosmological expansion weakens familiar kinetic effects] \label{rmk:Vlasovmixing}
    Note that the decay due to dissipation in Theorem \ref{Torus Existence} is stronger for small \(\qq\), and thus cosmological expansion weakens this dissipative effect. A related weakening due to cosmological expansion of another familiar kinetic effect --- namely phase mixing for the Vlasov equation --- has also been shown to hold on FLRW spacetimes \cite{TaylorVelozo2025}.
\end{remark}

\subsubsection{Global stability of the vacuum solution}

We will now state our main result concerning the future global stability of the vacuum solution for the massless Boltzmann equation on FLRW. Given a solution \(F(t,x,p)\) of the massless Boltzmann equation \eqref{rBoltz.eqn0}, we define \(f(t,x,p)\) by the relation   
\begin{align*}
    F(t,x,\pv) = J(\tqq \pv) f(t,x,\pv), \qquad 
    F(t=1,x,\pv) =F_1(x,\pv) = J( \pv) f_1(x,\pv).
\end{align*}
The unknown \(f\) satisfies the equation
\begin{equation}\label{perturb.massless.rBoltz0}
   \partial_t f+ \frac{\pv^i}{|\tq \pv|} \partial_{x^i} f
    -\frac{2\qq }{t} p^i\partial_{p^i} f = \Gamma(f,f),
\end{equation}
with initial datum 
\begin{equation}\label{data:vacuumprob}
    f(t=1, x, p) = f_1(x,p)\eqdef F_1(x,\pv)/J( \pv).
\end{equation}
Above in \eqref{perturb.massless.rBoltz0}, the nonlinear term $\Gamma(f,f)$ is defined as in \eqref{gamma0}. We now state our main result in the near vacuum case.  

\begin{theorem}[Future global stability of the vacuum solution]\label{thmvacuum}
    Let \(\qq\in (\frac{1}{3},1]\). The vacuum is nonlinearly stable as a solution to the massless Boltzmann equation on FLRW spacetime \eqref{rBoltz.eqn0} in the case of hard-ball interaction \eqref{hardSPHERE}.
    
    More precisely, let \(f_1(x,p)\) be an initial datum. There is \(\epsilon_0>0\) such that if \(F_1(x,p)=J^{\frac{1}{2}}(p)f_1\geq 0\) and \(\|f_1\|_{L^1_kL^2_{\pv}}\leq \epsilon_0\), then there is a unique global mild solution \(f(t,x,p)\) to \eqref{perturb.massless.rBoltz0}--\eqref{data:vacuumprob} for the massless Boltzmann equation with hard-ball interaction \eqref{hardSPHERE} on FLRW spacetime \((\mathcal{M},g_{\qq})\) such that \(F(t,x,p)=J^{\frac{1}{2}}(t^{2\qq}p)f(t,x,p)\geq 0\) and 
\begin{equation}\label{apriorivacuumestimate0}
\forall T\geq 1,\qquad {\norm{t^{3\qq} f}_{L^1_kL^\infty_TL^2_{\pv}}}
\lesssim 
\norm{f_1}_{L^1_kL^2_{\pv}}.    
\end{equation}
\end{theorem}

\begin{remark}[The restriction \(\qq>\frac{1}{3}\)] The restriction \(\qq>\frac{1}{3}\) in Theorem \ref{thmvacuum} arises from the expansion rate used to estimate the nonlinear term in terms of the energy norm \(\|t^{3\qq} f\|_{L^1_kL^\infty_TL^2_{\pv}}\), and is thus a priori unrelated to the appearance of \(\qq=\frac{1}{3}\) as a threshold in Theorem \ref{Torus Existence}. 
\end{remark}

\begin{remark}[No previous near vacuum works on \(\mathbb{T}^3\)]
    In the Newtonian Boltzmann theory, the standard approach near vacuum in the style of \cite{MR0760333} uses the dispersion of the transport operator $\partial_t + p^i \partial_{x^i}$ in the whole space $\R^3_x$ to estimate the non-linear term. As far as we know, there are no global existence results near the vacuum solution, outside of homogeneity, on the torus $\T^3_x$ for the Newtonian Boltzmann equation due to the lack of dispersion. (Note however recent developments on Landau damping for Vlasov--Poisson on \(\mathbb{T}^3\) \cite{MouhotVillani2011, CMM16, GNR21}.)
\end{remark}

\subsection{Ideas and strategy of the proof}

In this section, we overview the main steps of the proof of Theorem \ref{Torus Existence}.  In Section \ref{subseclinearmodel}, we consider a linear model problem that motivates the time-decay rates we will show for the Boltzmann equation on FLRW depending on the value of \(\qq\). Later, we will sketch the proof of our main results in Subsection \ref{subsecstratproof}.  

\subsubsection{A linear model problem}\label{subseclinearmodel}

In this section, we will consider the model problem
\begin{equation}
    \partial_t f+ \frac{\pv^i}{|\tq \pv|} \partial_{x^i} f
    -\frac{2\qq }{t} p^i\partial_{p^i} f
    +
    \nu_0 t^{-3\qq} f
=
0, \qquad f(t=1,x,p)=f_1(x,p).
\label{rBoltz.model.second}
\end{equation}
This equation arises from neglecting nonlinear terms and the operator \(\mathcal{K}\) in the massless Boltzmann equation (recall Proposition \ref{LEM:nuASYMP} for the collision frequency in the hard ball case).

This is the simplest linear equation displaying the behaviour shown in Theorem \ref{Torus Existence}, and in particular does not admit the difficulties pertaining to the kernel of \(\mathcal{L}\) as in the full linearised problem. We will now study the decay properties of \(\|f(t)\|_{L^2_{x,p}}\) for the solutions of \eqref{rBoltz.model.second}.

Below it will be useful to note that for a function $h(t)$, and $\ell\in \mathbb{R}$, we have 
\begin{align}\label{timeDeriv.pq}
   \left( \frac{d}{d t}+ \frac{\ell\qq }{t}  \right)h(t)=t^{-\ell\qq} \frac{d}{d t}\left( t^{\ell\qq} h(t) \right).
\end{align}

\begin{proposition}[Decay estimates for the model problem]\label{proplinearmodel}
    Let \(\qq\geq 0\) and some \(f_1\in L^2(\mathbb{T}^3\times \mathbb{R}^3)\). Let \(f\) be the unique solution of the model problem \eqref{rBoltz.model.second}. Then, the \(L_{x,p}^2\)-norm of \(f\) satisfies
    \begin{align}\notag
\forall \,t\geq 1,\qquad    {\Vert{ f(t)}\Vert_{L^2_{x,p}}}
\leq t^{-3\qq}
(\tw^{\circ}(t))^{-1}\norm{ {f}_1}_{ L^2_{x,p}},
\end{align}
where 
\begin{equation}\notag
    \tw^{\circ}(t)\eqdef
    \begin{cases}
    \exp( t),\quad &\text{if} \,\,\qq=0,\\
        \exp (\frac{\nu_0}{1-3\qq}t^{1-3\qq}),\quad &\text{if} \,\,0<\qq<\frac{1}{3},\\
        t^{\nu_0},\qquad &\text{if} \,\,\qq=\frac{1}{3},\\
        1,\qquad &\text{if} \,\,\frac{1}{3}<\qq .
    \end{cases}
\end{equation}
\end{proposition}

\begin{proof} Multiplying \eqref{rBoltz.model.second} by $f$, and integrating over phase space, we get
\begin{gather}\notag
   \frac{1}{2}\partial_t \| f(t)\|_{L^2_{x,p}}^2
    +\frac{3\qq }{t} \| f(t) \|_{L^2_{x,p}}^2
    +
    \nu_0t^{-3\qq} \| f(t) \|_{L^2_{x,p}}^2
=
0.
\end{gather}
We therefore obtain using \eqref{timeDeriv.pq} that
\begin{gather}\label{L2normdecay}
  \frac{1}{2} t^{-6\qq} \frac{d}{d t}\left( t^{6\qq} \| f(t)\|_{L^2_{x,p}}^2 \right)
    +
    \nu_0t^{-3\qq} \| f(t) \|_{L^2_{x,p}}^2
=
0.
\end{gather}
The decay of the \(L^2\)-norm when \(\qq=0\) follows by integrating \eqref{L2normdecay}. Suppose now that \(\qq \neq 0,\) \(\frac{1}{3}\). We observe that $t^{6\qq}\|f(t)\|_{L^2_{x,p}}^2$ satisfies the ODE,
\begin{align}\label{ODE.faster.decay}
      \frac{d}{d t}\left( t^{6\qq} \| f(t)\|_{L^2_{x,p}}^2 \right)
    +
    \frac{2\nu_0}{t^{3\qq}} \left( t^{6\qq} \| f(t)\|_{L^2_{x,p}}^2 \right)
=
0,
\end{align}
which is explicitly solved by $\|f(t)\|_{L^2_{x,p}}^2=t^{-6\qq}\exp(\frac{2\nu_0}{1-3\qq}t^{1-3\qq})\|f_1\|_{L^2_{x,p}}^2.$ Finally, the decay estimate when $\qq = \frac{1}{3}$ is obtained by integrating again \eqref{ODE.faster.decay} from which we conclude that \(\|f(t)\|_{L^2_{x,p}}^2=t^{-2-2\nu_0}\|f_1\|_{L^2_{x,p}}^2.\)
\end{proof}

\subsubsection{Strategy of the proof}\label{subsecstratproof}

In this section we outline the proofs of Theorem \ref{Torus Existence} and Theorem \ref{thmvacuum}. The broad strategy of proof of Theorem \ref{Torus Existence} is familiar from previous works in the nonrelativistic and special relativistic settings \cite{MR1211782, 1904.12086}.  However previous approaches would not be sufficient to close our estimates globally in time due to the $\tqq$ rescaling from the cosmological expansion in the FLRW metric.  More precisely in \secref{sectionmacromicro} we derive the macroscopic equations using a  time-rescaled orthonormal momentum basis.  This new approach is essential to obtaining the precise large time estimates for the solution with cosmological expansion.  We also expect it to be useful in future developments.  Furthermore we determine the precise effects of cosmological expansion on the large time behavior of solutions, and we develop new methodologies to overcome difficulties due to the masslessness of the particles. The proof of Theorem \ref{thmvacuum} relies solely on time decay that is caused by cosmological expansion.

\textit{The energy method}. The proof of Theorem \ref{Torus Existence} is based on Guo's energy method \cite{Guo-VMB} to control the perturbation \(f\) of the Maxwell--J\"uttner equilibrium.  We also use the low regularity Wiener algebra method using the Fourier transform from \cite{1904.12086}.  We will use the \emph{energy norm} \(\|\tTq f\Vert_{L^1_kL^\infty_TL^2_{\pv}}\) and the \emph{dissipation norm} \(\|t^{3\qq/2} f\|_{L^1_kL^2_TL^2_{\pv}}\). The main energy estimate concerning these norms states that for all \(T\geq 1\),
\begin{equation}\label{mainestimateintro0}
        {\Vert{\tTq f}\Vert_{L^1_kL^\infty_TL^2_{\pv}}}
+
{\norm{t^{3\qq/2} f}_{L^1_kL^2_TL^2_{\pv}}}
\lesssim 
\norm{ {f}_1}_{L^1_kL^2_{\pv}}
,
\end{equation}
which is established under the assumption that the initial perturbation is sufficiently small. From \eqref{mainestimateintro0} global existence of the perturbation \(f\) can be shown by a standard continuity argument. We outline below the main arguments to derive \eqref{mainestimateintro0}.

\textit{The basic energy estimate.} Recall that the perturbation \(f\) satisfies equation \eqref{rBoltz00} by the massless Boltzmann equation. Applying the Fourier transform in \(x\) to equation \eqref{rBoltz00},
\begin{equation}\label{BoltzFour}
\left(\partial_t  + i\frac{\pv}{|\tq \pv|}\cdot k 
-\frac{2\qq }{t} p^i \partial_{p^i} \right) \hat{f}(t,k,\pv)
+\linL\hat{f}(t,k,\pv)=\widehat{\Gamma}(\hat{f},\hat{f})(t,k,\pv),
\end{equation}
where $\widehat{\Gamma}(\hat{f},\hat{f})$ corresponds to the nonlinearity of the system on the Fourier side (defined as in \eqref{def.GaF}). Then, an energy estimate can be obtained by taking the product of \eqref{BoltzFour} with the complex conjugate \(\bar{\hat{f}}\) and taking the real part
\begin{multline}\label{pro.tpm.p10}
\frac{1}{2} \tQQq \Vert \hat{f}\Vert_{L^2_{\pv}}^2(t,k)+\int^t_1  \tau^{6\qq}\real{\irnp{\linL\hat{f}}{\hat{f}}}d\tau
\\
=\frac{1}{2}\Vert \hat{f}_1(k,\cdot)\Vert_{L^2_{\pv}}^2
+
\int^t_1 \tau^{6\qq}\real{\irnp{\hat{\Gamma}(\hat{f},\hat{f})}{\hat{f}}}
d\tau.
\end{multline}
This energy estimate calls for a coercivity estimate for the linearised operator \(\mathcal{L}\) (defined by \eqref{L0}).

\textit{Coercivity of the linearised operator}.  Observe that the linearised operator \(\mathcal{L}\) has a 5-dimensional null space given by 
\begin{equation}\notag 
{\CN}(\linL)= {\rm span}\left\{
\tTq \sqrt{J(\tqq p)},~~\tTq (\tqq p^{i})  \sqrt{J(\tqq p)},  ~~
\tTq |\tqq \pv |\sqrt{J(\tqq p)}\right\},
\end{equation} 
over smooth functions of \(t\) and \(x\). The elements in \({\CN}(\linL)\) arise from the conservation laws of mass, energy, and momentum discussed in \eqref{mass_mom_entrop}. For \(\qq>0\), a new difficulty compared to previous works is that now the basis above is time-dependent.

Due to the non-trivial kernel, one can only hope to prove a coercivity estimate for \(\mathcal{L}\) modulo ${\CN}(\linL)$. We will prove in Proposition \ref{lowerN} that there exists a uniform constant \(\delta_0>0\) such that for all \(h\in L^2_p\) and all \(t\geq 1\),
\begin{equation}\label{coercivitylinearintro}
  \tTq\langle \linL h, h \rangle
\geq
\delta_0 
\| \micro  h \|_{L^2_p}^2.
\end{equation}
where \(\mathbf{I}\) is the identity operator on $L^2_p(\rth)$, and ${\bf P}$ the orthogonal projection from $L^2_p(\rth)$ onto ${\CN}(\linL)$.

Recall the decomposition of the linearised operator 
\(
\mathcal{L}h=\nu h-\mathcal{K}h
\)
in terms of the integral operator \(\mathcal{K}\) defined in \eqref{compactK}, and the multiplication operator \(\nu(t,p)= \nu_0 t^{-3\qq}\) computed in Proposition \ref{LEM:nuASYMP}. The coercivity estimate \eqref{coercivitylinearintro} follows by using an explicit decomposition of \(\mathcal{K}\) into a small piece and a compact piece, using suitable cut-off functions in the integral form of the operator. Indeed, we show in Proposition \ref{noKestimate} that for any small \(\eta>0\) we can explicitly decompose \(\mathcal{K}\) in terms of a compact operator \(\mathcal{K}_c\) by
\begin{equation}\label{smallcompactdecomp0}
\mathcal{K}=\mathcal{K}_c+\mathcal{K}_s, \qquad \text{where}\quad  \|\mathcal{K}_s\|_{L^2_p\to L^2_p}\leq \eta,
\end{equation}
where \(\|\cdot\|_{L^2_p\to L^2_p}\) is the operator norm from \(L^2_p\) to \(L^2_p\).

Recall the decomposition \(\mathcal{K}=\mathcal{K}_1-\mathcal{K}_2\) in \eqref{compactK}. Following the special relativistic calculation in \cite[Appendix]{MR2728733} or \cite{MR635279}, the operator \(\mathcal{K}_2\) can be written as
\begin{equation}\label{K2intker}
      \compK_2 h =     t^{3\qq}\int_{\mathbb{R}^3} ~  dq ~
k_{2}(\tqq p,\tqq q) ~ h(\qv),
\end{equation}
in terms of the explicit kernel \(k_2(t^{2\qq}p,t^{2\qq}p)\) given by
\[k_2(\tqq p,\tqq q) = ~
2c_1\bigg(1+\frac{|p|+|q|}{|p-q|}+2\frac{|p|+|q|}{t^{2\qq}|p-q|^2}\bigg)\left| \frac{\pv}{\pZ} -\frac{\qv}{\qZ}\right|^2
 \frac{e^{ -t^{2\qq} |p-q|/2}}{t^{2\qq}|p-q|},\]
for some $c_1>0$. A general integral operator of the form \eqref{K2intker} with kernel in \(L^2(\mathbb{R}^3_p\times \mathbb{R}^3_q)\) is compact, however \(k_2 \notin L^2(\mathbb{R}^3_p\times \mathbb{R}^3_q)\) because of the decay at infinity and the singularity at \(\{p=q\}\). With this in mind, we decompose \(\compK_2=\compK_{2c}+\compK_{2s}\) in terms of operators defined by the following kernels, 
\[
k_{2c}(p,q)=\phi_{R, \epsilon} k_2(\tqq p,\tqq q),\qquad k_{2s}(\tqq p,\tqq q)= (1-\phi_{R, \epsilon})k_2(\tqq p,\tqq q),
\]
where \(\phi_{R, \epsilon}\) is a cut-off supported in the domain \(\{t^{2\qq}|p|+t^{2\qq}|q|\leq R\}\cap \{ t^{2\qq}|p-q|\geq 2\epsilon\}\) for constants \(R>\epsilon>0\). The operator \(\compK_{2c}\) is then compact since \(k_{2c}\in L^2(\mathbb{R}^3_p\times \mathbb{R}^3_q)\), and the operator norm of \(\compK_{2s}\) is small for \(R\) sufficiently large and \(\epsilon\) sufficiently small. On the other hand, in this case of hard-ball interactions, the operator \(\mathcal{K}_1\) is compact since its kernel belongs to \(L^2(\mathbb{R}^3_p\times \mathbb{R}^3_q)\) (however \(\mathcal{K}_1\) is also decomposed as \eqref{K.one.decomp} so that \(\mathcal{K}_c\) additionally satisfies the estimate \eqref{compact.improved.est}).

The coercivity estimate \eqref{coercivitylinearintro} can then be proved by contradiction using the decomposition \eqref{smallcompactdecomp0} (see Proposition \ref{lowerN}). However, the coercivity estimate \eqref{coercivitylinearintro} only controls the perturbation \(f\) projected out of ${\CN}(\linL)$. To close the energy estimate for the perturbation \(f\), we will need further estimates for the projection \({\bf P}f\) into ${\CN}(\linL)$.

\textit{The macroscopic estimates}. In the following, we will decompose the perturbation as \(f=\mathbf{P}f+\{\mathbf{I}-\mathbf{P}\}f,\) into its \emph{macroscopic part} \(\mathbf{P}f\) and \emph{microscopic part} \(\{\mathbf{I}-\mathbf{P}\}f\). We write \(\mathbf{P}f\) as a linear combination of basis vectors
\begin{equation}\notag 
{\bf P}f=\left(  \mathcal{A}^f+\sum_{i=1}^{3} \mathcal{B}^f_i\tqq p^i/2+\frac{1}{\sqrt{3}}\mathcal{C}^f(\vnrm{\tqq\pv}-3)\right) \tTq\sqrt{J(\tqq p)},
\end{equation}
in terms of the coefficients \(\mathcal{A}^f(t,x)\), \(\mathcal{B}^f(t,x)\), and \(\mathcal{C}^f(t,x)\), given by the explicit momentum averages of \(f\) stated in \eqref{defmacrocoef}. Plugging the decomposition \(f=\mathbf{P}f+\{\mathbf{I}-\mathbf{P}\}f\) into \eqref{rBoltz00}, we obtain 
\begin{multline}\notag 
\left\{\partial_t+\frac{\pv^i}{|\tq \pv|} \partial_{x^i} -\frac{2\qq }{t} p^i \partial_{p^i} \right\}\macro f
= -\left\{\partial_t -\frac{2\qq }{t} p^i \partial_{p^i} \right\}\micro f
\\ 
-\frac{\pv^i}{|\tq \pv|} \partial_{x^i} (\micro f)  -\linL(\micro f)
+ \Gamma(f,f).
\end{multline}
When the operator \(\partial_t + \frac{\pv^i}{|\tq \pv|} \partial_{x^i} 
    -\frac{2\qq }{t} p^i \partial_{p^i} \) on the right hand side acts on ${\CN}(\linL)$, the result lives in a 13-dimensional space. Then, by considering a basis \(\{t^{3\mathfrak{q}}\mathfrak{e}_{j}(t^{2\mathfrak{q}}p)\}_{j=1}^{13}\) of this space, we will obtain a system of local conservation laws \eqref{local.cons} and macroscopic equations \eqref{macroscopic} for the coefficients \(\mathcal{A}^f\), \(\mathcal{B}^f\), and \(\mathcal{C}^f\). In particular, by the local conservation laws and \eqref{zeroconserv}, the coefficients have zero spatial average for all \(t\geq 1\).

The system of macroscopic equations for \(\widehat{\mathcal{A}}(t,k)\), \(\widehat{\mathcal{B}}(t,k)\), and \(\widehat{\mathcal{C}}(t,k)\), is estimated on the Fourier side. Recall that the spatial average of the coefficients vanish, \(\widehat{\mathcal{A}}(t,0)=\widehat{\mathcal{B}}(t,0)=\widehat{\mathcal{C}}(t,0)\) for all \(t\geq 1\). We first estimate \(\widehat{C}\) which satisfies the following schematic equations
\begin{align}\label{macroeqnsCo1}
    \frac{2}{\sqrt{3}}\tmq ik_j  \widehat{\mathcal{C}}+ \tmTq\partial_t \left(\tTq \widehat{\mathcal{B}}_j \right)  &=
\mathrm{RHS}_1,\qquad 1\leq j\leq 3,\\
 \tmTq\partial_t \left(\tTq \widehat{\mathcal{C}} \right)
  +
\frac{1}{2\sqrt{3}}  \tmq  ik \cdot \widehat{\mathcal{B}}
&=\mathrm{RHS}_2, \label{macroeqnsCo2}
\end{align}
where the right hand sides \(\mathrm{RHS}_1\), and \(\mathrm{RHS}_2\), are defined in terms of \(L^2_p\)-projections using the orthonormal functions \(\{t^{3\mathfrak{q}}\mathfrak{e}_{j}(t^{2\mathfrak{q}}p)\}_{j=1}^{13}\) considered earlier, as for example \(  \langle  \micro \hat{f},\,\tTq \mathfrak{e}_{\ell}(\tqq \pv)\rangle\), and also nonlinear terms. Multiplying \eqref{macroeqnsCo1} by  the term$-t^{4\qq} ik_j  \overline{\widehat{\mathcal{C}}}(t,k)$, summing on \(j\), and integrating in time, one obtains
\begin{multline}\notag 
 \frac{2}{\sqrt{3}}\int_1^{T} \tTq|\widehat{\mathcal{C}}(t,k)|^2  dt
=
\left. \frac{i}{|k|^2} \left( \left( \widehat{\mathcal{B}}\cdot k  \right)t^{4\qq}\overline{\widehat{\mathcal{C}}}(t,k) \right)\right|_{t=1}^T \\
-
\frac{i}{|k|^2} \int_1^{T} \left( \widehat{\mathcal{B}}\cdot k  \right) \tTq\partial_t  \left(\tq\overline{\widehat{\mathcal{C}}}(t,k)\right) dt+\mathrm{RHS_3},
\end{multline}
for some right hand side \(\mathrm{RHS_3}\). 
The first term on the right hand side can be bounded by \(T^{6\qq} \| \hat{f}(T,k)\|_{L^2_{\pv}}^2\) and \(\| \hat{f}_1(k)\|_{L^2_{\pv}}^2\), whereas the second term can be bounded using \eqref{macroeqnsCo2} to estimate the term involving \( \widehat{\mathcal{C}}\). As a result, one can show 
\begin{multline}\label{FT.C.estimate0}
\int_1^{T} \tTq|\widehat{\mathcal{C}}(t,k)|^2  dt
\lesssim
 T^{6\qq}\Vert \hat{f}(T,k)\Vert_{L^2_{\pv}}^2+\norm{ \hat{f}_1(k)}_{L^2_{\pv}}^2
  +
\int^{T}_{1} \tTq \left|\widehat{{\mathcal{B}}}(t,k)\right|^2 dt  \\
+\int_1^{T} \tTq\norm{\micro \hat{f}(t,k)}_{L^2_{\pv}}^2dt
+\mathrm{NLE},
\end{multline}
where \(\mathrm{NLE}\) denotes non-linear errors. The second integral on the right hand side of \eqref{FT.C.estimate0} comes from estimating \(\mathrm{RHS}_3\).
In a similar fashion, one can later estimate \(\widehat{A}-\sqrt{3}\widehat{C}\), and finally \(\widehat{B}\), exploiting an upper triangular structure in the system. The main estimate for the macroscopic coefficients is then 
\begin{equation}\label{energy_estimate_macro0}
\begin{split}
\norm{t^{3\qq/2} [\mathcal{A},\mathcal{B},\mathcal{C}]}_{L^1_k L^2_T} &\lesssim
\Vert t^{3\qq/2}\micro f \Vert_{L^1_k L^2_T L^2_{\pv}} +\Vert \tTq f \Vert_{L^1_k L^\infty_T L^2_{\pv}} + \Vert f_1 \Vert_{L^1_k L^2_{\pv}} \\ 
&\qquad+ {\Vert{\tTq f}\Vert_{L^1_kL^\infty_TL^2_{\pv}}} {\norm{t^{3\qq/2} f}_{L^1_kL^2_TL^2_{\pv}}}, 
\end{split}
\end{equation}
where \(\| [\mathcal{A},\mathcal{B},\mathcal{C}]\|\eqdef\| \mathcal{A}\|+\| \mathcal{B}\|+\| \mathcal{C}\|\). The estimate \eqref{energy_estimate_macro0} controls the dissipation norm of the macroscopic part since 
\begin{equation}\label{ineq:macrodiss}
\norm{t^{3\qq/2} \macro f}_{L^1_k L^2_TL^2_p}\lesssim 
\norm{t^{3\qq/2} [\mathcal{A},\mathcal{B},\mathcal{C}]}_{L^1_k L^2_T} .    
\end{equation}

\textit{Estimate of the nonlinear term}. To close the energy estimate, one needs to estimate the cubic nonlinear term on the right hand side of \eqref{pro.tpm.p10}. We will prove an \(L^2_p\)-trilinear estimate in Lemma \ref{thm1} according to which
\[
	 |\langle  \Gamma(f,h),\eta\rangle| \lesssim \norm{ f}_{L^2_p(\threed)}\norm{ h}_{L^2_p(\threed)}\norm{ \eta}_{L^2_p(\threed)},
\]
for all \(f\), \(h\), \(\eta\in L^2_p\). By using this trilinear estimate on the Fourier side, one can conclude for all \(f\), \(h\), and \(\eta\in L^2_p\) that the following uniform estimate holds
\begin{align}\label{ineq: Boltzmann nonlinear00}
\left| \irnp{\widehat{\Gamma}(\hat{f},\hat{h})(k)}{\hat{\eta}(k)}\right|\lesssim \norm{\hat{\eta}(k)}_{L^2_{\pv}}\int_{\Z^3}\Vert \hat{f}(k-l)\Vert_{L^2_{\pv}} \norm{\hat{h}(l)}_{L^2_{\pv}}  \,d\Sigma(l).
\end{align}
This estimate can be used to estimate the nonlinear term in \eqref{pro.tpm.p10}.

\textit{Global existence and uniform energy estimates}. We can then perform the main energy estimate for the perturbation \(f\). By using the coercivity estimate \eqref{coercivitylinearintro} on top of the basic energy estimate \eqref{pro.tpm.p10}, 
\begin{align}\notag
  \tTq \Vert \hat{f}\Vert_{L^2_{\pv}}(t,k)
&+
\left(\int^t_1  \tau^{3\qq}\Vert \micro \hat{f} \Vert_{L^2_{\pv}}^2(\tau,k)d\tau \right)^{\frac12}
\\  \notag
&\lesssim  \Vert \hat{f}_1\Vert_{L^2_{\pv}}(k) 
+
\left(\int^t_1 \tau^{6\qq} \left|{\irnp{\hat{\Gamma}(\hat{f},\hat{f})}{\hat{f}}}\right|
d\tau \right)^{\frac12}
.
\end{align}
Then, by using the identity \(\langle  \Gamma(f,f),h\rangle
=
\langle  \Gamma(f,f), \micro h\rangle\) on the Fourier side, together with the estimate \eqref{ineq: Boltzmann nonlinear00} for the nonlinear term, one obtains 
\begin{multline}\notag 
\norm{\tTq f}_{L^1_kL^\infty_TL^2_{\pv}}
+
{\norm{t^{3\qq/2}\micro f}_{L^1_kL^2_TL^2_{\pv}}}
\\
\lesssim 
\Vert {f}_1\Vert_{L^1_kL^2_{\pv}}
+{\Vert{\tTq f}\Vert_{L^1_kL^\infty_TL^2_{\pv}}} {\norm{t^{3\qq/2} f}_{L^1_kL^2_TL^2_{\pv}}},
\end{multline}
for all \(T\geq 1\). See Proposition \ref{Torus Boltzmann Micro}. Finally, applying the estimate \eqref{energy_estimate_macro0}--\eqref{ineq:macrodiss} for the macroscopic part \({\bf P}f\), one can show that for all \(T\geq 1\),
\begin{equation}\label{mainestimateintro}
        {\Vert{\tTq f}\Vert_{L^1_kL^\infty_TL^2_{\pv}}}
+
{\norm{t^{3\qq/2} f}_{L^1_kL^2_TL^2_{\pv}}}
\lesssim 
\norm{ {f}_1}_{L^1_kL^2_{\pv}}
+{\norm{\tTq f}_{L^1_kL^\infty_TL^2_{\pv}}} {\norm{t^{3\qq/2} f}_{L^1_kL^2_TL^2_{\pv}}}.
\end{equation}
The last term on the right hand side can be absorbed for a small initial perturbation, and \eqref{mainestimateintro0} concluded. Global existence for the massless Boltzmann equation follows from this estimate and a standard continuity argument.

\textit{Improved energy estimates}. A posteriori we show improved time-decay rates for the perturbation \(f\) when \(0\leq \mathfrak{q}\leq \frac{1}{3}\). These estimates rely on the computation of the collision frequency \(\nu (t,p)=\nu_0t^{-3\mathfrak{q}}\) in Proposition \ref{LEM:nuASYMP}, and a perturbative treatment of the integral operator \(\mathcal{K}\) based on the compact-small decomposition \eqref{smallcompactdecomp0}. Indeed, we consider the equation satisfied by \(t^{3\qq}\mathcal{T}_{\qq}(t)\hat{f}\) (with \(\tw(t)\) defined as in \eqref{weight.decay}) and then perform a similar energy estimate to \eqref{mainestimateintro}. We will show in Proposition \ref{estimated.nonlin.second} that  
\begin{align}\notag
       \forall t\geq 1,\qquad      {\Vert{\tTq \tw(t) f}\Vert_{L^1_kL^\infty_TL^2_{\pv}}}
+
{\norm{t^{3\qq/2} \tw(t) f}_{L^1_kL^2_TL^2_{\pv}}}
\lesssim 
\norm{ {f}_1}_{L^1_kL^2_{\pv}}.
\end{align} 
A posteriori we show weighted energy estimates for \( w_m \hat{f}\) in terms of the weights \(w_m^2(t^{2\qq}p)=1+|t^{2\qq}p|^m\), and propagate spatial regularity by performing energy estimates for \(\langle k\rangle^{l} \hat{f}(t,k,p)\) on the Fourier side. See Proposition \ref{prop:weightenergy} and Proposition \ref{prop:spatreg}, respectively.

\textit{The proof of Theorem \ref{thmvacuum}}. The near vacuum result is also based on the energy method. For this problem, the main energy estimate states that for all \(T\geq 1\),
\begin{equation}\label{apriorivacuumestimate1}
{\norm{t^{3\qq} f}_{L^1_kL^\infty_TL^2_{\pv}}}
\lesssim 
\norm{f_1}_{L^1_kL^2_{\pv}}.    
\end{equation}
The linearised operator around vacuum is trivial, and the estimate \eqref{apriorivacuumestimate1} is therefore established using the time decay arising from expansion only. In particular, previous difficulties regarding the kernel and coercivity of the linearised operator are not present anymore. 

Recall that the perturbation \(f\) satisfies \eqref{perturb.massless.rBoltz0}. Performing the basic energy estimate, for all \(t\geq 1\),
\begin{align}\notag 
\frac{1}{2} \tQQq \Vert \hat{f}\Vert_{L^2_{\pv}}^2(t,k)=&~\frac{1}{2}\Vert \hat{f}_1(k,\cdot)\Vert_{L^2_{\pv}}^2+\int^t_1 \tau^{6\qq}\real{\irnp{\hat{\Gamma}(\hat{f},\hat{f})}{\hat{f}}}
d\tau.
\end{align}
The estimate \eqref{apriorivacuumestimate1} is then shown by estimating the time integral of the nonlinear term on the right hand side. Using the trilinear estimate \eqref{ineq: Boltzmann nonlinear00} on the Fourier side, it can be shown that for all \(\tau\geq 1\),
\begin{multline*}
   \tau^{6\qq}\left| \irnp{\hat{\Gamma}(\hat{f},\hat{f})(\tau,k)}{\hat{f}(\tau,k)}\right|
   \\
   \lesssim 
   \tau^{-3\qq}{\norm{t^{3\qq} \hat{f}(k)}_{L^\infty_TL^2_{\pv}}}
\int_{\Z^3}  {\norm{t^{3\qq} \hat{f}(k-l)}_{L^\infty_TL^2_{\pv}}}{\norm{t^{3\qq} \hat{f}(l)}_{L^\infty_TL^2_{\pv}}} \,d\Sigma(l). 
\end{multline*}
Assuming \(\qq>\frac{1}{3}\), so that $\tau^{-3\qq}$ is integrable, one can show that,
\begin{align*}
\forall T\geq 1,\qquad {\norm{t^{3\qq} f}_{L^1_kL^\infty_TL^2_{\pv}}}
\lesssim   \norm{f_1}_{L^1_kL^2_{\pv}}
+
{\norm{t^{3\qq} {f}}^{2}_{L^1_k L^\infty_TL^2_{\pv}}}
.  
\end{align*}
The last term on the right hand side can be absorbed for a small initial perturbation, and \eqref{apriorivacuumestimate1} concluded. Global existence follows from this estimate and a standard continuity argument.

\subsection{Discussion of more general collision kernels}\label{generalkernels}

In this section, we discuss scattering kernels of the form 
\begin{equation}\label{kernelgeneral2}
    \sigma (\mrel,\theta ) = \mrel^a \sigma_0(\theta), \quad  \quad -4<a\leq 2,
\end{equation}
for an integrable non-negative function \(\sigma_0(\theta)\). We first formulate a conjecture on the dissipation threshold for these more general kernels. We then discuss the behaviour of the associated model problem.

\subsubsection{Dissipation threshold for more general kernels}

We expect the Maxwell--J\"uttner equilibrium to be future nonlinearly stable for all kernels of the form \eqref{kernelgeneral2}. Recall (see Remark \ref{rmkdissthres}) that, for \(a=0\), expansion rate \(\qq=1/3\) is a threshold below which there is dissipation, and above which there is not. We expect a different threshold value of \(\qq\) to arise for the more general kernels in \eqref{kernelgeneral2}.

\begin{conjecture}[Dissipation threshold at \((3+a)\qq=1\) for \(-3<a\leq 2\).]
    Consider \(\qq\in [0,1]\), and the massless Boltzmann equation with scattering kernel \eqref{kernelgeneral2} on FLRW \eqref{rBoltz.eqn0}, \eqref{collisionop}. The Maxwell--J\"uttner equilibrium \(J(t^{2\qq}p)\) is future nonlinearly stable. Moreover, for any \(\max\{3+a, 0\}\qq\leq 1\) the norm \(\|t^{3\qq}f(t)\|_{L^{1}_kL^2_p}\) decays uniformly as \(t\to +\infty\) --- with quantitative rates, depending on \(\qq\) in \eqref{metricFLRW} and on \(a\) in \eqref{kernelgeneral2} --- and in general the norm $\|t^{3\qq}f(t)\|_{L^{1}_kL^2_p}$ does not decay for \(\max\{3+a, 0\}\qq>1\), where the perturbation \(f\) is defined by \eqref{standard.perturbation}.  Notably for \(-4<a\leq -3\) the norm \(\|t^{3\qq}f(t)\|_{L^{1}_kL^2_p}\) decays uniformly as \(t\to +\infty\) with quantitative rates for any \(\qq\in [0,1].\)
\end{conjecture}

In particular, for \(a\in (-4,-1]\) we expect that \(t^{3/2}\|f(t)\|_{L^1_{k}L^2_{p}}\) decays for solutions of the massless Boltzmann equation with scattering kernel \eqref{kernelgeneral2} on FLRW with scale factor \(\qq=\frac{1}{2}\) (which satisfies the Einstein--massless Boltzmann system, as discussed in Section \ref{subsecMJFLRW}). Moreover, in the case of Israel particles, in which \(\sigma(\varrho,\theta)\) is defined by \eqref{israelkernel}, we expect \(\qq=\frac{1}{2}\) to be the dissipation threshold.

In the next section, we will consider a linear model problem that motivates the threshold scale factor \(\max\{ 3+a, 0\}\qq=1\) considered earlier. From this model problem one can also extract quantitative decay rates, which are informative regarding the decay rates one may expect for the massless Boltzmann equation with these scattering kernels.

One may formally compare the above threshold behaviour with the Euler equations on FLRW where, for each linear equation of state, there is an expected threshold value of expansion rate \(\qq\), below which shock waves form in perturbations of constant density fluids at rest, and above which such shock waves do not form \cite{FMOOW}.

\subsubsection{Model problem for more general kernels}

Motivated by the stability of the Maxwell--J\"uttner equilibrium for the massless Boltzmann equation on FLRW spacetime \((\mathcal{M},g_{\qq})\) in the case of scattering kernels of the form \eqref{kernelgeneral2}, we consider the following model problem for $a>-4$
\begin{equation} \label{eq:Boltzmanntoygeneral}
	\partial_t f
	+
	\frac{p^i}{\vert t^\qq p \vert} \partial_{x^i} f
	-
	\frac{2\qq}{t} p^i \partial_{p^i} f
	+
	 \nu_0\frac{ \vert \tqq p \vert^{\frac{a}{2}}}{t^{(3+a)\qq}} f
	=
	0,\quad f(t=1,x,p)=f_1(x,p),
\end{equation}
with \(\nu_0>0\) a fixed constant. We have written the last term using \(t^{2\qq}p\) since this weight belongs to the kernel of the transport operator \(\partial_t 
	+
	\frac{p^i}{\vert t^\qq p \vert} \partial_{x^i} 
	-
	\frac{2\qq}{t} p^i \partial_{p^i}\) in \eqref{eq:Boltzmanntoygeneral}. This model problem should be compared with the problem \eqref{rBoltz.model.second} previously considered in the hard-ball case.

When considering more general scattering kernels, the linearised operator \(\mathcal{L}\) in the Boltzmann equation can also be decomposed as in \eqref{decomlinop} in terms of the multiplication operator \(\nu(t,p)\) in \eqref{nuDEF}. We now compute \(\nu(t,p)\) for the kernel \eqref{kernelgeneral2}.

\begin{proposition}
    For scattering kernels of the form \eqref{kernelgeneral2} with an integrable function $\sigma_0(\theta)$, the collision frequency \eqref{nuDEF} is 
    \[
    \nu(t,p)= \nu_0\frac{ |t^{2\qq }\pv|^{\frac{a}{2}} }{t^{(3+a)\qq}},
    \] where \(\nu_0\) is a fixed constant depending on \(a\).
\end{proposition}

\begin{proof}
    By the explicit form \eqref{moller} of the Møller velocity in FLRW, 
    \begin{align} \notag 
    \nu(t,p) =&~
\frac{t^{3\qq}}{16\pi}
\int_{\mathbb{R}^3} ~  dq 
\int_{\mathbb{S}^{2}} ~ d\omega
~  \left(1 - \frac{p\cdot q}{\pZ\qZ}\right) ~
\mrel^a \sigma_0(\theta)
    e^{-{| \tqq q|}}
    \\ \notag
    =&~
\frac{t^{-3\qq}}{16\pi}
  |\pv|^{\frac{a}{2}}
\int_{\mathbb{R}^3} ~  dq 
\int_{\mathbb{S}^{2}} ~ d\omega ~\sigma_0(\theta)
~  \left(1 - \frac{p\cdot q}{\pZ\qZ}\right) ~
  \left(|q| - \frac{p}{\pZ}\cdot q\right)^{\frac{a}{2}}  e^{-{|  q|}}
      \\ \notag
    =&~
    \nu_0 |\pv|^{\frac{a}{2}} t^{-3\qq},
\end{align}
in terms of a uniform positive constant \(\nu_0\) depending on \(a\). 
\end{proof}

The model problem \eqref{eq:Boltzmanntoygeneral} then arises from neglecting nonlinear terms and the operator \(\mathcal{K}\) in the massless Boltzmann equation \eqref{rBoltz00}. This problem is distinguished from the more well known cases of {\bf (I)} the massive special relativistic Boltzmann equation and {\bf (II)} the Newtonian Boltzmann equation by at least two complicating factors. First, the massless case presents a degeneracy when $a>0$ for small values of $\vert \tqq p \vert$ in the damping factor ${ \vert \tqq p \vert^{\frac{a}{2}}} t^{-(3+a)\qq}$. Second, the damping is degenerating for large $t$ with $t^{-(3+a)\qq}$.  For the soft potentials $a<0$, this degeneration is diluted by small values of $\vert \tqq p \vert$ in an energy norm, which makes ${ \vert \tqq p \vert^{\frac{a}{2}}}$ large, and thereby increases the damping rate. Note one can show using techniques from \cite{MR2366140} the following decay rates for the \(L^2\) norm of solutions \(f\) to \eqref{eq:Boltzmanntoygeneral}.

\begin{proposition}[Decay estimates for general model problem]\label{proplineargeneralmodel}
Let \(\qq \in [0,1]\), \(-4<a\leq 2\), \( s>0 \), and \(f_1\) a sufficiently regular initial data. If either \(-3<a\leq 2\) and  $(3+a)\qq < 1$ or \(-4<a\leq -3\) and \(\qq \in [0,1]\), then the unique solution \(f\) of the model problem \eqref{eq:Boltzmanntoygeneral} satisfies 
\begin{align}\notag 
 \forall \,t\geq 1,\qquad       			\|t^{3\qq} f(t)\|_{L^2_{x,p}}
        \lesssim
		e^{-C t^{\dr}} \| e^{ \frac{{  \nu_0}}{s}(2s +|a|)|p|^s}f_1\|_{L^2_{x,p}},
\end{align}
where \( d  = \frac{2s}{2s+|a|} \left( 1-(3+a)\qq\right) \in [0, 1)\), and $C>0$ is a uniform constant. Otherwise, for \(-3<a\leq 2\) and $(3+a)\qq = 1$ the unique solution \(f\) of \eqref{eq:Boltzmanntoygeneral} satisfies 
\begin{align} \notag
  \forall \,t\geq 1,\qquad      			\|t^{3\qq} f(t)\|_{L^2_{x,p}}
      &  \lesssim
		e^{-C \left(\log (t)\right)^{\frac{2s}{2s +|a|}}} \| e^{  \frac{{ \nu_0}}{s}(2s +|a|)|p|^s}f_1\|_{L^2_{x,p}},
\end{align}
where $C>0$ is a possibly different uniform constant. 
\end{proposition}

This behavior is very different from the  massive cases {\bf (I)} and {\bf (II)} when $\qq=0$.  In those cases the soft potentials $a<0$ are well known to substantially reduce the damping rate from exponential to sub-exponential.  Specifically in those other cases one expects exponential time decay for $a\geq 0$ as in \cite{Ukai-86,MR1211782}, and one only expects rapid polynomial or super-polynomial decay for $a<0$ as in \cite{MR2728733,strain2006almost,MR2366140,Caf1,Ca-1980}. Thus we see a {\it reduction} in the time decay rates when we {\it reduce} the value of $a$ in {\bf (I)} and {\bf (II)}.  In contrast in Proposition \ref{proplineargeneralmodel} above we have an {\it improvement} in the range of $\qq>0$ for which the fastest superpolynomial decay rates hold when {\it reducing} the value of $a$.  Specifically, for $-4<a\leq -3$, in Proposition \ref{proplineargeneralmodel} we obtain the superpolynomial decay rates for any value of $\qq \in(0,1]$.

\subsection{Related works} 
Let us present here a nonexhaustive overview of related previous works.

\subsubsection{Collisional kinetic theory in general relativity}
The formulation of the kinetic theory of gases in general relativity began with the early contributions of Synge \cite{Synge1934} in 1934, and Walker \cite{walker1936boltzmann} in 1936. See also the work by Einstein \cite{Einstein1938} in 1939. Concerning the collisional theory, the Boltzmann equation was derived in special relativity by Lichnerowicz--Marrot \cite{MR0004796} in 1940, and in general relativity by Chernikov \cite{Chernikov1961,zbMATH03209008} in 1961. Around the same time Israel derived his well-known collision kernel for the relativistic Boltzmann equation \cite{MR165921} in 1963. The collisional theory was further developed in works by Lindquist \cite{LINDQUIST1966487} in 1966, Bel \cite{Bel1969} in 1969, Ehlers \cite{Ehlers1971GeneralRA} in 1969, Ehlers--Sachs \cite{EhlersSachs1971} in 1971, among other. Later physics references include the works of Dijkstra--Van Leeuwen  \cite{MR508285,MR471665} in 1978, Naumov \cite{Naumov1982} in 1982, Ray \cite{Ray1982} in 1982, Hiscock--Salmonson \cite{hiscock1991dissipative} in 1991, and Kremer \cite{Kremer2012,Kremer2013} around 2012. See also the introductions to relativistic kinetic theory by Sarbach--Zannias \cite{Sarbach2013} from 2013, and Acu\~na-C\'ardenas--Gabarrete--Sarbach \cite{AcuaCrdenas2022} from 2022. For textbook references, we refer to van Weert--De Groot--van Leeuwen \cite{MR635279} from 1980, to Bernstein \cite{Bernstein1988} from 1988, and to Cercignani--Kremer \cite{CercignaniKremer2002} from 2002.

\subsubsection{Local well-posedness for the Einstein--Boltzmann system}

Regarding rigorous mathematical treatments of the full Einstein--Boltzmann system, note the work of Bancel--Choquet-Bruhat \cite{MR356790} on local well-posedness from 1973; this result builds upon \cite{MR0337248,MR0308609}.  Later, Rendall--Lee \cite{MR3043490} revisited the problem of local well-posedness for Einstein--Boltzmann in 2013. For textbook references about the Einstein--Boltzmann system see the books by Rendall \cite{book:1276737} from 2008, Choquet-Bruhat \cite{ChoquetBruhat2009} from 2009, and Ringstr\"om \cite{MR3186493} from 2013.

\subsubsection{Global existence of equilibria for the Boltzmann equation in special relativity}

For the mathematical work on the asymptotic stability of Maxwell--J\"uttner equilibria for the Boltzmann equation without symmetry assumptions, there have been several works on fixed backgrounds. Glassey--Strauss \cite{MR1211782} proved the stability of equilibria for the relativistic Boltzmann equation on \(\T^3_x\) with hard interactions in 1993. Glassey--Strauss \cite{MR1321370} extended that result to the whole space \(\R^3_x\) in 1995. Later, the relativistic Boltzmann equation in \(\R^3_x\) with soft interactions was addressed by Strain--Zhu \cite{MR2911100} in 2012. The problem of deriving energy estimates with momentum regularity in these settings was resolved by Guo--Strain \cite{MR2891870} in 2012.

More recently, a definitive stability of the Maxwellian equilibrium for the relativistic Boltzmann equation on \(\T^3_x\) without angular cut-off was obtained by Jang--Strain \cite{RelBolNoncut2020} in 2022. In the non-relativistic case, nonlinear stability for the Boltzmann equation on \(\T^3_x\) without angular cut-off was shown earlier in  \cite{MR2784329} by Gressman--Strain in 2011. This stability problem in the non-relativistic case was later revisited by Duan--Liu--Sakamoto--Strain \cite{MR471665} who obtained global existence in a low regularity Chemin--Lerner function space. In the present article, the function space we use is motivated by that function space. Global existence of Maxwellians for the non-relativistic Boltzmann equation has been studied in many other works. See, for example, \cite{AMUXY-2011-CMP, AMUXY-2012-JFA, MR3213301, MR3532066, DSa, UY, Guo-2010} and the references therein. 

 We also comment the result on the Hilbert expansion for the relativistic Boltzmann equation on \(\R^3_x\) by Speck--Strain \cite{MR2793935} in 2011.

\subsubsection{Global existence results for the spatially homogeneous Boltzmann equation in cosmology}

The Boltzmann equation and the Einstein--Boltzmann system have both been studied previously in cosmology under symmetry assumptions. For certain kernels, for large data, global existence for spatially homogeneous FLRW solutions of the Einstein--Boltzmann system with a positive cosmological constant was obtained by Noutchegueme--Takou \cite{NoutcheguemeTakou2006} in 2006. Later, the asymptotic behaviour of these solutions was studied by Takou \cite{Takou2009} in 2009. On fixed Minkowski and FLRW backgrounds, large data global existence for spatially homogeneous solutions of the Boltzmann equation with certain hard interactions was studied by Lee--Rendall \cite{lee2013spatially} in 2013.

Also, small data global existence for the Boltzmann equation and the Einstein--Boltzmann system have been considered. Global existence near vacuum for the Einstein--Boltzmann system with Bianchi I symmetry was shown by Lee--Nungesser \cite{MR3620020} in 2017. Later, small data global existence for the Einstein--Boltzmann--scalar field system with Bianchi I-VII symmetry for Israel particles was shown by Lee--Lee--Nungesser \cite{LeeLeeNungesser2023} in 2023.  There is additionally the work of Lee which considers a spatially-homogeneous Newtonian cosmological setting \cite{LeeHo2016} from 2016.  We also refer to the global existence of homogeneous solutions for the Einstein--Boltzmann system with a positive cosmological constant by Lee--Nungesser \cite{LeeNungesser2024} in the case of certain soft interactions in 2024.

\subsubsection{Works on the spatially homogeneous massless Boltzmann equation in cosmology}

We discuss related works for massless Boltzmann in cosmology. Small data global existence for the spatially homogeneous massless Boltzmann equation on FLRW spacetimes was shown by Lee \cite{Lee2021} for some collision kernels in a range of soft and hard potentials in 2021. Note also the work \cite{Taylor2024} on spatially homogeneous FLRW solutions with \(\qq =1/2\) for the Einstein--massless Vlasov system with spatial topology \(\R^3_x\) in spherical symmetry. For data posed at the initial singularity, Lee--Nungesser--Tod \cite{MR4063736} proved finite-time existence for the massless Boltzmann equation on FLRW backgrounds for some soft scattering kernels in 2020. Together also with Stalker \cite{LeeNungesserStalkerTod2024}, this result was extended to Bianchi I spacetimes in 2024.

\subsubsection{Works on other equations on FLRW spacetimes}

There have been several works on other related equations on FLRW.  It was shown recently by Taylor--Velozo Ruiz \cite{TaylorVelozo2025} that a phase mixing effect occurs for the Vlasov equation on FLRW.  See Remark \ref{rmk:Vlasovmixing}. See also \cite{FajmanUrban2026} for the study of a related nonlinear, non-relativistic model.

There have been several works \cite{FOOW25a, FMOOW, S13} concerned with the future stability of homogeneous solutions for the Euler equations on FLRW spacetimes. In the recent work \cite{FMOOW} by Fajman--Maliborski--Ofner--Oliynyk--Wyatt it is suggested that shock formation occurs for slowly expanding spacetimes. See also \cite{FMOOW25a}. 

There have also been many works on the linear wave equation on FLRW.  In the decelerated regime, for the case where the spatial slices are flat copies of $\mathbb{R}^3$, the combined effects of dispersion and expansion are relevant.  See, for example, the recent articles of Nat\'{a}rio–Rossetti \cite{NR23} and Haghshenas \cite{H25} and references within.

\subsection{Outline of the paper}

In Section \ref{sectionmacromicro}, the distribution function is decomposed in its microscopic and macroscopic parts. Moreover, the system of macroscopic equations is derived. In Section \ref{sectcollisionop}, the main estimates for the collision operator are performed. In particular, the coercivity estimate for the linearised operator is shown. In Section \ref{macroscopic.eqns.sec}, the system of equations derived in Section \ref{sectionmacromicro} is used to estimate the macroscopic part of the distribution. In Section \ref{sectionWienernorm}, the energy estimates for the full solution \(f\) are shown. The proof of Theorem \ref{Torus Existence} is completed in Section \ref{sectproofsmainresults}, and the proof of Theorem \ref{thmvacuum} is given in Section \ref{sec.sol.nearby.vacuum}. Appendix \ref{appconservlaws} contains the proof of the H-theorem and the conservation of mass, energy, and momentum, discussed in Section \ref{subsconserved}. In Appendix \ref{MJparamet}, the Maxwell--J\"uttner parameters are identified for fixed values of the initial mass, energy, and momentum.

\subsection{Acknowledgements}
We thank the Simons Center for Geometry and Physics (SCGP) at Stony Brook University for its hospitality and collaborative environment.  This collaboration began at SCGP during the 3rd Simons Math Summer Workshop on “Partial Differential Equations of Classical Physics”, 7-25th July 2025.

\section{Equations for macroscopic quantities}\label{sectionmacromicro}

In this section, equations are derived for macroscopic quantities of solutions of the Boltzmann equation. In Section \ref{subsec_macromicro}, we decompose the distribution \(f\) into its macroscopic and microscopic parts. In Section \ref{subs_orth_norm14}, we introduce a set of suitable time-normalised orthonormal basis functions, and we show its orthonormality in Section \ref{orthogonal.vectors.sec}. Finally in Section \ref{derive.macroscopic}, we derive the macroscopic equations \eqref{macroscopic}--\eqref{macroscopic.B}.  Later in Section \ref{macroscopic.eqns.sec}, we will use \eqref{macroscopic}--\eqref{macroscopic.B} to estimate the  part of the distribution \(f\) that is contained in the null space of the linearised collision operator.    The calculations in this section hold for any $\qq \ge 0$. 

\subsection{Macro-micro decomposition of the distribution}\label{subsec_macromicro}

In this section, we decompose $f(t,x,p)$ into its macroscopic part \(\macro f\) and microscopic part \(\micro f\) as
\begin{equation}
\label{decomp}
f=\macro f+\micro f,
\end{equation}
where \(\macro\) is the orthogonal projection from $L^2_p(\rth)$ onto the null space of the linearised operator $\linL$ defined in \eqref{L0}.  The operator $\linL$ has a 5-dimensional null space given by
\[
\mathrm{nul}(\linL)=\mathrm{span}\Big\{ \tTq J^{\frac{1}{2}}(\tqq p),~~\tTq (\tqq p^{i})  J^{\frac{1}{2}}(\tqq p),  ~~
\tTq |\tqq \pv |J^{\frac{1}{2}}(\tqq p)\Big\},
\]
where the linear span is taken over smooth functions of \(t\) and \(x\). For convenience we work with an orthonormal basis of this space. Define the following functions
\begin{align} \notag
   \tTq \mathfrak{e}_1(\tqq \pv) \eqdef &~ \tTq\sqrt{J(\tqq p)},
   \\ 
   \tTq \mathfrak{e}_{j+1}(\tqq \pv) \eqdef &~ 
   \frac{1}{2}\tqq \pv^j \tTq\sqrt{J(\tqq p)},\qquad \text{for }\quad 1\leq j \leq 3, \label{orthogonal.basis.kernel}
   \\ \notag
   \tTq \mathfrak{e}_5(\tqq \pv) \eqdef &~ \frac{1}{\sqrt{3}}(\tqq\pZ-3)\tTq\sqrt{J(\tqq p)}.
\end{align}
We denote this set of functions as $\{ \tTq\mathfrak{e}_\ell(\tqq p)\}_{\ell=1}^5$.

\begin{proposition}\label{prop_orthnorm_nul}
    The set $\{ \tTq\mathfrak{e}_\ell(\tqq p)\}_{\ell=1}^5$ forms an orthonormal basis with respect to \(L^2(\R_p^3)\) of the space \(\mathrm{nul}(\linL)\) for all \(t\geq 1\).
\end{proposition}

The orthonormal property of $\{ \tTq\mathfrak{e}_\ell(\tqq p)\}_{\ell=1}^5$ is shown next in Section \ref{orthogonal.vectors.sec}.  Then, we write ${\bf P}f$ as a linear combination of the basis vectors above as
\begin{equation}
\label{macrofbasis}
{\bf P}f=\left(  \mathcal{A}^f+\sum_{j=1}^{3} \mathcal{B}^f_j\tqq p^j/2+\frac{1}{\sqrt{3}}\mathcal{C}^f(\vnrm{\tqq\pv}-3)\right) \tTq\sqrt{J(\tqq p)},
\end{equation}
where the coefficients are given by
\begin{equation}\label{defmacrocoef}
\begin{split}
    \mathcal{A}^f(t,x) & = \tTq\int_\rth f(t,x,\pv)\sqrt{J(\tqq p)}dp,
    \\
    \mathcal{B}_i^f (t,x)
    & =
    \frac{\tTq}{2}\int_\rth f(t,x,\pv) \tqq p^i\sqrt{J(\tqq p)}dp, 
    \\
    \mathcal{C}^f(t,x)
    &=
    \frac{\tTq}{\sqrt{3}}\int_{\rth}f(t,x,\pv)( |\tqq \pv |-3)\sqrt{J(\tqq p)} dp.
\end{split}
\end{equation}

\subsection{Orthonormal time-normalised basis}\label{subs_orth_norm14}

As will be seen in \secref{derive.macroscopic}, when the transport operator \(\partial_t + \frac{\pv^i}{|\tq \pv|} \partial_{x^i} 
    -\frac{2\qq }{t} p^i \partial_{p^i} \) acts on the span of the null space elements 
\begin{equation}
 \tTq \sqrt{J(\tqq p)},\qquad \tTq \frac{\tqq p^{i}}{2}  \sqrt{J(\tqq p)},  \qquad 
\frac{\tTq}{\sqrt{3}} (|\tqq \pv |-3)\sqrt{J(\tqq p)}
\end{equation}
over smooth functions of \(t\) and \(x\), the result lives in a 13-dimensional space.  We now extend the orthonormal velocity  functions $\{ \tTq\mathfrak{e}_\ell(\tqq p)\}_{\ell=1}^5$ in \eqref{orthogonal.basis.kernel}. To this end, for $1\leq l,j \leq 3$ we define  the additional orthonormal velocity functions
\begin{equation}\label{orthogonal.basis.use}
    \begin{split}
         \tTq \mathfrak{e}_{j+5}(\tqq \pv) \eqdef &~ 
   \sqrt{\frac{3}{7}}
\left(\frac{p^j}{{\pZ}} -\tqq \pv^j\right)\tTq\sqrt{J(\tqq p)},\qquad  \text{for }\quad 1\leq j \leq 3,
      \\ 
   \tTq \mathfrak{e}_{9}(\tqq \pv) \eqdef &~ 
{t^{2\qq}} \frac{\sqrt{5}}{2}\left(\frac{(p^1)^2 -(p^3)^2}{{|p|}} \right)t^{3\qq}\sqrt{J(t^{2\qq} p)},
      \\ 
   \tTq \mathfrak{e}_{10}(\tqq \pv) \eqdef &~ 
{t^{2\qq}} \frac{\sqrt{5}}{2} \left(\frac{(p^1)^2+(p^3)^2 -2(p^2)^2}{{|p|}} \right)t^{3\qq}\sqrt{J(t^{2\qq} p)},
      \\ 
   \tTq \mathfrak{e}_{i+j+8}(\tqq \pv) \eqdef &~ 
\frac{\sqrt{5}}{2}{\tqq}\frac{ \pv^l \pv^{j}}{|p|}\tTq\sqrt{J(\tqq p)},\qquad  \text{for }\quad  1\leq l<j \leq 3.
    \end{split}
\end{equation}
We denote this set of $L^2(\mathbb{R}^3_p)$ orthonormal functions with \eqref{orthogonal.basis.kernel} as $\{\tTq \mathfrak{e}_{\ell}(\tqq \pv)\}_{\ell=1}^{13}$.

We remark that the diagonal of the symmetric tensor $\frac{ \pv^l \pv^{j}-\frac{1}{3} |p|^2}{|p|}$ vanishes
\begin{align*}
    \frac{(p^1)^2-\frac13|p|^2}{{|\pv|}}+\frac{(p^2)^2-\frac13|p|^2}{{|\pv|}}+\frac{(p^3)^2-\frac13|p|^2}{{|\pv|}}=0.
\end{align*}
Thus the diagonal only contains two independent components that we represent in the components of  $\tTq \mathfrak{e}_{9}(\tqq \pv)$ and $\tTq \mathfrak{e}_{10}(\tqq \pv)$.  In this system of orthonormal functions, \eqref{orthogonal.basis.kernel} with \eqref{orthogonal.basis.use}, $\mathfrak{e}_{1}$ is the density mode, $\mathfrak{e}_{j+1}$ are the three momentum modes, $\mathfrak{e}_{5}$ is the energy mode, $\mathfrak{e}_{j+5}$ are the three heat-flux modes, $\mathfrak{e}_{9}$ and $\mathfrak{e}_{10}$ are the two independent diagonal components of the stress tensor modes,
and $\mathfrak{e}_{i+j+8}$ are the three off-diagonal stress tensor modes.  These orthonormal functions in \eqref{orthogonal.basis.kernel} and \eqref{orthogonal.basis.use} are motivated by the Grad 13-moment system introduced in \cite{Grad1949} in 1949, and also by the polynomial momentum basis used in \cite{MR2911100, RelBolNoncut2020} for the relativistic Boltzmann equation in the case of massive particles.

\begin{proposition}\label{proporthobasis}
    The set $\{\tTq \mathfrak{e}_{\ell}(\tqq \pv)\}_{\ell=1}^{13}$ forms an orthonormal family with respect to \(L^2(\R_p^3)\) for all \(t\geq 1\).
\end{proposition}

The orthonormal property of $\{\tTq \mathfrak{e}_{\ell}(\tqq \pv)\}_{\ell=1}^{13}$ is shown later in Section \ref{orthogonal.vectors.sec}.

\subsection{Orthonormal property of the bases}\label{orthogonal.vectors.sec} 

In this section, we prove Proposition \ref{proporthobasis}, in other words, we show that $\{\tTq \mathfrak{e}_{\ell}(\tqq \pv)\}_{\ell=1}^{13}$ is an orthonormal basis with respect to \(L^2(\R_p^3)\). This follows from some elementary integral computations that we write below.

It is well-known for the Gamma function \(\Gamma\) that \(\int_0^\infty \rho^n e^{-\rho} d\rho = n!\) for all \(n\in \{0,1,2,\ldots\}\). Hence
\begin{align}\label{integralp}
   \int_\rth | \pv |^n  \frac{e^{- |\pv|}}{8\pi} dp = \frac{(n+2)!}{2},\qquad \forall n\in \{0,1,2,\ldots\}. 
\end{align}
Recall the definition of \(J(\tqq \pv)\) in \eqref{juttner}. We use the integral \eqref{integralp} to calculate the values of various integrals related to $J(\tqq \pv)$ as in the following for $1 \leq a \leq 3$:
\begin{align} \notag
\tQQq\int_{\rth}J(\tqq p)dp &= \int_\rth   \frac{e^{- |\pv|}}{8\pi} dp =1,
\\\notag
\tQQq\int_\rth |\tqq \pv |J(\tqq p)dp
     &=\int_\rth | \pv |  \frac{e^{- |\pv|}}{8\pi} dp = 3, 
\\ \label{integralsVALUES}
\tQQq\int_\rth |\tqq \pv |^2J(\tqq p)dp
         &=  \int_\rth | \pv |^2  \frac{e^{- |\pv|}}{8\pi} dp = 12 , 
         \\ \notag
\tQQq\int_\rth (\tqq \pv^a)^2J(\tqq p)dp
              &=\int_\rth ( \pv^1)^2  \frac{e^{- |\pv|}}{8\pi} dp
              =
              \frac{1}{3} \int_\rth | \pv |^2  \frac{e^{- |\pv|}}{8\pi} dp = 4,
              \\ \notag
\tQQq\int_\rth \frac{(\tqq p^a)^2}{{|\tqq \pv |}}J(\tqq p)dp
                  &=\int_\rth \frac{( p^1)^2}{{| \pv |}}  \frac{e^{- |\pv|}}{8\pi} dp               =
              \frac{1}{3}\int_\rth | \pv |  \frac{e^{- |\pv|}}{8\pi} dp = 1,
                            \\ \notag
\tQQq\int_\rth (|\tqq \pv |-3)^2 J(\tqq p)dp
                  &=\tmQq\int_\rth (| \pv |-3)^2 \frac{e^{- |\pv|}}{8\pi}dp 
                = 3.
\end{align}
From the integrals above, the functions $\{\tTq \mathfrak{e}_{\ell}(\tqq \pv)\}_{\ell=1}^{5}$ have unit \(L^2_p\) norm. Moreover, one can easily show that the functions in $\{\tTq \mathfrak{e}_{\ell}(\tqq \pv)\}_{\ell=1}^{5}$ are orthogonal. This shows Proposition \ref{prop_orthnorm_nul}.

Together with \eqref{integralsVALUES} we also use the values of the following integrals for $1 \leq a \leq 3$:
\begin{align*}
    \int_\rth | \pv |^3  \frac{e^{- |\pv|}}{8\pi} dp &= 60,
    \\
\int_{\rth}   \frac{ (\pv^a )^2}{{\pZ}}( {\pZ}-3) \frac{e^{- |\pv|}}{8\pi}      dp
&=
\frac{1}{3}
\int_{\rth}   {\pZ}( {\pZ}-3) \frac{e^{- |\pv|}}{8\pi}dp =1,
\\
\int_{\rth}   \frac{ (\pv^a )^2}{{\pZ}}( {\pZ}-4) \frac{e^{- |\pv|}}{8\pi}      dp
&=
\frac{1}{3}
\int_{\rth}   {\pZ}( {\pZ}-4) \frac{e^{- |\pv|}}{8\pi}      dp
=0,
\\
\int_{\rth}   ( {\pZ}-3)( {\pZ}-4) \frac{e^{- |\pv|}}{8\pi}      dp
&=
\int_{\rth}   ( {\pZ}^2-7{\pZ}+12) \frac{e^{- |\pv|}}{8\pi}      dp
=3.
\end{align*}
We also compute for $1 \leq a \leq 3$ that
\begin{align*}
 \tQQq   \int_\rth \frac{(p^a)^2}{{| \pv |^2}}J(\tqq p)dp
                  &=\int_\rth \frac{( p^1)^2}{{| \pv |^2}}  \frac{e^{- |\pv|}}{8\pi} dp               =
              \frac{1}{3}\int_\rth  \frac{e^{- |\pv|}}{8\pi} dp = \frac{1}{3},
\end{align*}
and also the values of the integrals $\int_{\rth} \frac{(\pv^a )^2(\pv^b)^2}{{\pZ^2}}  \frac{e^{- |\pv|}}{8\pi} 
    dp$ for $1\leq a$, $b \leq 3$, given by
\begin{equation}\notag
\begin{split}
\int_{\rth} \frac{(\pv^a )^2(\pv^b)^2}{{\pZ^2}}  \frac{e^{- |\pv|}}{8\pi} 
    dp
&=
\frac{4}{5},
\qquad \text{for}\quad a \neq b,
\\
\int_{\rth} \frac{(\pv^a )^2(\pv^a)^2}{{\pZ^2}}  \frac{e^{- |\pv|}}{8\pi} 
    dp
&=
\frac{12}{5}.
\end{split}
\end{equation}
For the basis elements \(\left(\sqrt{\frac{3}{7}}
\left(\frac{p^a}{{\pZ}} -\tqq \pv^a\right)\tTq\sqrt{J}\right)_{1\leq a\leq 3}\) considered in \eqref{orthogonal.basis.use}, we have
\begin{multline*}
       \tQQq\int_\rth \left(\frac{p^a}{{| \pv |}}- \tqq p^a\right)^2 J(\tqq p)dp
                  =\int_\rth \left(\frac{p^a}{{| \pv |}}-  p^a\right)^2  \frac{e^{- |\pv|}}{8\pi} dp               
                  \\
=\int_\rth \left(\frac{\left(p^a\right)^2 }{{| \pv |^2}} +  \left(p^a\right)^2 -2\frac{\left(p^a\right)^2}{{| \pv |}} \right) \frac{e^{- |\pv|}}{8\pi} dp  
=\frac{7}{3}.
\end{multline*}
The computations above, and similar ones, show the functions $\{\tTq \mathfrak{e}_{\ell}(\tqq \pv)\}_{\ell=1}^{13}$ have unit \(L^2_p\) norm. One can easily show that the functions in $\{\tTq \mathfrak{e}_{\ell}(\tqq \pv)\}_{\ell=1}^{13}$ are also \(L^2_p\)  orthogonal. This proves Proposition \ref{proporthobasis}.

\subsection{Macroscopic Equations}\label{derive.macroscopic}

In this section we derive the macroscopic equations for the coefficients \(\mathcal{A}^f\), \(\mathcal{B}^f\), and \(\mathcal{C}^f\) in \eqref{defmacrocoef}, when \(f\) is a solution of the Boltzmann equation. For given function \(f\), define for $\ell\in\{1,2,\ldots, 13\}$ the quantities
\begin{align}
\begin{aligned}\label{quantmacroscopeqns}
\mathfrak{m}_{\ell} \eqdef &~  \irnp{  \micro f}{\tTq \mathfrak{e}_{\ell}(\tqq \pv)},
\\
\highG_{\ell} \eqdef &~  \irnp{\micro f}{\tTq \frac{\pv}{| \pv|}\mathfrak{e}_{\ell}(\tqq \pv)},
\\
\mathfrak{l}_{\ell} \eqdef &~  -\irnp{\linL(\micro f)}{\tTq \mathfrak{e}_{\ell}(\tqq \pv)},
\\
\mathfrak{h}_{\ell}  \eqdef &~ 
\irnp{\Gamma(f,f)}{\tTq \mathfrak{e}_{\ell}(\tqq \pv)},
 \\
 \mathfrak{u}_{\ell} \eqdef &~  -\tmq \nabla_x \cdot \highG_{\ell}
 +\mathfrak{l}_{\ell}+\mathfrak{h}_{\ell},
 \end{aligned}
\end{align}
where \(\nabla_x ~\cdot\) is the standard \(\R^3_x\) divergence of a vector field.

\begin{proposition}[Macroscopic equations]\label{macroscopic.prop}
If \(f\) solves the Boltzmann equation \eqref{rBoltz00}, and \(\mathcal{A}\), \(\mathcal{B}\), and \(\mathcal{C}\) are defined by \eqref{defmacrocoef}, then the following local conservation laws are satisfied
\begin{align} \label{local.cons}
    \tmTq\partial_t \left(\tTq \mathcal{A} \right)
     +
\frac{1}{2}  \tmq \nabla_x \cdot \mathcal{B}
&=-\tmq \nabla_x \cdot \highG_{1},
\\ \notag
 \tmTq\partial_t \left(\tTq \mathcal{B}_{j} \right) 
 +
2\tmq\partial_j \left( \mathcal{A}
-\frac{2}{\sqrt{3}}
\mathcal{C}  \right)
&=-\tmq \nabla_x \cdot \highG_{j+1}, \quad  (1 \leq j \leq 3),
\\ \notag
\tmTq\partial_t \left(\tTq \mathcal{C} \right) 
  +
\frac{1}{2\sqrt{3}}  \tmq \nabla_x \cdot \mathcal{B} 
&=-\tmq \nabla_x \cdot \highG_{5},
\end{align}
and also the macroscopic equations (for $1 \leq l, j \leq 3$),
\begin{align} \notag
\sqrt{\frac{7}{3}}\tmq\left(\partial_j\mathcal{A}
-
\sqrt{3}\partial_j\mathcal{C} \right)
&=-\tmTq  \partial_t\left( \tTq \mathfrak{m}_{j+5} \right)+\mathfrak{u}_{j+5}, 
\\  \label{macroscopic}
\tmq\frac{1}{2\sqrt{5}} \left( \partial_1 \mathcal{B}_1 - \partial_3 \mathcal{B}_3\right)
&=-\tmTq  \partial_t\left( \tTq \mathfrak{m}_{9} \right)+\mathfrak{u}_{9},
 \\ \notag
\tmq\frac{1}{6\sqrt{5}}\left( \nabla_x\cdot \mathcal{B}-3\partial_2 \mathcal{B}_2\right)
&=-\tmTq  \partial_t\left( \tTq \mathfrak{m}_{10} \right)+\mathfrak{u}_{10},
 \\ \notag
\frac{1}{\sqrt{5}}  \tmq\left(\partial_j \mathcal{B}_{l}+\partial_l \mathcal{B}_{j} \right)
&=-\tmTq  \partial_t\left( \tTq \mathfrak{m}_{l+j+9} \right)+\mathfrak{u}_{l+j+9},\quad \text{for}
\quad (l\neq j),
\end{align}
where the fixed indices above are $l,j\in\{1,2,3\}$. 
\end{proposition}

\begin{proof}[Proof of Proposition \ref{macroscopic.prop}]
Let us plug the decomposition \eqref{decomp} into \eqref{rBoltz00}, then we obtain 
\begin{multline}\label{hydro}
\left\{\partial_t+\frac{\pv^i}{|\tq \pv|} \partial_{x^i} -\frac{2\qq }{t} p^i \partial_{p^i} \right\}\macro f
= -\left\{\partial_t -\frac{2\qq }{t} p^i \partial_{p^i} \right\}\micro f
\\ 
-\frac{\pv^i}{|\tq \pv|} \partial_{x^i} (\micro f)  -\linL(\micro f)
+ \Gamma(f,f).
\end{multline}
We also recall $\linL(f)=\linL(\micro f)$ by the definition of \(\macro\) in \eqref{macrofbasis}. We have expressed here the macroscopic part $\macro f$ in terms of the microscopic part $\micro f$ and the non-linear term $\Gamma (f,f)$.

We consider the expression \eqref{macrofbasis} of $\macro f$ with respect to the basis functions \eqref{orthogonal.basis.kernel}. Also using \eqref{nullRHStp}, we obtain that the left-hand side of \eqref{hydro} can be written as
\begin{multline}\label{lhs.expand.macro}
\partial_t \left(\tTq \mathcal{A} \right)\sqrt{J(\tqq\pv)}
+\partial_t \left(\tTq \mathcal{C} \right) ({|\tqq \pv|}-3)\sqrt{J(\tqq\pv)}/\sqrt{3}
\\
+
\tTq\sum_{j=1}^3\partial_j(\mathcal{A}+\mathcal{C}({|\tqq \pv|}-3)/\sqrt{3})\frac{p^j}{{|\tq \pv|}}\sqrt{J(\tqq\pv)}
\\
+\sum_{j=1}^{3}\partial_t \left(\tTq \mathcal{B}_{j} \right) \frac{\tqq}{2} p^j\sqrt{J(\tqq\pv)}
 +
 \sum_{l,j=1}^{3}  \tTq\partial_j \mathcal{B}_{l} \frac{\tqq}{2}\frac{p^lp^j}{{|\tq \pv|}}\sqrt{J(\tqq\pv)},
\end{multline}
where $\mathcal{A}=\mathcal{A}^f$, $\mathcal{B}=\mathcal{B}^f$, $\mathcal{C}=\mathcal{C}^f$ and $\partial_j=\partial_{x_j}$.   We rewrite \eqref{lhs.expand.macro} as
\begin{multline}\label{macroscopeqn}
\tmTq\partial_t \left(\tTq \mathcal{A} \right)\tTq\sqrt{J(\tqq\pv)}
+\tmTq\partial_t \left(\tTq \mathcal{C} \right) ({|\tqq \pv|}-3)\tTq\sqrt{J(\tqq\pv)}/\sqrt{3}
\\
+
\sum_{i=1}^3\tmq\partial_i\mathcal{A}\frac{p^i}{{| \pv|}}\tTq\sqrt{J(\tqq\pv)}
+
\sum_{i=1}^3\tmq\partial_i\mathcal{C}\frac{p^i}{{| \pv|}} ({|\tqq \pv|}-3)\tTq\sqrt{J(\tqq\pv)} /\sqrt{3}
\\
+\sum_{i=1}^{3} \tmTq\partial_t \left(\tTq \mathcal{B}_i \right) \frac{\tqq}{2} p^i \tTq\sqrt{J(\tqq\pv)}
 +
 \sum_{i,j=1}^{3}  \tmq\partial_j \mathcal{B}_i \frac{\tqq}{2}\frac{p^ip^j}{{|\pv|}} \tTq\sqrt{J(\tqq\pv)} 
 \\
 \eqdef \mathcal{Z}_1+\mathcal{Z}_2+\mathcal{Z}_3.
\end{multline}
Here
\begin{multline}\notag
\mathcal{Z}_1 \eqdef 
\tmTq\partial_t \left(\tTq \mathcal{A} \right)\tTq\sqrt{J(\tqq\pv)}
+\tmTq\partial_t \left(\tTq \mathcal{C} \right) ({|\tqq \pv|}-3)\tTq\sqrt{J(\tqq\pv)}/\sqrt{3}
\\
+\tTq \mathfrak{e}_1(\tqq \pv) \frac{\tmq}{2}   \nabla_x \cdot \mathcal{B}
+
\tTq \mathfrak{e}_5(\tqq \pv)\frac{\tmq}{2\sqrt{3}}   \nabla_x \cdot \mathcal{B}.
\end{multline}
Thus
\begin{multline}\label{Z1eqn}
\mathcal{Z}_1 
 =
\tTq \mathfrak{e}_1(\tqq \pv)\left( \tmTq\partial_t \left(\tTq \mathcal{A} \right) +
\frac{1}{2}  \tmq \nabla_x \cdot \mathcal{B} \right)
\\
+
\tTq \mathfrak{e}_5(\tqq \pv)
\left(\tmTq\partial_t \left(\tTq \mathcal{C} \right)   
+
\frac{1}{2\sqrt{3}}  \tmq \nabla_x \cdot \mathcal{B}\right).
\end{multline}
The above is our main expression for \eqref{Z1eqn} and it leads directly to \eqref{local.cons}$_1$ and \eqref{local.cons}$_3$.

Next we consider $\mathcal{Z}_3$, but first we notice that
\begin{align*}
     \frac{1}{2}  \tTq \mathfrak{e}_1(\tqq \pv)
  +
  \frac{1}{2\sqrt{3}} 
\tTq \mathfrak{e}_5(\tqq \pv)
= \frac{\tqq}{6}|p| \tTq\sqrt{J(\tqq p)}
= \frac{\tqq}{2}\frac{\frac13|p|^2}{|p|} \tTq\sqrt{J(\tqq p)}. 
\end{align*}
Thus we expand
\begin{multline*}
\mathcal{Z}_3 \eqdef 
 \sum_{i,j=1}^{3}  \tmq\partial_j \mathcal{B}_i \frac{\tqq}{2}\frac{p^ip^j}{{|\pv|}} \tTq\sqrt{J(\tqq\pv)}
-
\frac{\tqq}{2}\frac{\frac13|p|^2}{|p|} \tTq\sqrt{J(\tqq p)}
\tmq\nabla_x \cdot \mathcal{B}
    \\
    =
     \sum_{\substack{i,j=1 \\ i<j}}^3
     \tmq\left(\partial_j \mathcal{B}^i+\partial_i \mathcal{B}_{j} \right) \frac{\tqq}{2}\frac{p^ip^j}{{|\pv|}} \tTq\sqrt{J(\tqq\pv)}
     \\
     +
     \sum_{j=1}^{3}\tmq\partial_j \mathcal{B}_j \frac{\tqq}{2}\frac{(p^j)^2-\frac13|p|^2}{{|\pv|}} \tTq\sqrt{J(\tqq\pv)}.
\end{multline*}
For the ortho-normal functions in \eqref{orthogonal.basis.use} we have 
\begin{align*}
       \frac{\tqq}{2}\frac{(p^1)^2-\frac13|p|^2}{{|\pv|}} \tTq\sqrt{J(\tqq\pv)}
    & = \frac{\tTq \mathfrak{e}_{9}(\tqq \pv)}{2\sqrt{5}} + \frac{\tTq \mathfrak{e}_{10}(\tqq \pv)}{6\sqrt{5}},
    \\ 
           \frac{\tqq}{2}\frac{(p^2)^2-\frac13|p|^2}{{|\pv|}} \tTq\sqrt{J(\tqq\pv)}
    & = -  \frac{\tTq \mathfrak{e}_{10}(\tqq \pv)}{3\sqrt{5}},
    \\ 
           \frac{\tqq}{2}\frac{(p^3)^2-\frac13|p|^2}{{|\pv|}} \tTq\sqrt{J(\tqq\pv)}
    & = -\frac{\tTq \mathfrak{e}_{9}(\tqq \pv)}{2\sqrt{5}} + \frac{\tTq \mathfrak{e}_{10}(\tqq \pv)}{6\sqrt{5}}.
\end{align*}
Thus we express $\mathcal{Z}_3$ in terms of the orthonormal functions as 
\begin{multline}\notag 
\mathcal{Z}_3 =
\tTq \mathfrak{e}_{9}(\tqq \pv) \left(\tmq\frac{\partial_1 \mathcal{B}_1 - \partial_3 \mathcal{B}_3}{2\sqrt{5}} \right)
+
\tTq \mathfrak{e}_{10}(\tqq \pv) \left(\tmq\frac{\nabla_x\cdot \mathcal{B}-3\partial_2 \mathcal{B}_2}{6\sqrt{5}} \right)
 \\
 +
 \frac{1}{\sqrt{5}}
 \sum_{\substack{i,j=1 \\ l<j}}^3 \tTq \mathfrak{e}_{i+j+9}(\tqq \pv)
 \tmq\left(\partial_j \mathcal{B}^i+\partial_i \mathcal{B}_{j} \right). 
\end{multline}
This is our main expression for $\mathcal{Z}_3$ and it leads directly to \eqref{macroscopic}$_2$-\eqref{macroscopic}$_4$.

Next we consider $\mathcal{Z}_2$ in \eqref{macroscopeqn} given by 
\begin{multline}\notag
\mathcal{Z}_2 \eqdef 
\sum_{i=1}^3\tmq\partial_i\mathcal{A}\frac{p^i}{{| \pv|}}\tTq\sqrt{J(\tqq\pv)}
+
\sum_{i=1}^3\tmq\partial_i\mathcal{C}\frac{p^i}{{| \pv|}} ({|\tqq \pv|}-3)\tTq\sqrt{J(\tqq\pv)} /\sqrt{3}
\\
+\sum_{i=1}^{3} \tmTq\partial_t \left(\tTq \mathcal{B}_i \right) \frac{\tqq}{2} p^i \tTq\sqrt{J(\tqq\pv)}.
\end{multline}
From \eqref{orthogonal.basis.use} we have 
\begin{align*}
    \frac{p^j}{{| \pv|}}\tTq\sqrt{J(\tqq\pv)} = \sqrt{\frac{7}{3}}\tTq \mathfrak{e}_{j+5}(\tqq \pv)
    +
2\tTq \mathfrak{e}_{j+1}(\tqq \pv),
\end{align*}
and
\begin{multline*}
       \frac{1}{\sqrt{3}} \frac{p^j}{{| \pv|}}({|\tqq \pv|}-3)\tTq\sqrt{J(\tqq\pv)} = 
    \frac{2}{\sqrt{3}}\tTq \mathfrak{e}_{j+1}(\tqq \pv)
    -
    \sqrt{3} \frac{p^j}{{| \pv|}}\tTq\sqrt{J(\tqq\pv)}
    \\
    =
    \frac{2}{\sqrt{3}}\tTq \mathfrak{e}_{j+1}(\tqq \pv)
   -
    \sqrt{7}\tTq \mathfrak{e}_{j+5}(\tqq \pv)
    -
2\sqrt{3}\tTq \mathfrak{e}_{j+1}(\tqq \pv)
    \\
    =
   -
    \sqrt{7}\tTq \mathfrak{e}_{j+5}(\tqq \pv)
    -
\frac{4}{\sqrt{3}}\tTq \mathfrak{e}_{j+1}(\tqq \pv).
\end{multline*}
Thus
\begin{multline}\label{Z2eqn}
\mathcal{Z}_2 
=
\sqrt{\frac{7}{3}}
\sum_{j=1}^3\tTq \mathfrak{e}_{j+5}(\tqq \pv)\tmq\partial_j\left(\mathcal{A}
-\sqrt{3}\mathcal{C}\right)
\\
+
\sum_{j=1}^3\tTq \mathfrak{e}_{j+1}(\tqq \pv)\left(\tmTq\partial_t \left(\tTq \mathcal{B}_{j} \right)+2\tmq\partial_j\left(\mathcal{A}
-\frac{2}{\sqrt{3}}\mathcal{C}\right) \right).
\end{multline}
This equation \eqref{Z2eqn} is our main expression for $\mathcal{Z}_2 $.

We now consider the first term on the right side of \eqref{hydro}. Multiplying by the individual basis functions in $\{\tTq \mathfrak{e}_{\ell}(\tqq \pv)\}_{\ell=1}^{13}$ in \eqref{orthogonal.basis.use}, then integrating and using \eqref{nullRHStp}, 
we get 
\begin{align*}
     \int_\rth d \pv\,\tTq \mathfrak{e}_{\ell}(\tqq \pv)~ & \left\{\partial_t -\frac{2\qq }{t} p^i\partial_{p^i} \right\}\micro f
    \\
    = &~
     \int_\rth d \pv~ \left\{\partial_t -\frac{2\qq }{t} p^i \partial_{p^i} \right\}\left( \tTq \mathfrak{e}_{\ell}(\tqq \pv)\micro f \right)
    \\
     &~
     -
          \int_\rth d \pv~ 3 \mathfrak{q} t^{3 \mathfrak{q}-1}  \mathfrak{e}_{\ell}(\tqq \pv)\micro f 
           \\
    = &~
    \left(\partial_t +\frac{3\qq }{t}  \right)
    \irnp{\micro f}{\tTq \mathfrak{e}_{\ell}(\tqq \pv)}.
\end{align*}
Thus from \eqref{timeDeriv.pq} we have 
\begin{align*}
     \int_\rth d \pv~ \tTq \mathfrak{e}_{\ell}(\tqq \pv)  \left\{\partial_t -\frac{2\qq }{t} p^i \partial_{p^i} \right\}  \micro f
    = &~
  \tmTq  \partial_t\left( \tTq
    \langle  \micro f,\tTq \mathfrak{e}_{\ell}(\tqq \pv)\rangle \right).
\end{align*}
And we also have 
\begin{align*}
     \int_\rth d \pv~  \tTq \mathfrak{e}_{\ell}(\tqq \pv) & 
     \frac{\pv}{|\tq \pv|}\cdot \nabla_x \micro f
    = 
   \tmq \nabla_x \cdot \int_\rth d \pv~ \left( \tTq \frac{\pv}{| \pv|}\mathfrak{e}_{\ell}(\tqq \pv)\micro f \right)
           \\
    = &~
    \tmq \nabla_x \cdot
    \irnp{\micro f}{\tTq \frac{\pv}{| \pv|}\mathfrak{e}_{\ell}(\tqq \pv)}.
\end{align*}
We recall now the definitions \eqref{quantmacroscopeqns} of \(\mathfrak{m}_{\ell},\) \(\highG_{\ell}\), \(\mathfrak{l}_{\ell}\), \(\mathfrak{h}_{\ell}\), and \(\mathfrak{u}_{\ell}\). We observe that $\mathfrak{m}_j = \mathfrak{l}_j = \mathfrak{h}_j = 0$ for $1  \leq j\leq  5$. Indeed, the linearised term \(\mathfrak{l}_{\ell}\) vanishes because of the self-adjointness of \(\mathcal{L}\), and the nonlinear term \(\mathfrak{h}_{\ell}\) can be shown to vanish using the definition of \(\Gamma(f,f)\) and \(p^0+q^0=p'^0+q'^0\). This enables us to obtain the local conservation laws in \eqref{local.cons} directly from the microscopic equations.

Multiplying \eqref{hydro} in the form \eqref{macroscopeqn}--\eqref{Z1eqn}--\eqref{Z2eqn} by each basis function in $\{\tTq \mathfrak{e}_{\ell}(\tqq \pv)\}_{\ell=1}^{13}$ defined in \eqref{orthogonal.basis.use} and integrating over $\threed_p$, we obtain the stated local conservation laws and macroscopic equations for the coefficients \(\mathcal{A}\), \(\mathcal{B}^i\), and \(\mathcal{C}\).
\end{proof}

The equations \eqref{local.cons} directly imply the following conservation laws for the coefficients $\mathcal{A}$, $\mathcal{B}$, $\mathcal{C}$:

\begin{corollary}[Conserved quantities]
	\label{L8.5}  From the system of conservation laws  \eqref{local.cons},  we obtain for every \(j\in\{1,2,3\}\) and all $t\geq 1$ that
\begin{align*}
   	\tTq\int_{\mathbb{T}^3} \mathcal{A}(t,x)dx=\tTq\int_{\mathbb{T}^3} \mathcal{B}_{j}(t,x)dx=\tTq\int_{\mathbb{T}^3} \mathcal{C}(t,x)dx=0. 
\end{align*}
    Equivalently, for all $t\geq 1$, we have 
    \begin{equation}\label{conservlawfourier}
\widehat{[\mathcal{A},\mathcal{B},\mathcal{C}]}(t,0)=\widehat{[\mathcal{A},\mathcal{B},\mathcal{C}]}(1,0)=0,
\end{equation}
where $[\mathcal{A},\mathcal{B},\mathcal{C}]$ denotes the vector with coefficients \(\mathcal{A}\), \(\mathcal{B}^i\), and \(\mathcal{C}\).
\end{corollary}

\begin{remark}
    Note that we have assumed in our main results that condition \eqref{conservlawfourier} holds initially at $t=1$.
\end{remark}

For later use, we  rewrite the local conservation laws in \eqref{local.cons} as
\begin{equation}\label{local.cons.alt}
\begin{split}
 \partial_t  \mathcal{A} 
     +
\frac{1}{2}  \tmq \nabla_x \cdot \mathcal{B}
+
\frac{3\qq }{t} {\mathcal{A}}
&=-\tmq \nabla_x \cdot \highG_{1},
\\ 
\partial_t \mathcal{B}_{j} 
 +
2\tmq\partial_j \left( \mathcal{A}
-\frac{2}{\sqrt{3}}
\mathcal{C}  \right)
+
\frac{3\qq }{t} {\mathcal{B}_{j}}
&=-\tmq \nabla_x \cdot \highG_{j+1}, \quad \text{for}\quad 1 \leq j \leq 3,
\\ 
\partial_t  \mathcal{C} 
  +
\frac{1}{2\sqrt{3}}  \tmq \nabla_x \cdot \mathcal{B} 
+
\frac{3\qq }{t} {\mathcal{C}}
&=-\tmq \nabla_x \cdot \highG_{5}.
\end{split}
\end{equation}
We next subtract $2\sqrt{\frac{3}{7}}$ times \eqref{macroscopic}$_1$ from \eqref{local.cons}$_2$ to obtain instead 
\begin{align}\label{macroscopic.B}
 \tmTq\partial_t \left(\tTq \mathcal{B}_{j} \right) 
 +
\tmq\frac{2}{\sqrt{3}}\partial_j \mathcal{C}  
&=
-\tmTq  \partial_t\left( \tTq \mathfrak{m}_{j+13} \right)+\mathfrak{u}_{j+13}\quad  (1 \leq j \leq 3),
\end{align}
where
\begin{align*}
\mathfrak{m}_{j+13}  \eqdef  -2\sqrt{\frac{3}{7}}\mathfrak{m}_{j+5}, \qquad 
\mathfrak{u}_{j+13} \eqdef  -\tmq \nabla_x \cdot \highG_{j+1}-2\sqrt{\frac{3}{7}}\mathfrak{u}_{j+5}.
\end{align*}
We will use this system \eqref{macroscopic}, \eqref{local.cons.alt} and \eqref{macroscopic.B} in \secref{macroscopic.eqns.sec} to obtain the crucial time-normalised coercivity estimate.

\section{Estimates for the collision operator}\label{sectcollisionop}

In this section, estimates for the linear and nonlinear parts of the collision operator are obtained. In Section \ref{sectintegralop}, we will estimate the integral operator \(\mathcal{K}\) defined in \eqref{compactK}. In Section \ref{seccoerciveestima}, we show the main coercivity estimate for the linearised operator \(\linL\). In Section \ref{secttrilnonlinear}, we prove weighted trilinear estimates for the nonlinear term \(\Gamma(f,f)\). Finally,  we show in Section \ref{subaveragequanti} estimates for averaged quantities in the momentum variables against the basis functions in \eqref{orthogonal.basis.use}. The estimates in this section hold for any $\qq \ge 0$.

\subsection{Estimates for the integral momentum operator \(\mathcal{K}\)}\label{sectintegralop}

In this section, we show that the operator \(\mathcal{K}\) can be decomposed into a compact and a small part. This is the main step to prove the coercivity for the linearised operator \(\mathcal{L}\).

\begin{proposition}[Compact-small decomposition of \texorpdfstring{\(\mathcal{K}\)}{}]
\label{noKestimate} 
Let $m\ge 0$ and $w_{m}(\tqq \pv)$ as in \eqref{weight.fcn.def}. For any small $\eta >0,$ we can decompose $\compK$ from \eqref{compactK} as 
$$
\compK = \compK_c + \compK_s,
$$
where we have that for all \(t\geq 1\),
\begin{equation}
 \tTq |\langle  w_{m}^2 \compK_s h_1, h_2\rangle |
\le 
 \eta 
\| w_{m} h_1\|_{L^2_p}
 \| w_{m}  h_2\|_{L^2_p}.
\notag
\end{equation}
for any suitably regular \(h_1\) and \(h_2\). Furthermore, for some $R=R(\eta)>0$ chosen sufficiently large, and for any $0< \spr<1$ there is a large constant $C=C_{\eta,\spr}>0$ such that for all \(t\geq 1\),
\begin{align}\label{compact.improved.est}
    \tTq |\langle w_{m}^2 \compK_c h_1, h_2\rangle |
\le
C t^{-3\spr\qq} \|\ind_{\le R} h_1 \|_{L^2_p} \| \ind_{\le R} h_2 \|_{L^2_p},
\end{align}
where $\ind_{\le R}$ is the indicator function of the ball $B_R$.
In addition $\compK_c$ is a compact operator on $L^2_p$.  
\end{proposition}

\begin{corollary}
	\label{2.10} 
	Let $m\ge 0$ and $w_{m}(\tqq \pv)$ as in \eqref{weight.fcn.def}. For any $R>0$ large enough, there is some $C_R>0$ such that for all suitably regular \(f\), and for all \(t\geq 1\),
\begin{equation}\label{estimlinweight}
	     \tTq \irnp{ w_{m}^2\linL f}{f} \gtrsim \norm{w_{m} f}^2_{L^2_p(\threed)}-C_R\norm{\ind_{\le R}f}^2_{L^2_p}. 
	 \end{equation}
\end{corollary}

Corollary \ref{2.10} follows directly from Proposition \ref{noKestimate}. The estimate \eqref{estimlinweight} is similar to \cite[Section 2.6, Lemma 2.9]{RelBolNoncut2020} for the relativistic Boltzmann equation on \(\T^3_x\) in the massive case.

\begin{proof}[Proof of Proposition \ref{noKestimate}]   
We write the linear operators $\compK_1$ and $\compK_2$ from \eqref{compactK} in terms of a kernel as
\begin{align}\label{hilbert.sch.K}
    \compK_j h =     t^{3\qq}\int_{\mathbb{R}^3} ~  dq ~
k_{j}(\tqq p,\tqq q) ~ h(\qv),\qquad \text{for}~~j \in \{1, 2\}.
\end{align}
Here, from \eqref{moller} and \eqref{compactK} for $k_1=k_1(\tqq p,\tqq q)$ we have
\begin{equation}
k_1 \eqdef 
 \frac{1}{2}\left(1-\frac{ p \cdot q}{|\pv||\qv|} \right)e^{-\tqq (\pZ+\qZ)/2}
 =
 \frac{1}{4}\left| \frac{\pv}{\pZ} -\frac{\qv}{\qZ}\right|^2 e^{-\tqq (\pZ+\qZ)/2},
\label{k1def}
\end{equation}
and $k_2(\tqq p,\tqq q)$ is defined in \eqref{k2def} below.  We first prove the estimates for $\compK_1$.

Given $R \ge 1$, we set a positive and smooth cut-off function $\phi_R = \phi_R(\tqq p,\tqq q)$ such that 
\begin{equation}
\label{smmothPr}
\begin{split}
\phi_R = 1, 
\quad
\text{if} 
\quad
\tqq |p|+\tqq |q| \le R/2, \qquad \left| \phi_R \right| \le 1,
\\
\text{supp} (\phi_R)
\subset
\left\{(p,q) \left| ~
\tqq |p|+\tqq |q| \le R\right.
\right\}.
\end{split}
\end{equation}
We use this cut-off function with several different $R$ values in the rest of the proof.  

First we estimate $k_1(\tqq p,\tqq q)$ in \eqref{compactK}
and \eqref{hilbert.sch.K}.  For large $R_1>0$ we split 
\begin{align*}
k_1(\tqq p,\tqq q) &= \phi_{R_1} k_1(\tqq p,\tqq q) + (1-\phi_{R_1})k_1(\tqq p,\tqq q) \\
&\eqdef k_{1c}(\tqq p,\tqq q)+k_{1s}(\tqq p,\tqq q),
\end{align*}
and then we similarly decompose 
\begin{equation}\label{K.one.decomp}
\compK_1=\compK_{1c}+\compK_{1s}.
\end{equation}
Recall now, for $m\ge 0$, the weight $w_{m}(\tqq \pv)$ in \eqref{weight.fcn.def}. From \eqref{k1def} and \eqref{smmothPr} there is a uniform constant such that
\begin{align*}
   w_{m}^2(\tqq \pv) \left| k_{1s}(\tqq p,\tqq q)\right| 
    \lesssim  e^{-R_1/4} e^{-\tqq (\pZ+\qZ)/4}.
\end{align*}
Thus we estimate $\compK_{1s}$ using the Cauchy--Schwarz inequality as
\begin{multline} \label{CauchySchwarz.ineq}
    \left|
\langle  w_{m}^2 \compK_{1s}(h_1), h_2\rangle 
\right|
\le 
\tTq\int_{\mathbb{R}^3}dq ~  \int_{\mathbb{R}^3}dp ~ w_{m}^2(\tqq \pv) ~
\left| k_{1s}(\tqq p,\tqq q) \right| 
\left|
 h_1(q)
h_2(p)
\right|
\\ 
\le 
\tTq\left( \int_{\mathbb{R}^3}dq ~ \left|
 h_1(q)
\right|^2 \int_{\mathbb{R}^3}dp ~ w_{m}^2(\tqq \pv) ~
\left| k_{1s}(\tqq p,\tqq q) \right| 
\right)^{\frac12}
\\ 
\times
\left( \int_{\mathbb{R}^3}dp ~ \left|
 h_2(p)
\right|^2 \int_{\mathbb{R}^3}dq ~ w_{m}^2(\tqq \pv) ~
\left| k_{1s}(\tqq p,\tqq q) \right| 
\right)^{\frac12}
\\ 
\lesssim   
\tmTq e^{-R_1/4} \|  h_1\|_{L^2_p}
 \|   h_2\|_{L^2_p}.
\end{multline}
By choosing $R_1>0$ sufficiently large, this gives the desired estimate for $\compK_{1s}$.

For $\compK_{1c}$, from \eqref{k1def} and \eqref{smmothPr}, there is a large constant such that 
\begin{align*}
    w_{m}^2(\tqq \pv) \left| k_{1c}(\tqq p,\tqq q) \right| 
\lesssim e^{-\tqq (\pZ+\qZ)/2} ~ \ind_{\tqq|p|\le R} ~ \ind_{\tqq|q|\le R}.
\end{align*}
Now in addition, for some $0< \spr < 1$, we consider the sets 
\begin{align}\label{time.split.cmpt}
  A \eqdef \big\{   |p|+ |q| \geq t^{-(1+\spr)\qq}\big\}, \qquad B \eqdef A^c = \big\{   |p|+ |q| < t^{-(1+\spr)\qq}\big\},
\end{align}
and then split
\begin{align*}
    k_{1c}(\tqq p,\tqq q)&=k_{1c}(\tqq p,\tqq q) \ind_{A}+k_{1c}(\tqq p,\tqq q) \ind_{B}\\
    &\eqdef k_{1cA}(\tqq p,\tqq q)+k_{1cB}(\tqq p,\tqq q),
\end{align*}
and we similarly decompose \[\mathcal{K}_{1c}=\mathcal{K}_{1cA}+\mathcal{K}_{1cB}.\]
Then we have 
\begin{align*}
   w_{m}^2  k_{1cA}(\tqq p,\tqq q)  
\lesssim e^{-t^{(1-\spr)\qq} /4}e^{-\tqq (\pZ+\qZ)/4} ~ \ind_{\tqq|p|\le R} ~ \ind_{\tqq|q|\le R}.
\end{align*}
Thus for $\compK_{1cA}$ with kernel $k_{1cA}(\tqq p,\tqq q)$, using Cauchy--Schwarz as in \eqref{CauchySchwarz.ineq},
\begin{align*}
    |\langle  w_{m}^2 \compK_{1cA} h_1, h_2\rangle |
\lesssim
\tmTq e^{-t^{(1-\spr)\qq}/4}\|\ind_{\le R} h_1 \|_{L^2_p} \| \ind_{\le R} h_2 \|_{L^2_p}.
\end{align*}
On $B$ instead, we use that 
\begin{align*}
\max \{ |p|, |q|\} \leq  |p|+ |q| < t^{-(1+\spr)\qq}.
\end{align*}
Hence, for $k_{1cB}(\tqq p,\tqq q)$ we have 
\begin{align*}
w_{m}^2(\tqq \pv)\int_{\mathbb{R}^3}dq ~
 k_{1cB}(\tqq p,\tqq q)
\lesssim \ind_{|p|\le R} \int_0^{t^{-(1+\spr)\qq}} \rho^2 d\rho \lesssim \ind_{|p|\le R} t^{-3\qq} t^{-3\spr\qq}.
\end{align*}
The same estimate holds for $\int_{\mathbb{R}^3}dp ~ w_{m}^2(\tqq \pv)
k_{1cB}(\tqq p,\tqq q)$.  Thus for $\compK_{1cB}$ with kernel $k_{1cB}(\tqq p,\tqq q)$,  similarly to \eqref{CauchySchwarz.ineq}, we obtain
\begin{align*}
    |\langle w_{m}^2 \compK_{1cB} h_1, h_2\rangle |
\lesssim
\tmTq  t^{-3\spr\qq}\|\ind_{\le R} h_1 \|_{L^2_p} \| \ind_{\le R} h_2 \|_{L^2_p}.
\end{align*}
This establishes the desired estimate for $\compK_{1c}$. Then $\compK_{1c}$ is a compact operator because $k_{1c}(\tqq p,\tqq q) \in L^2(\threed_p\times \threed_q)$.

We now consider $\compK_{2}$ as in \eqref{compactK} and in \eqref{hilbert.sch.K}.   We first explain how to obtain the form \eqref{hilbert.sch.K} for the kernel. Due to the diagonal structure of the FLRW metric, we can directly follow the special relativistic calculation in \cite[Appendix, starting on page 588]{MR2728733} to obtain, for some fixed constant $c_1>0$, that we have 
\begin{align}\notag 
    k_2(\tqq p,\tqq q) = &~
\frac {c_1 s^{3/2}}{\mrel{\pZ}{\qZ}\tqq}
\kerU_1(\tqq p,\tqq q)
 \exp\left( -\kerU_2(\tqq p,\tqq q)\right),
\end{align}
where 
\begin{equation}\notag 
\kerU_2(\tqq p, \tqq q)  \eqdef \frac{ \sqrt{s}}{2\mrel}\tqq|p-q|= \frac{\tqq}{2}|p-q|,    
\end{equation}
and
\begin{align}\notag  
    \kerU_1(\tqq p,\tqq q)\eqdef &~ \frac{1}{\kerU_2}+\frac{\tqq}{2}\frac{{\pZ}+{\qZ}}{\kerU_2^2}+\frac{\tqq}{2}\frac{{\pZ}+{\qZ}}{\kerU_2^3}
    \\
    =  &~
    \frac{2}{\tqq }\left(\frac{1}{|p-q|}+\frac{{\pZ}+{\qZ}}{|p-q|^2}
+
\frac{2}{\tqq}\frac{{\pZ}+{\qZ}}{ |p-q|^3}\right).\nonumber
\end{align}
Notice from \eqref{sDEFINITIONFLRW} that we have
\begin{align}\label{difference.pq.ktwo}
    \frac {s^{3/2}}{\mrel{\pZ}{\qZ}\tqq}
        =
    \frac{2\big(\pZ\qZ - p\cdot q\big)}{{\pZ}{\qZ}}
    =
\left| \frac{\pv}{\pZ} -\frac{\qv}{\qZ}\right|^2.
\end{align}
Using additionally \eqref{difference.pq.ktwo}, 
we  decompose
\begin{align*}
    k_2(\tqq p,\tqq q)=k_{21}(\tqq p,\tqq q)+k_{22}(\tqq p ,\tqq q)+k_{23}(\tqq p,\tqq q),
\end{align*}
where
\begin{align}\notag
k_{21}(\tqq p,\tqq q)= &~
2c_1
\frac{\left| \frac{\pv}{\pZ} -\frac{\qv}{\qZ}\right|^2}{\tqq|p-q|}
 e^{- \tqq|p-q|/2},
 \\ \label{k2def}
 k_{22}(\tqq p,\tqq q)= &~
2c_1
\frac{\tqq\left(|\pv|+|\qv|\right)\left| \frac{\pv}{\pZ} -\frac{\qv}{\qZ}\right|^2}{\left(\tqq|p-q|\right)^2}
 e^{- \tqq|p-q|/2},
 \\ \notag
  k_{23}(\tqq p,\tqq q)= &~
4c_1
\frac{\tqq\left(|\pv|+|\qv|\right)\left| \frac{\pv}{\pZ} -\frac{\qv}{\qZ}\right|^2}{\left(\tqq|p-q|\right)^3}
 e^{- \tqq|p-q|/2}.
\end{align}
 To simplify no tation, we suppress the dependence on $\tqq$ in the rest of this proof. In particular, we write $k_2(p,q)=k_2(\tqq p, \tqq q)$ and similarly for $k_{21}$, $k_{22}$ and $k_{23}$.

The following inequality shows that $k_2(p,q)$ as in \eqref{k2def} is integrable.  We write
\begin{align*}
    \frac{\pv}{\pZ} -\frac{\qv}{\qZ}
    =
    \frac{\pv-\qv}{\pZ} +\qv\left(\frac{1}{\pZ}-\frac{1}{\qZ}\right).
\end{align*}
Thus, by the triangle inequality
\begin{align*}
   \left| \frac{\pv}{\pZ} -\frac{\qv}{\qZ}\right|
    \leq
    \frac{|\pv-\qv|}{\pZ} +|\qv|\left|\frac{1}{\pZ}-\frac{1}{\qZ}\right|.
\end{align*}
As a result, 
\begin{align*}
  \vnrm{\pv} \left| \frac{\pv}{\pZ} -\frac{\qv}{\qZ}\right|
    \leq
    \vnrm{\pv-\qv}
 +\vnrm{\vnrm{\qv}-\vnrm{\pv}}
     \leq
    2\vnrm{\pv-\qv}.
\end{align*}
After  using symmetry, we conclude that 
\begin{align}\label{unit.vector.ineq}
  \max\{\pZ, \qZ\} \left| \frac{\pv}{\pZ} -\frac{\qv}{\qZ}\right|
    \leq
    2|\pv-\qv|.
\end{align}
From \eqref{k2def} and \eqref{unit.vector.ineq} there is a uniform constant such that
\begin{align}\notag
    k_{21}(p,q)  +k_{22}(p,q) 
    \lesssim
    \frac{e^{- \tqq|p-q|/2}}{\tqq|p-q|},
    \qquad 
    k_{23}(p,q) 
    \lesssim
    \frac{e^{- \tqq|p-q|/2}}{(\tqq|p-q|)^2}.
\end{align}
Then \eqref{kernel.bounds.hs} shows that $k_2(p,q)$ as in \eqref{k2def} is  integrable in $\mathbb{R}^3_q$ uniformly in time for any fixed $\pv \in \threed$.  For the weight $w_{m}^2(\tqq \pv)$ in \eqref{weight.fcn.def} with $m\ge 0$ we use $|p| \leq |p-q| + |q|$ to see that 
\begin{align*}
    w_{m}^2(\tqq \pv) \lesssim w_{m}^2(\tqq \qv)+w_{m}^2(\tqq (\pv-\qv)).
\end{align*}
We thus conclude that 
\begin{equation}\label{kernel.bounds.hs}
\begin{split}
    w_{m}^2(\tqq \pv)\left(k_{21}(p,q)  + k_{22}(p,q) \right)
   & \lesssim
    \min\big\{w_{m}^2(\tqq \pv), w_{m}^2(\tqq \qv)\big\}\frac{e^{- \tqq|p-q|/4}}{\tqq|p-q|},
\\
   w_{m}^2(\tqq \pv) k_{23}(p,q) 
   & \lesssim
   \min\big\{w_{m}^2(\tqq \pv), w_{m}^2(\tqq \qv)\big\} \frac{e^{- \tqq|p-q|/4}}{(\tqq|p-q|)^2}.
\end{split}
\end{equation}
These will be our main estimates for the kernels $k_{2j}(p,q)$ above.

To proceed, we will employ a splitting to cut out the singularity, analogous to \cite[Equation (3.1)]{MR2728733} and also to \cite{MR2366140}. Fix a small $\smep\in(0,1)$.  We choose a smooth nonnegative cut-off function $\chi (r)$ satisfying
\begin{equation}
\chi=\chi (\tqq|\pv - \qv|)=
\left\{
\begin{array}{cl}
 0 & {\rm if } ~~  \tqq|\pv - \qv| \ge 2\smep,
 \\
  1 &
  {\rm if } ~~ \tqq|\pv - \qv| \leq \smep.
\end{array}
\right.
\label{cut}
\end{equation}
Now with \eqref{k2def} and \eqref{cut} for $j \in \{1, 2, 3\}$ we define
\begin{align}\label{kCUT}
    k_{2j}^{1-\chi}(p,q) \eqdef  \left( 1 - \chi \right)k_{2j}(p,q), 
    \qquad 
    k_{2j}^{\chi}(p,q)\eqdef   \chi k_{2j}(p,q)=k_{2js1}(p,q),
\end{align}
and will use the analogous splitting
$$
\compK_{2j}
\eqdef
\compK_{2j}^{1-\chi}
+
\compK_{2j}^{\chi}, \quad j \in \{1, 2, 3\}.
$$
We also denote $\compK_{2j}^{\chi}=\compK_{2js1}$ with kernel $k_{2js1}(p,q)$.

Next we consider $\compK_{2js1}$. To prove the desired estimate for the terms $\compK_{2js1}$, it suffices to prove for $j \in \{1, 2, 3\}$ the following uniform bounds
\begin{align}\label{integral.est.needed}
  \int_{\threed} dq~ k_{2js1}(p,q)\lesssim 
  \tmQq  \smep,
 \quad\int_{\threed} d\pv~ w_{m}^2(\tqq \pv)k_{2js1}(p,q)  \lesssim 
  w_{m}^2(\tqq \qv)\tmQq  \smep. 
\end{align}
We will only prove \eqref{integral.est.needed} for the $d\qv$ integral because, after using \eqref{kernel.bounds.hs} with $m>0$, the other proof is identical.  
From \eqref{kernel.bounds.hs} with $m=0$ we have 
\begin{align*}
    \int_{\threed} d\qv~ k_{2js1}(p,q)
    \lesssim \tmQq\int_0^{2\smep} d\rho ~ e^{- \rho/4} \rho^{b(j)}
    \lesssim \tmQq \smep^{b(j)+1}.
\end{align*}
Here $b(1) =b(2) = 1$ and  $b(3) = 0$ as in \eqref{kernel.bounds.hs}.  In particular the estimate \eqref{integral.est.needed} holds.  Then proceeding as in \eqref{CauchySchwarz.ineq} we have
\begin{align}\notag
\left|
\langle  w_{m}^2 \compK_{2js1}(h_1), h_2\rangle 
\right|
\lesssim   
\tmTq \smep \| w_{m} h_1\|_{L^2_p}
 \|  w_{m} h_2\|_{L^2_p}.
\end{align}
This is the desired estimate for $\compK_{2j}^{\chi}=\compK_{2js1}$.

We next consider $\compK_{2j}^{1-\chi}$ with kernel $k_{2j}^{1-\chi}(p,q)$ as in \eqref{kCUT}.  For this term using \eqref{smmothPr} with a different $R_2>0$ for $j \in \{1, 2, 3\}$ we decompose 
\begin{equation}\notag
k_{2j}^{1-\chi}(p,q) = \phi_{R_2} k_{2j}^{1-\chi}(p,q) + (1-\phi_{R_2})k_{2j}^{1-\chi}(p,q) \eqdef k_{2jc}(p,q)+k_{2js2}(p,q).
\end{equation}
We will first estimate $\compK_{2jc}$ with kernel $k_{2jc}$.  From \eqref{k2def}, \eqref{kernel.bounds.hs} and \eqref{cut}, there is a uniform constant such that 
\begin{align*}
    w_{m}^2(\tqq \pv)  k_{2jc}(p,q) 
\lesssim \smep^{-2} e^{- \tqq|p-q|/2} ~ \ind_{|\tqq p|\le R_2} ~ \ind_{|\tqq q|\le R_2}.
\end{align*}
Now, similarly to \eqref{time.split.cmpt}, for any $0< \spr < 1$ we consider the regions
\begin{align}\notag
  A \eqdef \{   |p-q| \geq t^{-(1+\spr)\qq}\}, \qquad B \eqdef A^c = \{   |p-q| < t^{-(1+\spr)\qq}\},
\end{align}
and then split \[k_{2jc}= k_{2jc}\ind_A+k_{2jc}\ind_B\eqdef k_{2jcA}+k_{2jcB},\]
and we use the analogous decomposition \[\mathcal{K}_{2jc}=\mathcal{K}_{2jcA}+\mathcal{K}_{2jcB},\qquad \text{for}\quad j\in\{1,2,3\}.\]
On the set $A$ we have 
\begin{align*}
    w_{m}^2(\tqq \pv)  k_{2jcA} 
\lesssim \smep^{-2} e^{- t^{(1-\spr)\qq}/4} e^{- \tqq|p-q|/4} ~ \ind_{\tqq|p|\le R_2} ~ \ind_{\tqq|q|\le R_2}.
\end{align*}
Thus for $\compK_{2jcA}$ with kernel $k_{2jcA}$, as in \eqref{CauchySchwarz.ineq} we have
\begin{align*}
    |\langle w_{m}^2 \compK_{2jcA} h_1, h_2\rangle |
\lesssim
\tmTq e^{-t^{(1-\spr)\qq}/4}\|\ind_{\le R} h_1 \|_{L^2_p} \| \ind_{\le R} h_2 \|_{L^2_p}.
\end{align*}
On $B$ instead, we estimate the contribution of $k_{2jcB}$ by 
\begin{align*}
\int_{\mathbb{R}^3}dq ~
w_{m}^2(\tqq \pv) k_{2jcB}(\tqq p,\tqq q)
\lesssim \ind_{|p|\le R} \int_0^{t^{-(1+\spr)\qq}} \rho^2 d\rho  
\lesssim \ind_{|p|\le R} t^{-3\qq} t^{-3\spr\qq}.
\end{align*}
 Thus for $\compK_{2jcB}$ with kernel $k_{2jcB}$, similarly to \eqref{CauchySchwarz.ineq}, we obtain
\begin{align*}
    |\langle w_{m}^2 \compK_{2jcB} h_1, h_2\rangle |
\lesssim
\tmTq  t^{-3\spr\qq}\|\ind_{\le R} h_1 \|_{L^2_p} \| \ind_{\le R} h_2 \|_{L^2_p}.
\end{align*}
This establishes the desired estimate for $\compK_{2jcB}$ and again $\compK_{2jcB}$ is a compact operator because  $k_{2jcB} \in L^2(\threed_p\times \threed_q)$.

It remains to estimate $\compK_{2js2}$ with kernel $k_{2js2}$.  However from \eqref{smmothPr} we have $ |\tqq p|+ | \tqq q| > R_2/2$, thus either $ | \tqq \pv|$ or $ | \tqq \qv|$ is strictly larger than $R_2/4$.  Suppose that either $ |\tqq \qv| \leq R_2^{\frac12}$ or $ | \tqq \pv| \leq R_2^{\frac12}$.  If for instance $ |\tqq \qv| \leq R_2^{\frac12}$, then 
\begin{align*}
 1=   \frac{ |\tqq p|+  |\tqq q|}{ |\tqq p|+ |\tqq q|} \leq \frac{\tqq |p-q|+2 |\tqq q|}{ |\tqq p|+ |\tqq q|} 
    \lesssim \frac{\tqq |p-q|+2 R_2^{\frac12}}{R_2}. 
\end{align*}
The same estimate holds if $ |\tqq \pv| \leq R_2^{\frac12}$.  We plug the bounds above into \eqref{kernel.bounds.hs}, for $\smep>0$ fixed and $R_2>0$ large,  to obtain
\begin{equation}\label{kernel.bounds.qbd}
\begin{split}
    w_{m}^2(\tqq \pv)\left(k_{21s2}  + k_{22s2} \right)
   & \lesssim
    \min\big\{w_{m}^2(\tqq \pv), w_{m}^2(\tqq \qv)\big\}\frac{e^{- \tqq|p-q|/8}}{R_2^{\frac12} \smep},
\\
   w_{m}^2(\tqq \pv) k_{23s2}
   & \lesssim
   \min\big\{w_{m}^2(\tqq \pv), w_{m}^2(\tqq \qv)\big\} \frac{e^{- \tqq|p-q|/8}}{R_2^{\frac12} \smep^2}.
\end{split}
\end{equation}
This is our main estimate when either $ |\tqq\qv| \leq R_2^{\frac12}$ or $ |\tqq\pv| \leq R_2^{\frac12}$.

Now, alternatively we suppose that both $ |\tqq\pv| > R_2^{\frac12}$ and $|\tqq \qv| > R_2^{\frac12}$. Notice that we can instead express \eqref{k2def} using \eqref{difference.pq.ktwo} as
\begin{align}\notag
k_{21}(p,q)= &~
\frac{2c_1}{\tqq|p-q|}\frac{\tqq s}{{\pZ}\tqq{\qZ}\tqq}
 e^{- \tqq|p-q|/2},
 \\ \label{k2def.Second}
 k_{22}(p,q)= &~
2c_1
\frac{\tqq s}{\left(\tqq|p-q|\right)^2}
\left(\frac{1}{{\pZ}}+\frac{1}{{\qZ}} \right)\frac{1}{\tqq}
 e^{- \tqq|p-q|/2},
 \\ \notag
  k_{23}(p,q)= &~
4c_1
\frac{\tqq s}{\left(\tqq|p-q|\right)^3}
\left(\frac{1}{{\pZ}}+\frac{1}{{\qZ}} \right)\frac{1}{\tqq}
 e^{- \tqq|p-q|/2}.
\end{align}
By using \eqref{mrel.upper.est} and \eqref{cut} in the expressions \eqref{k2def.Second} above, for $\smep>0$ fixed and $R_2>0$ large, for $j \in \{1, 2, 3\}$ we obtain
\begin{align}\label{kernel.bounds.last.small}
  w_{m}^2(\tqq \pv)  k_{2js2}(p,q)  
    \lesssim  
    \min\big\{w_{m}^2(\tqq \pv), w_{m}^2(\tqq \qv)\big\}
    R_2^{-\frac12} \smep^{-1} {e^{- \tqq|p-q|/8}}.
\end{align}
Again proceeding as in \eqref{CauchySchwarz.ineq}, using \eqref{kernel.bounds.qbd} and \eqref{kernel.bounds.last.small}, we have
\begin{align}\notag
\left|
 \irnp{w_{m}^2 \compK_{2js2}(h_1)}{h_2} 
\right|
\lesssim   
\tmTq R_2^{-\frac12} \smep^{-1} \| w_{m} h_1\|_{L^2_p}
 \|  w_{m} h_2\|_{L^2_p}.
\end{align}
First choose $\smep>0$ small, then choose $R_2>0$ large to complete the proof.
\end{proof}

\subsection{Coercivity of the linearised operator}\label{seccoerciveestima}

Let us study the basic functional analytic properties of the linearised operator $\linL$.

\begin{proposition}
    The linear operator $\linL$ from \eqref{L0} is self-adjoint, that is, $\langle \linL f,h\rangle =\langle  f,\linL h\rangle$ for all \(f\), \(h\in L^2_p\). Moreover $\linL$ is non-negative, in other words $\langle\linL h,h\rangle \ge 0$ for all \(h\in L^2_p\). Moreover $\linL h = 0$ if and only if $h = {\bf P} h$ for ${\bf P}$ from \eqref{macrofbasis}.  
\end{proposition}

    The properties above are standard, and their proofs follow directly from \cite{GL1996}. We now derive the main coercivity estimate for the linearised problem.

\begin{proposition}[Coercivity estimate]\label{lowerN}  
There exists a uniform constant
$\delta_0>0$ such that for all \(h\in L^2_p\), and for all \(t\geq 1\),
\begin{equation}\label{coercestprop2}
  \tTq\langle \linL h, h \rangle
\geq
\delta_0 
\| \micro  h \|_{L^2_p}^2.
\end{equation}
\end{proposition}

\begin{proof}[Proof of Proposition \ref{lowerN}]
Suppose that the coercive estimate is false.  Then there is a time \(t\geq 1\) and a sequence of functions $h^n(\pv)$ satisfying that
$
\macro h^n = 0,
$ 
$\nu_0\| h^n \|_{L^2_p}=1$,
and
$$
\tTq\langle {\linL} h^n, h^n \rangle
=
\nu_0\| h^n \|_{L^2_p}^2
- 
\tTq\langle {\compK} h^n, h^n \rangle
\le \frac{1}{n}.
$$ We recall that \(\nu_0>0\) is the uniform constant obtained earlier in Proposition \ref{LEM:nuASYMP}.
Since $\nu_0\| h^n \|_{L^2_p}=1$, then 
 $ \{ h^n \} $ is weakly compact in ${L^2_p}$ with limit point $h^0$.  
By weak lower-semi continuity $\nu_0\| h^0 \|_{L^2_p} \le 1$.  Furthermore,
$$
\tTq\langle {\linL} h^n, h^n \rangle
=
1
- 
\tTq\langle {\compK} h^n, h^n \rangle.
$$ 
From Proposition \ref{noKestimate}, we can split $\compK = \compK_c + \compK_s$ where $\tTq\langle {\compK_s} h^n, h^n \rangle$ can be chosen arbitrarily small and $\compK_c$ is a compact operator.  More precisely, for any small $\eta>0$ we can choose $\compK_s$ as in Proposition \ref{noKestimate} so that we have 
\begin{equation}
\tTq |\langle   \compK_s h^n, h^n\rangle |
\le 
 \eta 
\| h^n\|_{L^2_p}^2 = \eta.
\notag
\end{equation}
Then the remaining part $\compK_c$ has an integrable kernel which makes it a compact operator.  Thus in particular
$$
\lim_{n\to\infty}
\langle {\compK}_c h^n, h^n \rangle
=
\langle {\compK}_c h^0, h^0 \rangle.
$$
Now fix $\epsilon>0$.  Then choose $\eta = \frac{\epsilon}{4}$.  Then 
\begin{equation}
\tTq |\langle   \compK_s h^n, h^n\rangle |
\le 
\frac{\epsilon}{4}, 
\quad 
\tTq |\langle   \compK_s h^0, h^0\rangle |
\le 
\frac{\epsilon}{4}.
\notag
\end{equation}
Then for the resulting  ${\compK}_c$ from Proposition \ref{noKestimate} (which depends upon $\eta>0$) we can choose $N \in \mathbb{N}$ large enough so that 
\begin{align*}
   \tTq \left|\langle {\compK}_c h^n, h^n \rangle
-
\langle {\compK}_c h^0, h^0 \rangle \right| 
< \frac{\epsilon}{2}, \quad \forall n \ge N.
\end{align*}
Thus 
\begin{align*}
   \tTq \left|\langle {\compK} h^n, h^n \rangle
-
\langle {\compK} h^0, h^0 \rangle \right| 
< \frac{\epsilon}{2}+\frac{\epsilon}{4}+\frac{\epsilon}{4}=\epsilon, \quad \forall n \ge N.
\end{align*}
Combining these statements, we conclude  that
$$
\lim_{n\to\infty}
\tTq\langle {\compK} h^n, h^n \rangle
=
\tTq\langle {\compK} h^0, h^0 \rangle.
$$
In particular, we have that 
$\tTq\langle {\compK} h^0, h^0 \rangle= 1$,
or equivalently
$$
\tTq\langle {\linL} h^0, h^0 \rangle
=
\nu_0\| h^0 \|_{L^2_p}^2 
- 
1.
$$
Since $\linL\ge 0$, we conclude that $\nu_0\| h^0 \|_{L^2_p}^2  
=
1$.  Thus $\linL h^0= 0$ which implies $h^0 = {\bf P} h^0$.  On the other hand,
 $h^n = \{{\bf I- P}\} h^n$ so that weak convergence implies $h^0 = \{{\bf I- P}\} h^0$.
Thus $h^0 =0$, which is a contradiction to $\nu_0\| h^0 \|_{L^2_p}^2 = 1$. 
\end{proof}

\begin{remark}
    The proof of Proposition \ref{lowerN} shows by contradiction the existence of a constant \(\delta_0>0\) such that \eqref{coercestprop2} is satisfied. Alternatively, one could attempt to provide a quantitative proof of the existence of \(\delta_0>0\) as in previous works in the non-expanding case \cite{MR2322149, BarangerMouhot2005, mouhot2006explicit}.
\end{remark}

The proof above follows the strategy of \cite[Lemma 6.9]{RelBolNoncut2020} for the relativistic Boltzmann equation on \(\T^3_x\) in the massive case. See also the analysis in \cite{MR1211782}. 

\subsection{Trilinear estimate for the nonlinear term}\label{secttrilnonlinear}

We now prove $L^2_p$-trilinear estimates for the non-linear term \(\Gamma(f,f)\) in \eqref{gamma0}.

\begin{lemma}
	\label{thm1} Let \(m\geq 0\) and recall the weight \(w_{m}^2 = 1+|\tqq \pv|^m\). Then, for all \(f\), \(h\), and \(\eta\), the following uniform estimate holds,
	 $$
	 |\langle w_{m}^2 \Gamma(f,h),\eta\rangle| \lesssim \norm{w_{m}f}_{L^2_p(\threed)}\norm{w_{m}h}_{L^2_p(\threed)}\norm{w_{m}\eta}_{L^2_p(\threed)}.
	$$
\end{lemma}
	
\begin{proof}[Proof of Lemma \ref{thm1}]
From \eqref{gamma0}, we have 
\begin{align*}
  \langle w_{m}^2\Gamma(f,h),\eta\rangle   
 = &~
\tTq \int_{\mathbb{R}^3} ~  d\pv ~w_{m}^2(\tqq \pv)\eta (\pv)\int_{\mathbb{R}^3} ~  d\qv 
\int_{\mathbb{S}^{2}} ~ d\omega
~ v_{\Emptyset} ~   \sqrt{J(\tqq q)}
f(p^{\prime })h(q^{\prime})
 \\
   &~
-\tTq \int_{\mathbb{R}^3} ~  d\pv ~w_{m}^2(\tqq \pv)\eta (\pv) f(\pv)\int_{\mathbb{R}^3} ~  d\qv 
\int_{\mathbb{S}^{2}} ~ d\omega
~ v_{\Emptyset} ~   \sqrt{J(\tqq q)}
h(q)
 \\
  \eqdef &~\Gamma_1 - \Gamma_2.
\end{align*}
We will estimate separately $\Gamma_1$ and $\Gamma_2$.

First, we have the following uniform estimate
\begin{equation*}
    \int_{\mathbb{R}^3} ~  d\qv 
   \sqrt{J(\tqq q)}
 \left| h(q)\right|
 \lesssim
 \norm{J^{\frac14} h}_{L^2_p(\threed)}  \left(    \int_{\mathbb{R}^3} ~  d\qv 
   \sqrt{J(\tqq q)} \right)^{\frac12}
   \lesssim
   \tmTq
    \norm{J^{\frac14}h}_{L^2_p(\threed)}.
\end{equation*}
Thus, we estimate $\Gamma_2$ using the Cauchy--Schwarz inequality as
\begin{align*}
\left| \Gamma_2 \right|  
 \lesssim &~
\tTq   \int_{\mathbb{R}^3} ~  d\pv ~ w_{m}^2(\tqq \pv)\left| \eta (\pv) f(p)\right|\int_{\mathbb{R}^3} ~  d\qv 
   \sqrt{J(\tqq q)}
 \left| h(q)\right|
 \\
  \lesssim &~
  \norm{J^{\frac14}h}_{L^2_p(\threed)}
\norm{w_{m}f}_{L^2_p(\threed)} \norm{w_{m}\eta}_{L^2_p(\threed)}.
\end{align*}
This is the desired estimate for $\Gamma_2$.

Next using \eqref{collisionalCONSERVATION}, we have the uniform estimate
\begin{equation*}
    w_{m}(\tqq \pv) \lesssim w_{m}(\tqq \pv')+w_{m}(\tqq \qv')  \lesssim w_{m}(\tqq \pv')w_{m}(\tqq \qv').
\end{equation*}
Next we estimate $\Gamma_1$ using again the Cauchy--Schwarz inequality as
\begin{align*}
\left| \Gamma_1 \right|  
 \lesssim &~
\tTq  \int_{\mathbb{R}^3} ~  d\pv ~ w_{m}(\tqq \pv)\left| \eta (\pv) \right|\int_{\mathbb{R}^3} ~  d\qv 
~ v_{\Emptyset} ~
   \sqrt{J(\tqq q)}
 \left| w_{m}(\tqq \pv)h(q^\prime)f(p^\prime)\right|
 \\
  \lesssim &~
\tTq \norm{w_{m}\eta}_{L^2_p}
\mathcal{I}_1,
\end{align*}
where
\begin{align*}
\mathcal{I}_1  
\eqdef &~
\left( \int_{\mathbb{R}^3} ~  d\pv ~ 
\left(\int_{\mathbb{R}^3} ~  d\qv 
~ v_{\Emptyset} ~
   \sqrt{J(\tqq q)}
 \left| w_{m}(\tqq \qv')h(q^\prime)w_{m}(\tqq \pv')f(p^\prime)\right|\right)^{2}\right)^{\frac12}
  \\
  \lesssim &~
\tmTq 
\left( \int_{\mathbb{R}^3} ~  d\pv ~ 
\int_{\mathbb{R}^3} ~  d\qv 
~ v_{\Emptyset} ~
  \left| w_{m}(\tqq \qv')h(q^\prime)w_{m}(\tqq \pv')f(p^\prime)\right|^{2}\right)^{\frac12}.
\end{align*}
We use the change of variables $(p^\prime, q^\prime) \mapsto (p,q)$ with Jacobian \eqref{Jacobian.COV} to obtain 
\begin{align*}
\mathcal{I}_1   
 \lesssim &~
\tmTq \norm{w_{m}f}_{L^2_p(\threed)}\norm{w_{m}h}_{L^2_p(\threed)}.
\end{align*}
Note that, as explained in \cite{MR2765751}, one can equivalently write $\mathcal{I}_1$ in the Glassey-Strauss representation of \eqref{collision.GS.flrw} and then use the change of variables $(p^\prime, q^\prime) \mapsto (p,q)$ in \eqref{Jacobian.COV}.  Afterwards we can switch back to the center of momentum representation shown in this proof.  This justifies the use of \eqref{Jacobian.COV}.  \end{proof}

\subsection{Estimates for averaged quantities}\label{subaveragequanti} 

Now we prove estimates for averaged quantities in the momentum variables against the basis functions in \eqref{orthogonal.basis.use}.  

\begin{lemma}
	\label{L8.6} For any suitably regular \(f\), and any element of $\{\tTq \mathfrak{e}_{\ell}(\tqq \pv)\}_{\ell=1}^{13}$ in \eqref{orthogonal.basis.use},  we have the following uniform estimates
\begin{align}\label{dispersion.term.basis.est}
   \left| \left\langle \micro  f, \tTq \mathfrak{e}_{\ell}(\tqq \pv)\right\rangle \right|
        \lesssim &~   \norm{J^{\frac{1}{4}} \micro   f}_{L^2_{\pv}},
\end{align}
and
    \begin{align}\label{lin.basis.est}
   \left| \left\langle \linL( \micro  f), \tTq \mathfrak{e}_{\ell}(\tqq \pv)\right\rangle \right|
        \lesssim &~ \tmTq \norm{J^{\frac{1}{8}} \micro  f}_{L^2_{\pv}}.
\end{align}
\end{lemma}

\begin{proof}[Proof of Lemma \ref{L8.6}]	
  For any element of \eqref{orthogonal.basis.use}, since $t \ge 1$, we can always split 
\begin{align}\label{basis.vector.split.bdd}
    \tTq \mathfrak{e}_{\ell}(\tqq \pv)
    =  \tTq\mathfrak{e}_{\ell}(\tqq \pv) J(\tqq \pv)^{-\frac{1}{4}} J(\tqq \pv)^{\frac{1}{4}}
    \leq 
    \tTq \mathfrak{e}_{\ell}(\tqq \pv) J(\tqq \pv)^{-\frac{1}{4}} J( \pv)^{\frac{1}{4}}.
\end{align}
Then we can include the factor $J( \pv)^{\frac{1}{4}}$ with $\micro f$ and the rest of the upper bound above is exponentially decaying and suitably time-normalised.  Then for $\langle \micro  f, \tTq \mathfrak{e}_{\ell}(\tqq \pv)\rangle$, using the Cauchy--Schwarz inequality, we obtain   \eqref{dispersion.term.basis.est}.

For the estimate of \eqref{lin.basis.est}, we use the self-adjoint property of $\linL$ as 
\begin{align*}
    \langle \linL( \micro  f), \tTq \mathfrak{e}_{\ell}(\tqq \pv)\rangle
    =
    \langle  \micro  f, \linL(\tTq \mathfrak{e}_{\ell}(\tqq \pv))\rangle.
\end{align*}
Then from \eqref{L0}, \eqref{nuDEFnext}, and \eqref{compactK}, we have for $\nu_0>0$ that
\begin{align*}
         \linL(\tTq \mathfrak{e}_{\ell}(\tqq \pv))
= \tTq \linL( \mathfrak{e}_{\ell}(\tqq \pv))
 =
 \nu_0 \mathfrak{e}_{\ell}(\tqq \pv)- \tTq\compK (\mathfrak{e}_{\ell}(\tqq \pv)).
\end{align*}
Further  $\compK = \compK_2 - \compK_1$ from \eqref{compactK} has the form 
\begin{align*}
 \tTq\compK_1 (\mathfrak{e}_{\ell}) 
 =
 \tQQq \int_{\mathbb{R}^3} ~  dq 
\int_{\mathbb{S}^{2}} ~ d\omega
~ v_{\Emptyset}   
~ \sqrt{J(\tqq\qv) J(\tqq\pv)} ~ \mathfrak{e}_{\ell}(\tqq \qv),
\end{align*}
and we split $\compK_2 = \compK_{21}+\compK_{22}$ where 
\begin{align*}
 \tTq\compK_{21} (\mathfrak{e}_{\ell}) 
 = &~
 \tQQq 
\int_{\mathbb{R}^3} ~  dq 
\int_{\mathbb{S}^{2}} ~ d\omega
~ v_{\Emptyset} ~  
\sqrt{J(\tqq \qv)} \sqrt{J(\tqq q^{\prime})} ~ \mathfrak{e}_{\ell}(\tqq \pv^{\prime}),
\\
 \tTq\compK_{22} (\mathfrak{e}_{\ell}) 
 = &~
 \tQQq 
\int_{\mathbb{R}^3} ~  dq 
\int_{\mathbb{S}^{2}} ~ d\omega
~ v_{\Emptyset} ~  
\sqrt{J(\tqq \qv)}\sqrt{J(\tqq p^{\prime})} ~ \mathfrak{e}_{\ell}(\tqq \qv^{\prime}).
\end{align*}
For simplicity we denote each term above as $\compK_* (\mathfrak{e}_{\ell}) =\compK_* (\mathfrak{e}_{\ell}(\tqq \cdot))$.  
By using the structure of the basis vectors \eqref{orthogonal.basis.use} and the conservation laws \eqref{collisionalCONSERVATION},  we have that each term $\compK_*$  is of the form 
\begin{align*}
 \tTq\compK_* (\mathfrak{e}_{\ell})
 =
 \tQQq J(\tqq\pv)^{\frac14}\int_{\mathbb{R}^3} ~  dq 
~ v_{\Emptyset}   ~ {\wb}_{*\ell}(\tqq \pv,\tqq \qv)
~ J(\tqq\qv)^{\frac14} ,
\end{align*}
where each ${\wb}_{*\ell}={\wb}_{*\ell}(\tqq \pv,\tqq \qv)$  can be written as 
\begin{align*}
{\wb}_{1\ell}
= &~
\int_{\mathbb{S}^{2}} ~ d\omega~ \mathfrak{e}_{\ell}(\tqq \qv) J(\tqq\qv)^{\frac14} J(\tqq\pv)^{\frac14},
\\
{\wb}_{21\ell}
= &~
\int_{\mathbb{S}^{2}} ~ d\omega~\mathfrak{e}_{\ell}(\tqq \pv^{\prime}) J(\tqq \pv^{\prime})^{-\frac12} \sqrt{J(\tqq \qv)}J(\tqq\qv)^{\frac14} J(\tqq\pv)^{\frac14},
\\
{\wb}_{22\ell}
= &~
\int_{\mathbb{S}^{2}} ~ d\omega~\mathfrak{e}_{\ell}(\tqq \qv^{\prime}) J(\tqq \qv^{\prime})^{-\frac12} \sqrt{J(\tqq \qv)}J(\tqq\qv)^{\frac14} J(\tqq\pv)^{\frac14}.
\end{align*}
Then each ${\wb}_{*\ell}$ is uniformly bounded by using, for example, \eqref{collisionalCONSERVATION}.  

In particular each term $\tTq\compK_j (\mathfrak{e}_{\ell})$ for $j\in \{1,2\}$ is of the form 
\begin{align*}
 \tTq\compK_j (\mathfrak{e}_{\ell})
 =
 \mathfrak{m}_{j}(\tqq \pv) J(\tqq\pv)^{\frac14}.
\end{align*}
Here $\mathfrak{m}_{j}(\tqq \pv)$ is a smooth uniformly bounded function.  We also denote for simplicity $\mathfrak{m}_{0}(\tqq \pv) = \nu_0 \mathfrak{e}_{\ell}(\tqq \pv) J(\tqq\pv)^{-\frac14}$.  Then $\tTq \linL( \mathfrak{e}_{\ell}) = \sum_{j=0}^2 \mathfrak{m}_{j}(\tqq \pv)J(\tqq\pv)^{\frac14}$.  By using again \eqref{basis.vector.split.bdd}, we have 
\begin{align*}
   \left| \left\langle \micro  f, \mathfrak{m}_{j}(\tqq \pv)J(\tqq\pv)^{\frac14}\right\rangle \right|
        \lesssim &~ \tmTq \norm{J^{\frac{1}{8}} \micro  \nabla_x f}_{L^2_{\pv}}.
\end{align*}
Thus, collecting the previous calculations, we obtain the estimate \eqref{lin.basis.est}.  
\end{proof}

See also \cite[Lemma 6.6]{RelBolNoncut2020} for a similar estimate in the massive case when \(\qq=0\).  We also prove estimates for the nonlinear term when averaged against a basis function.

\begin{lemma}
	\label{L8.7} 
    For any suitably regular functions \(f\), \(h\), and any element of $\{\tTq \mathfrak{e}_{\ell}(\tqq \pv)\}_{\ell=1}^{13}$ in \eqref{orthogonal.basis.use},  we have the following uniform estimates
    \begin{align*}
    	\left| \langle   \Gamma(f,h),\tTq \mathfrak{e}_{\ell}(\tqq \pv)\rangle \right|
        \lesssim &~
\| J^{\frac{1}{8}} f\|_{L^2_{\pv}} \| J^{\frac{1}{8}} h\|_{L^2_{\pv}}.
\end{align*}
\end{lemma}

\begin{proof}[Proof of Lemma \ref{L8.7}]	
From \eqref{gamma0} and \eqref{orthogonal.basis.use},  we must estimate 
\begin{align*}
    	\langle   \Gamma(f,h),\tTq \mathfrak{e}_{\ell}(\tqq \pv)\rangle
        = &~
\tQQq\int_{\mathbb{R}^3} ~  d\pv \int_{\mathbb{R}^3} ~  dq 
\int_{\mathbb{S}^{2}} ~ d\omega
~ v_{\Emptyset} ~  
 \sqrt{J(\tqq\qv)}\mathfrak{e}_{\ell}(\tqq \pv)
 f(p^{\prime })h(q^{\prime})
        \\
        &~
-\tQQq\int_{\mathbb{R}^3} ~  d\pv\int_{\mathbb{R}^3} ~  dq 
\int_{\mathbb{S}^{2}} ~ d\omega
~ v_{\Emptyset} ~  
 \sqrt{J(\tqq\qv)}\mathfrak{e}_{\ell}(\tqq \pv)
 f(p)h(q).
\end{align*}
After using the pre-post collision change of variables \((p,q)\mapsto (p',q')\) with Jacobian \eqref{Jacobian.COV} on the first term, we have
\begin{align*}
    	\langle   \Gamma(f,h),\tTq \mathfrak{e}_{\ell}(\tqq \pv)\rangle
        = &~
\tQQq\int_{\mathbb{R}^3} ~  d\pv\int_{\mathbb{R}^3} ~  dq ~
\mathfrak{m}_{\Gamma}(\tqq\pv,\tqq\qv)
J(\tqq\pv)^{\frac14}J(\tqq\qv)^{\frac14} 
 f(p)h(q),
\end{align*}
where now with \eqref{collisionalCONSERVATION} we have
\begin{align*}
\mathfrak{m}_{\Gamma}(\tqq\pv,\tqq\qv)
= &~
v_{\Emptyset}~J(\tqq\qv)^{\frac14} J(\tqq\pv)^{\frac14} \int_{\mathbb{S}^{2}} ~ d\omega~  
\mathfrak{e}_{\ell}(\tqq \pv^{\prime}) J(\tqq \pv^{\prime})^{-\frac12}
\\
 &~
-
v_{\Emptyset}~J(\tqq\qv)^{\frac14} J(\tqq\pv)^{\frac14} \mathfrak{e}_{\ell}(\tqq \pv) J(\tqq \pv)^{-\frac12}\int_{\mathbb{S}^{2}} ~ d\omega~  
.
\end{align*}
Then $\mathfrak{m}_{\Gamma}(\tqq\pv,\tqq\qv)$ is suitably time-normalised and uniformly bounded. Lemma \ref{L8.7} follows by using the Cauchy--Schwarz inequality.
\end{proof}
See also the proof of \cite[Lemma 6.7]{RelBolNoncut2020} in the case of massive particles when \(\qq=0\).

\section{Macroscopic estimates}\label{macroscopic.eqns.sec}

In this section we will estimate the solutions of the macroscopic equations from \eqref{macroscopic}, \eqref{local.cons.alt} and \eqref{macroscopic.B}.  In the course of these estimates, we will use that the zero mode is zero according to Corollary \ref{L8.5}. The estimates in this section hold for any $0 \leq \qq \leq 1$. We refer also to the macroscopic estimates obtained in \cite{EGKM-13} and in \cite[Section 5.1]{1904.12086} for the nonrelativistic Boltzmann equation

\begin{theorem}[Estimates for macroscopic quantities]\label{thmmacroscop}
Consider \(T>1\), and a suitably regular solution \(f\) of \eqref{rBoltz00} on \([1,T]\times \mathbb{T}^3_x\times \mathbb{R}^3_p\). Then the macroscopic quatities defined by \eqref{defmacrocoef} satisfy the following uniform estimate
\begin{align}\label{abcest.torus1}
\norm{t^{3\qq/2} [\mathcal{A},\mathcal{B},\mathcal{C}]}_{L^1_k L^2_T} &\lesssim
\Vert t^{3\qq/2}\micro f \Vert_{L^1_k L^2_T L^2_{\pv}} +\Vert \tTq f \Vert_{L^1_k L^\infty_T L^2_{\pv}} + \Vert f_1 \Vert_{L^1_k L^2_{\pv}} \\ \notag
&\qquad+ \int_{\mathbb{Z}^3_k} \left( \int^{T}_{1} t^{5\qq}\left| \widehat{\mathfrak{h}}(t,k)\right|^2
 dt \right)^{\frac12} d\Sigma (k),
\end{align}
where \(\norm{ [\mathcal{A},\mathcal{B},\mathcal{C}]}\eqdef\norm{ \mathcal{A}}+\norm{ \mathcal{B}}+\norm{ \mathcal{C}}\), and the definition of $ \widehat{\mathfrak{h}}(t,k)$ is provided below in \eqref{rhs.def.terms.FT}, \eqref{def.GaF}  and \eqref{index.sum.notation} with \eqref{orthogonal.basis.use}. 
\end{theorem}

\begin{proof}
First, we take the Fourier transform in \eqref{macroscopic} and \eqref{macroscopic.B} to obtain 
\begin{align} \label{macroscopic.FT.B}
 \tmTq\partial_t \left(\tTq \widehat{\mathcal{B}}_j \right) 
 +
\frac{2}{\sqrt{3}}\tmq ik_j  \widehat{\mathcal{C}} 
&=
-\tmTq  \partial_t\left( \tTq \widehat{\mathfrak{m}}_{j+13} \right)+\hat{\mathfrak{u}}_{j+13} , \quad (1 \leq j \leq 3),
\\ 
\label{macroscopic.FT.AC}
\tmq\sqrt{\frac{7}{3}}ik_j \left(\widehat{\mathcal{A}}
-
\sqrt{3}\widehat{\mathcal{C}} \right)
&=-\tmTq  \partial_t\left( \tTq \widehat{\mathfrak{m}}_{j+5} \right)+\hat{\mathfrak{u}}_{j+5},
\\ \label{macroscopic.FT.Bthree}
\tmq\frac{i}{2\sqrt{5}} \left( k_1 \widehat{\mathcal{B}}_1 - k_3 \widehat{\mathcal{B}}_3\right)
&=-\tmTq  \partial_t\left( \tTq \widehat{\mathfrak{m}}_{9} \right)+\widehat{\mathfrak{u}}_{9},
 \\ \label{macroscopic.FT.Btwo}
\tmq\frac{i}{6\sqrt{5}}\left( k\cdot \widehat{\mathcal{B}}-3k_2 \widehat{\mathcal{B}}_2\right)
&=-\tmTq  \partial_t\left( \tTq \widehat{\mathfrak{m}}_{10} \right)+\widehat{\mathfrak{u}}_{10},
 \\ \label{macroscopic.FT.Bjl}
\frac{1}{\sqrt{5}}  \tmq\left(ik_j \widehat{\mathcal{B}}_l+ik_l \widehat{\mathcal{B}}_j \right)
&=-\tmTq  \partial_t\left( \tTq \widehat{\mathfrak{m}}_{l+j+9} \right)+\hat{\mathfrak{u}}_{l+j+9},
\quad (l\neq j).
\end{align}
In particular using the orthonormal basis functions in \eqref{orthogonal.basis.use} we define
\begin{align} \notag
\widehat{\mathfrak{m}}_{\ell} \eqdef &~  \irnp{  \micro \hat{f}}{\tTq \mathfrak{e}_{\ell}(\tqq \pv)},
\\ \label{rhs.def.terms.FT}
\widehat{\highG}_{\ell} \eqdef &~  \irnp{\micro \hat{f}}{\tTq \frac{\pv}{| \pv|}\mathfrak{e}_{\ell}(\tqq \pv)},
\\ \notag
\hat{\mathfrak{l}}_{\ell} \eqdef &~  -\irnp{\linL(\micro \hat{f})}{\tTq \mathfrak{e}_{\ell}(\tqq \pv)},
\\ \notag
\widehat{\mathfrak{h}}_{\ell}  \eqdef &~ 
\irnp{\widehat{\Gamma}(\hat{f},\hat{f})}{\tTq \mathfrak{e}_{\ell}(\tqq \pv)},
 \\ \notag
\hat{\mathfrak{u}}_{\ell} \eqdef &~  -\tmq ik \cdot \widehat{\highG}_{\ell}
 +\hat{\mathfrak{l}}_{\ell}+\widehat{\mathfrak{h}}_{\ell},
\end{align}
where in the second to last term, from \eqref{gamma0}, we have
\begin{align}\label{def.GaF}
\widehat{\Gamma}(\hat{f},\hat{h})(k,\pv)
= &~
t^{3\qq}\int_{\mathbb{R}^3} ~  dq 
\int_{\mathbb{S}^{2}} ~ d\omega
~ v_{\Emptyset} ~  
 \sqrt{J(\tqq\qv)}
 [\hat{f}(p^{\prime})*\hat{h}(q^{\prime})](k)
 \\ \nonumber
 &~ -
 t^{3\qq}\int_{\mathbb{R}^3} ~  dq 
\int_{\mathbb{S}^{2}} ~ d\omega
~ v_{\Emptyset} ~  
 \sqrt{J(\tqq\qv)}
 [\hat{f}(p)*\hat{h}(q)](k),
\end{align}
with
\begin{eqnarray*}
{[\hat{f}(\pv^{\prime})*\hat{h}(\qv^{\prime})]}(k) & \eqdef &\int_{\Z^3_l} \hat{f}(k-l,\pv^{\prime})\hat{h}(l,\qv^{\prime})\,d\Si(l),\\
{[\hat{f}(\pv)*\hat{h}(\qv)]}(k)  & \eqdef &\int_{\Z^3_l}  \hat{f}(k-l,\pv)\hat{h}(l,\qv)\,d\Si(l).
\end{eqnarray*}
 We will also use the notation 
\begin{align}\label{index.sum.notation}
\big| \widehat{\mathfrak{m}} \big| = \sum_{\ell=1}^{13} \big| \widehat{\mathfrak{m}}_{\ell}\big|, \quad  
    \big| \widehat{\highG}\big| = \sum_{\ell=1}^{13} \big| \widehat{\highG}_{\ell}\big|, \quad 
    \big| \hat{\mathfrak{l}} \big| = \sum_{\ell=1}^{13} \big| \hat{\mathfrak{l}}_{\ell} \big|, 
    \quad 
    \big| \widehat{\mathfrak{h}} \big| = \sum_{\ell=1}^{13} \big| \widehat{\mathfrak{h}}_{\ell} \big|.
\end{align}
In the rest of these estimates we write $\widehat{\mathfrak{m}}_{\ell}$, $\widehat{\highG}_{\ell}$, $\hat{\mathfrak{l}}_{\ell}$, $\widehat{\mathfrak{h}}_{\ell}$, or $\hat{\mathfrak{u}}_{\ell}$ to denote an arbitrary term with index $\ell$ when the value of the index is unimportant. 

We take the Fourier transform of \eqref{local.cons.alt} to obtain the local conservation laws
\begin{align} \label{local.conservation.FTa}
\partial_t\widehat{\mathcal{A}}
     +
\frac{1}{2}  \tmq ik \cdot \widehat{\mathcal{B}}
+
\frac{3\qq }{t} \widehat{\mathcal{A}}
&=-\tmq ik \cdot \widehat{\highG}_{1},
\\  \label{local.conservation.FTb}
\partial_t\widehat{\mathcal{B}}_j
 +
2\tmq ik_j \left( \widehat{\mathcal{A}}
-\frac{2}{\sqrt{3}}
\widehat{\mathcal{C}}  \right)
+
\frac{3\qq }{t}\widehat{\mathcal{B}}_j
&=-\tmq ik \cdot \widehat{\highG}_{j+1}, \quad (1 \leq j \leq 3),
\\ \label{local.conservation.FTc}
\partial_t\widehat{\mathcal{C}}
  +
\frac{1}{2\sqrt{3}}  \tmq  ik \cdot \widehat{\mathcal{B}}
+
\frac{3\qq }{t} \widehat{\mathcal{C}}
&=-\tmq ik \cdot \widehat{\highG}_{5}.
\end{align}
We will use \eqref{macroscopic.FT.B}-\eqref{macroscopic.FT.Bjl} and \eqref{local.conservation.FTa}-\eqref{local.conservation.FTc} with \eqref{rhs.def.terms.FT} to prove our desired estimates.

First we will estimate $\widehat{\mathcal{C}}(t,k)$. We multiply \eqref{macroscopic.FT.B} by $-t^{4\qq} ik_j  \overline{\widehat{\mathcal{C}}}(t,k)$ and sum on $j$ to obtain
\begin{align*}
 \frac{2}{\sqrt{3}}\tTq |k|^2   \left|\widehat{\mathcal{C}}(t,k)\right|^2
& =   i \tq\overline{\widehat{\mathcal{C}}}(t,k)\partial_t \left(\tTq \widehat{\mathcal{B}}\cdot k +\sum_{j=1}^3 k_j   \tTq \widehat{\mathfrak{m}}_{j+13} \right)
\\
& \quad -
i t^{4\qq}\overline{\widehat{\mathcal{C}}}(t,k)  \sum_{j=1}^3 k_j\hat{\mathfrak{u}}_{j+13} .
\end{align*}
We rearrange the time derivative as
\begin{multline*}
 \frac{2}{\sqrt{3}}\tTq |k|^2   \left|\widehat{\mathcal{C}}(t,k)\right|^2
=   i \frac{d}{dt} \left( \left(\tTq \widehat{\mathcal{B}}\cdot k +\sum_{j=1}^3 k_j   \tTq \widehat{\mathfrak{m}}_{j+13} \right)\tq\overline{\widehat{\mathcal{C}}}(t,k) \right)
\\
-
i  \left( \widehat{\mathcal{B}}\cdot k +\sum_{j=1}^3 k_j    \widehat{\mathfrak{m}}_{j+13} \right) \tTq\partial_t  \left(\tq\overline{\widehat{\mathcal{C}}}(t,k)\right)
-
i t^{4\qq}\overline{\widehat{\mathcal{C}}}(t,k)  \sum_{j=1}^3 k_j\hat{\mathfrak{u}}_{j+13} .
\end{multline*}
After dividing by $|k|^2 \ge 1$ and integrating in time we have
\begin{multline*}
 \frac{2}{\sqrt{3}}\int_1^{T} \tTq|\widehat{\mathcal{C}}(t,k)|^2  dt
=
\left. \frac{i}{|k|^2} \left( \left( \widehat{\mathcal{B}}\cdot k +\sum_{j=1}^3 k_j    \widehat{\mathfrak{m}}_{j+13} \right)t^{4\qq}\overline{\widehat{\mathcal{C}}}(t,k) \right)\right|_{t=1}^T
\\
-
\frac{i}{|k|^2} \int_1^{T} \left( \widehat{\mathcal{B}}\cdot k +\sum_{j=1}^3 k_j    \widehat{\mathfrak{m}}_{j+13} \right) \tTq\partial_t  \left(\tq\overline{\widehat{\mathcal{C}}}(t,k)\right) dt
-
\frac{i}{|k|^2} \int_1^{T}t^{4\qq}\overline{\widehat{\mathcal{C}}}(t,k)  \sum_{j=1}^3 k_j\hat{\mathfrak{u}}_{j+13} dt.
\end{multline*}
We will estimate each of the terms on the right.

From \eqref{rhs.def.terms.FT} and \eqref{dispersion.term.basis.est} we have 
\begin{align}\label{mTH.est}
\left| \widehat{\mathfrak{m}}(t,k)\right|
+
\left| \widehat{\highG}(t,k)\right|
 \lesssim 
\norm{\micro \hat{f}(t,k)}_{L^2_{\pv}}.
\end{align}
Thus, for $|k|\ge 1$ and $T \ge 1$ we have 
\begin{multline}\label{upperBIGT.est}
    \left| \frac{1}{|k|^2} \sum_{j=1}^3 
 T^{4\qq}  k_j\widehat{\mathfrak{m}}_{j+13} (T) \overline{\widehat{\mathcal{C}}}(T)\right|
 \lesssim  
 T^{4\qq}\norm{\micro \hat{f}(T,k)}_{L^2_{\pv}} \left|\widehat{\mathcal{C}}(T,k)\right|
 \\
  \lesssim  
 T^{4\qq}\Vert \hat{f}(T,k)\Vert_{L^2_{\pv}}^2
   \lesssim 
 T^{6\qq}\Vert \hat{f}(T,k)\Vert_{L^2_{\pv}}^2,
\end{multline}
and similarly 
\begin{align}\label{upperONET.est}
    \left| \frac{1}{|k|^2} \sum_{j=1}^3 
   k_j\widehat{\mathfrak{m}}_{j+13} (1) \overline{\widehat{\mathcal{C}}}(1)\right|
 \lesssim  
\norm{ \hat{f}_1(k)}_{L^2_{\pv}}^2.
\end{align}
The term  
$\left. \frac{i}{|k|^2} \left( \left( \widehat{\mathcal{B}}\cdot k  \right)t^{4\qq}\overline{\widehat{\mathcal{C}}}(t,k) \right)\right|_{t=1}^T$ also satisfies the estimates in \eqref{upperBIGT.est} and \eqref{upperONET.est} at the respective points $t=T$ and $t=1$.  Thus 
\begin{align*}
    \left|
\left. \frac{1}{|k|^2} \left( \left( \widehat{\mathcal{B}}\cdot k +\sum_{j=1}^3 k_j   \widehat{\mathfrak{m}}_{j+13} \right)t^{4\qq}\overline{\widehat{\mathcal{C}}}(t,k) \right)\right|_{t=1}^T
\right|
\lesssim
 T^{6\qq}\Vert \hat{f}(T,k)\Vert_{L^2_{\pv}}^2
 +
 \norm{ \hat{f}_1(k)}_{L^2_{\pv}}^2.
\end{align*}
Next, from \eqref{local.conservation.FTc} we have
\begin{align*}
    \partial_t\left( \tq \widehat{\mathcal{C}} \right)
  +
\frac{1}{2\sqrt{3}}   ik \cdot \widehat{\mathcal{B}}
+
\frac{2\qq }{t} \tq\widehat{\mathcal{C}}
=- ik \cdot \widehat{\highG}_{5}.
\end{align*}
Hence
\begin{align}\label{usingqleq1}
  \left|  \tTq \partial_t\left(\tq\overline{\widehat{\mathcal{C}}}\right) \right|
 \lesssim 
\tTq |k| \left( \left|\widehat{\mathcal{B}}\right|+\left| \widehat{\highG}\right| \right)
+
2\qq  t^{4\qq -1}\left|\widehat{\mathcal{C}}\right|.
\end{align}
Thus for $\qq \leq 1$ and for any small $\eta>0$ we have
\begin{multline*}
    \left|
\frac{1}{|k|^2} \int_1^{T} \left( \widehat{\mathcal{B}}\cdot k +\sum_{j=1}^3 k_j    \widehat{\mathfrak{m}}_{j+13} \right) \tTq\partial_t  \left(\tq\overline{\widehat{\mathcal{C}}}(t,k)\right) dt
\right|
\\
\leq
C_\eta\int_1^{T} \tTq\norm{\micro \hat{f}(t,k)}_{L^2_{\pv}}^2dt
+
C_\eta \int^{T}_{1} \tTq \left|\widehat{{\mathcal{B}}}(t,k)\right|^2 dt 
+
\eta \int^{T}_{1} \tTq \left|\widehat{\mathcal{C}}(t,k)\right|^2 dt. 
\end{multline*}
Now from \eqref{lin.basis.est} we obtain
\begin{align}\label{linear.op.est.grad}
\big| \widehat{\mathfrak{l}}\big|
 \lesssim 
 \tmTq
\norm{\micro \hat{f}(t,k)}_{L^2_{\pv}}.
\end{align}
With \eqref{rhs.def.terms.FT} and \eqref{linear.op.est.grad} and the previous estimates for any small $\eta>0$ we obtain
\begin{align*}
\left|  \frac{1}{|k|^2} \int_1^{T}t^{4\qq}\overline{\widehat{\mathcal{C}}}(t,k)  \sum_{j=1}^3 k_j\hat{\mathfrak{u}}_{j+13} dt
 \right|
 \leq &~
\int_1^{T}
 t^{4\qq}\left( \tmq\left|\widehat{\highG}(t,k)\right|+\big| \widehat{\mathfrak{l}}\big|+\big| \widehat{\mathfrak{h}}\big|\right) \left|\widehat{{\mathcal{C}}}(t,k)\right| dt
 \\
  \leq &~
  \eta \int^{T}_{1} \tTq  \left|\widehat{{\mathcal{C}}}(t,k)\right|^2 dt
  +C_\eta\int_1^{T} t^{5\qq}\big| \widehat{\mathfrak{h}}\big|^2dt
  \\
&~
+C_\eta\int_1^{T} \tTq\norm{\micro \hat{f}(t,k)}_{L^2_{\pv}}^2dt.
\end{align*}
We collect these estimates for $\widehat{\mathcal{C}}(t,k)$, choosing $\eta>0$ sufficiently small, to obtain
\begin{multline}\label{FT.C.estimate}
\int_1^{T} \tTq|\widehat{\mathcal{C}}(t,k)|^2  dt
\lesssim
 T^{6\qq}\Vert \hat{f}(T,k)\Vert_{L^2_{\pv}}^2+\norm{ \hat{f}_1(k)}_{L^2_{\pv}}^2
  +
\int^{T}_{1} \tTq \left|\widehat{{\mathcal{B}}}(t,k)\right|^2 dt 
 \\
+\int_1^{T} \tTq\norm{\micro \hat{f}(t,k)}_{L^2_{\pv}}^2dt
+\int_1^{T} t^{5\qq}\big| \widehat{\mathfrak{h}}\big|^2dt.
\end{multline}
This is our primary estimate for $\widehat{\mathcal{C}}(t,k)$.

Next we estimate $\widehat{\mathcal{A}}- \sqrt{3}\widehat{\mathcal{C}}$.  Multiply \eqref{macroscopic.FT.AC} by $-t^{4\qq} ik_j  \left( \overline{\widehat{\mathcal{A}}}- \sqrt{3}\overline{\widehat{\mathcal{C}}}\right)(t,k)$ and sum on $j$ to obtain
\begin{align*}
\sqrt{\frac{7}{3}}\tTq |k|^2   \left|\widehat{\mathcal{A}}(t,k)- \sqrt{3}\widehat{\mathcal{C}}(t,k)\right|^2
= &~ 
 i \tq\left( \overline{\widehat{\mathcal{A}}}- \sqrt{3}\overline{\widehat{\mathcal{C}}}\right)(t,k)\partial_t \left(\tTq\sum_{j=1}^3 k_j    \widehat{\mathfrak{m}}_{j+5} \right)
\\
&~
-
i t^{4\qq}\left( \overline{\widehat{\mathcal{A}}}- \sqrt{3}\overline{\widehat{\mathcal{C}}}\right)(t,k)  \sum_{j=1}^3 k_j\hat{\mathfrak{u}}_{j+5}.
\end{align*}
We rearrainge the time derivative to conclude that
\begin{multline*}
\sqrt{\frac{7}{3}}\tTq |k|^2   \left|\widehat{\mathcal{A}}(t,k)- \sqrt{3}\widehat{\mathcal{C}}(t,k)\right|^2
=   i \frac{d}{dt} \left(  \left(\tTq\sum_{j=1}^3 k_j    \widehat{\mathfrak{m}}_{j+5} \right)\tq\left( \overline{\widehat{\mathcal{A}}}- \sqrt{3}\overline{\widehat{\mathcal{C}}}\right) \right)
\\
-
i  \left( \sum_{j=1}^3 k_j    \widehat{\mathfrak{m}}_{j+5} \right) \tTq\partial_t  \left(\tq\left( \overline{\widehat{\mathcal{A}}}- \sqrt{3}\overline{\widehat{\mathcal{C}}}\right)\right)
-
i t^{4\qq}\left( \overline{\widehat{\mathcal{A}}}- \sqrt{3}\overline{\widehat{\mathcal{C}}}\right)  \sum_{j=1}^3 k_j\hat{\mathfrak{u}}_{j+5}.
\end{multline*}
Next divide by $|k|^2 \ge 1$ and then integrate in time to achieve
\begin{multline*}
\sqrt{\frac{7}{3}}\int_1^{T} \tTq \left|\widehat{\mathcal{A}}(t,k)- \sqrt{3}\widehat{\mathcal{C}}(t,k)\right|^2  dt
\\
=
\left. \frac{i}{|k|^2} \left(  \left(\tTq\sum_{j=1}^3 k_j    \widehat{\mathfrak{m}}_{j+5} \right)\tq\left( \overline{\widehat{\mathcal{A}}}- \sqrt{3}\overline{\widehat{\mathcal{C}}}\right)(t,k) \right)\right|_{t=1}^T
\\
-
\frac{i}{|k|^2} \int_1^{T} \left( \sum_{j=1}^3 k_j    \widehat{\mathfrak{m}}_{j+5} \right) \tTq\partial_t  \left(\tq\left( \overline{\widehat{\mathcal{A}}}- \sqrt{3}\overline{\widehat{\mathcal{C}}}\right)(t,k)\right) dt
\\
-
\frac{i}{|k|^2} \int_1^{T} t^{4\qq}\left( \overline{\widehat{\mathcal{A}}}- \sqrt{3}\overline{\widehat{\mathcal{C}}}\right)(t,k)  \sum_{j=1}^3 k_j\hat{\mathfrak{u}}_{j+5} dt.
\end{multline*}
We will estimate each of the terms on the right.  

Then as in \eqref{upperBIGT.est} and \eqref{upperONET.est} we have 
\begin{align*}
    \left|
\left. \frac{1}{|k|^2} \left(  \left(\tTq\sum_{j=1}^3 k_j    \widehat{\mathfrak{m}}_{j+5} \right)\tq\left( \overline{\widehat{\mathcal{A}}}- \sqrt{3}\overline{\widehat{\mathcal{C}}}\right)(t,k) \right)\right|_{t=1}^T
\right|
\lesssim
 T^{6\qq}\Vert \hat{f}(T,k)\Vert_{L^2_{\pv}}^2
 +
 \norm{ \hat{f}_1(k)}_{L^2_{\pv}}^2.
\end{align*}
Next subtract $\sqrt{3}$ times \eqref{local.conservation.FTc} from  \eqref{local.conservation.FTa} to obtain 
\begin{align}\notag
    \partial_t \left( \widehat{\mathcal{A}}-\sqrt{3}\widehat{\mathcal{C}} \right)
+
\frac{3\qq }{t} \left( \widehat{\mathcal{A}}-\sqrt{3}\widehat{\mathcal{C}} \right)
&=-\tmq ik \cdot \left( \widehat{\highG}_{1}-\sqrt{3}\widehat{\highG}_{5} \right).
\end{align}
In particular 
\begin{align}\notag
    \partial_t \left(\tq\left( \widehat{\mathcal{A}}-\sqrt{3}\widehat{\mathcal{C}} \right)\right)
+
\frac{2\qq }{t} \tq \left( \widehat{\mathcal{A}}-\sqrt{3}\widehat{\mathcal{C}} \right)
&=- ik \cdot \left( \widehat{\highG}_{1}-\sqrt{3}\widehat{\highG}_{5} \right).
\end{align}
Thus we have 
\begin{align*}
  \left|  \tTq \partial_t \left(\tq\left( \overline{\widehat{\mathcal{A}}}- \sqrt{3}\overline{\widehat{\mathcal{C}}}\right)(t,k)\right) \right|
 \lesssim 
\tTq |k| \left| \widehat{\highG}\right|
+
2\qq  t^{4\qq -1}\left|{\widehat{\mathcal{A}}}- \sqrt{3}{\widehat{\mathcal{C}}}\right|.
\end{align*}
We conclude, again using $\qq \leq 1$, for any small $\eta>0$ that 
\begin{multline*}
    \left|
\frac{1}{|k|^2} \int_1^{T} \left( \sum_{j=1}^3 k_j    \widehat{\mathfrak{m}}_{j+5} \right) \tTq\partial_t  \left(\tq\left( \overline{\widehat{\mathcal{A}}}- \sqrt{3}\overline{\widehat{\mathcal{C}}}\right)(t,k)\right) dt
\right|
\\
\leq
C_\eta\int_1^{T} \tTq\norm{\micro \hat{f}(t,k)}_{L^2_{\pv}}^2dt
+
\eta \int^{T}_{1} \tTq \left|{\widehat{\mathcal{A}}}- \sqrt{3}{\widehat{\mathcal{C}}}\right|^2(t,k) dt. 
\end{multline*}
Also using \eqref{rhs.def.terms.FT} and \eqref{linear.op.est.grad} and the prior estimates  for any small $\eta>0$ we obtain 
\begin{multline*}
\left|  \frac{1}{|k|^2} 
\int_1^{T} t^{4\qq}\left( \overline{\widehat{\mathcal{A}}}- \sqrt{3}\overline{\widehat{\mathcal{C}}}\right)(t,k)  \sum_{j=1}^3 k_j\hat{\mathfrak{u}}_{j+5} dt
 \right|
\\
\leq
\int_1^{T}
 t^{4\qq}\left( \tmq\left|\widehat{\highG}(t,k)\right|+\left| \widehat{\mathfrak{l}}\right|+\left| \widehat{\mathfrak{h}}\right|\right) \left|{\widehat{\mathcal{A}}}- \sqrt{3}{\widehat{\mathcal{C}}}\right|(t,k) dt
\\
\leq
C_\eta\int_1^{T} \tTq\norm{\micro \hat{f}(t,k)}_{L^2_{\pv}}^2dt
  +
  C_\eta\int_1^{T} t^{5\qq}\left| \widehat{\mathfrak{h}}\right|^2dt
+
\eta \int^{T}_{1} \tTq \left|{\widehat{\mathcal{A}}}- \sqrt{3}{\widehat{\mathcal{C}}}\right|^2(t,k) dt. 
\end{multline*}
We collect the estimates above choosing $\eta>0$ sufficiently small to obtain
\begin{multline}\label{FT.AC.estimate}
\int_1^{T} \tTq\left|\left( {\widehat{\mathcal{A}}}- \sqrt{3}{\widehat{\mathcal{C}}}\right)(t,k)\right|^2  dt
\lesssim
 T^{6\qq}\Vert \hat{f}(T,k)\Vert_{L^2_{\pv}}^2+\norm{ \hat{f}_1(k)}_{L^2_{\pv}}^2
 \\
 +
\int_1^{T} \tTq\norm{\micro \hat{f}(t,k)}_{L^2_{\pv}}^2dt
+
\int_1^{T} t^{5\qq}\left| \widehat{\mathfrak{h}}\right|^2dt.
\end{multline}
This is our primary estimate for ${\widehat{\mathcal{A}}}- \sqrt{3}{\widehat{\mathcal{C}}}$.

Now we take a linear combination of \eqref{FT.C.estimate} and \eqref{FT.AC.estimate} to obtain 
\begin{multline}\label{FT.ACt.estimate}
    \int^{T}_{1} \tTq \left|\widehat{[\mathcal{A},\mathcal{C}]}(t,k)\right|^2 dt
    \lesssim
     T^{6\qq}\Vert \hat{f}(T,k)\Vert_{L^2_{\pv}}^2+\norm{ \hat{f}_1(k)}_{L^2_{\pv}}^2
  +
 \int^{T}_{1} \tTq \left|\widehat{{\mathcal{B}}}(t,k)\right|^2 dt 
 \\
+\int_1^{T} \tTq\norm{\micro \hat{f}(t,k)}_{L^2_{\pv}}^2dt
+\int_1^{T} t^{5\qq}\left| \widehat{\mathfrak{h}}\right|^2dt.
\end{multline}
This will be our primary estimate for $\widehat{[\mathcal{A},\mathcal{C}]}(t,k)$.

Lastly we will estimate $\widehat{\mathcal{B}}$.  To this end we multiply \eqref{macroscopic.FT.Bjl} by $-ik_j \tq \overline{\widehat{\mathcal{B}}}_l$ and sum over $l\neq j$, to obtain
\begin{multline}
        -\sum_{\substack{l=1 \\ l\neq j}}^3 \left(ik_j \widehat{\mathcal{B}}_l+ik_l \widehat{\mathcal{B}}_j \right)ik_j \overline{\widehat{\mathcal{B}}}_l
    = (k_j)^2 \left(|\widehat{\mathcal{B}}(t,k)|^2 
    -
    ( \widehat{\mathcal{B}}_j )^2\right) 
    + 
     k_j \overline{\widehat{\mathcal{B}}}_j\left( k \cdot \widehat{\mathcal{B}} -k_j \widehat{\mathcal{B}}_j\right)
    \\  \label{sum.bj.calc}
     = (k_j)^2 |\widehat{\mathcal{B}}(t,k)|^2 
    + 
     k_j \overline{\widehat{\mathcal{B}}}_j\left(\left( k \cdot \widehat{\mathcal{B}} \right)
        -
2   k_j   \widehat{\mathcal{B}}_j  \right).
\end{multline}
Next, by taking suitable linear combinations of \eqref{macroscopic.FT.Bthree} and \eqref{macroscopic.FT.Btwo} we observe that for any $1\leq j\leq 3$ we have 
\begin{align}\label{macroscopic.FT.Bjjj}
    \tmq\frac{i}{6\sqrt{5}}\left( k\cdot \widehat{\mathcal{B}}-3k_j \widehat{\mathcal{B}}_j\right)
&=-\tmTq  \partial_t\left( \tTq \widehat{\mathfrak{m}}_{*j} \right)+\widehat{\mathfrak{u}}_{*j},
\end{align}
where $\widehat{\mathfrak{m}}_{*j}$ is some linear combination of  $\widehat{\mathfrak{m}}_{9}$ and $\widehat{\mathfrak{m}}_{10}$ that depends upon $j$, and  $\widehat{\mathfrak{u}}_{*j}$ is a similar linear combination of $\widehat{\mathfrak{u}}_{9}$ and $\widehat{\mathfrak{u}}_{10}$ that depends upon $j$.  Next multiply \eqref{macroscopic.FT.Bjjj} by $-ik_j \tq \overline{\widehat{\mathcal{B}}}_j$ to obtain 
\begin{align*}
    \left( k\cdot \widehat{\mathcal{B}}-2k_j \widehat{\mathcal{B}}_j\right)
    \frac{k_j \overline{\widehat{\mathcal{B}}}_j}{6\sqrt{5}}
= \frac{(k_j)^2 |\widehat{\mathcal{B}}_j|^2}{6\sqrt{5}}
+
t^{-2\qq}  ik_j\overline{\widehat{\mathcal{B}}}_j\partial_t\left( \tTq \widehat{\mathfrak{m}}_{*j} \right)
-
ik_j \tq\overline{\widehat{\mathcal{B}}}_j\widehat{\mathfrak{u}}_{*j}.
\end{align*}
Note that $|k^T\widehat{\mathcal{B}}(t,k)|^2 = \sum_{j=1}^3(k_j)^2 |\widehat{\mathcal{B}}_j|^2$.  
We plug this into \eqref{sum.bj.calc}, use \eqref{macroscopic.FT.Bjl}, and sum on $j$ to obtain 
\begin{multline*}
    |k|^2 |\widehat{\mathcal{B}}(t,k)|^2 + |k^T\widehat{\mathcal{B}}(t,k)|^2
    =
    -\sqrt{5}\sum_{j=1}^3\sum_{\substack{l=1 \\ l\neq j}}^3 \left(ik_j \widehat{\mathcal{B}}_l+ik_l \widehat{\mathcal{B}}_j \right)ik_j \overline{\widehat{\mathcal{B}}}_l
    \\
    +
6\sqrt{5}it^{-2\qq} 
\sum_{j=1}^3 
k_j\overline{\widehat{\mathcal{B}}}_j\partial_t\left( \tTq \widehat{\mathfrak{m}}_{*j} \right)
-
6\sqrt{5}i\tq
\sum_{j=1}^3
k_j \overline{\widehat{\mathcal{B}}}_j\widehat{\mathfrak{u}}_{*j}
\\
=
- t^{-2\qq}  \sqrt{5} \sum_{j=1}^3\sum_{\substack{l=1 \\ l\neq j}}^3 ik_j \overline{\widehat{\mathcal{B}}}_l\partial_t\left( \tTq \widehat{\mathfrak{m}}_{l+j+9} \right)
+
\sqrt{5}\tq\sum_{j=1}^3\sum_{\substack{l=1 \\ l\neq j}}^3 ik_j \overline{\widehat{\mathcal{B}}}_l \hat{\mathfrak{u}}_{l+j+9}
    \\
    +
6\sqrt{5}it^{-2\qq} 
\sum_{j=1}^3 
k_j\overline{\widehat{\mathcal{B}}}_j\partial_t\left( \tTq \widehat{\mathfrak{m}}_{*j} \right)
-
6\sqrt{5}i\tq
\sum_{j=1}^3
k_j \overline{\widehat{\mathcal{B}}}_j\widehat{\mathfrak{u}}_{*j}.
\end{multline*}
We have shown that 
\begin{align*}
    |k|^2 |\widehat{\mathcal{B}}(t,k)|^2 + |k^T\widehat{\mathcal{B}}(t,k)|^2
    =
t^{-2\qq} 
\sum_{l=1}^3 
k_l\overline{\widehat{\mathcal{B}}}_l\partial_t\left( \tTq \widehat{\mathfrak{m}}_{\ddagger l} \right)
-
\tq
\sum_{l=1}^3
k_l \overline{\widehat{\mathcal{B}}}_l \widehat{\mathfrak{u}}_{\ddagger l},
\end{align*}
where $\widehat{\mathfrak{m}}_{\ddagger l}$ is a linear combination of the $\widehat{\mathfrak{m}}_{l+j+9}$ and the $\widehat{\mathfrak{m}}_{*l}$, 
and $\widehat{\mathfrak{u}}_{\ddagger l}$ is a similar linear combination of the $\widehat{\mathfrak{u}}_{l+j+9}$ and the $\widehat{\mathfrak{u}}_{*l}$.

After multiplying by $\tTq$ and dividing by $|k|^2$ we thus have
\begin{multline*}
 \tTq |\widehat{\mathcal{B}}(t,k)|^2 + \tTq \frac{|k^T\widehat{\mathcal{B}}(t,k)|^2}{|k|^2}
=
\frac{ \tq}{|k|^2}
\sum_{l=1}^3 
k_l\overline{\widehat{\mathcal{B}}}_l\partial_t\left( \tTq \widehat{\mathfrak{m}}_{\ddagger l} \right)
-
\frac{t^{4\qq} }{|k|^2}
\sum_{l=1}^3
k_l \overline{\widehat{\mathcal{B}}}_l \widehat{\mathfrak{u}}_{\ddagger l}
\\
=
\frac{d}{dt}\left(
\frac{ t^{4\qq}}{|k|^2}
\sum_{l=1}^3 
k_l\overline{\widehat{\mathcal{B}}}_l \widehat{\mathfrak{m}}_{\ddagger l} 
 \right)
 -
 \frac{1}{|k|^2 } 
 \sum_{l=1}^3 
  \tTq k_l\widehat{\mathfrak{m}}_{\ddagger l}\partial_t\left(\tq\overline{\widehat{\mathcal{B}}}^j\right) 
-
\frac{t^{4\qq} }{|k|^2}
\sum_{l=1}^3
k_l \overline{\widehat{\mathcal{B}}}_l \widehat{\mathfrak{u}}_{\ddagger l}.
\end{multline*}
Similar to the previous cases, after integrating in time over $[1,T]$, we obtain
\begin{multline*}
\int_1^{T} \tTq \left(  |\widehat{\mathcal{B}}(t,k)|^2 +  \frac{|k^T\widehat{\mathcal{B}}(t,k)|^2}{|k|^2}\right)  dt
=
 \frac{ T^{4\qq}}{|k|^2}
\sum_{l=1}^3 
k_l\overline{\widehat{\mathcal{B}}}_l(T) \widehat{\mathfrak{m}}_{\ddagger l} (T)
 \\
- \frac{ 1}{|k|^2}
\sum_{l=1}^3 
k_l\overline{\widehat{\mathcal{B}}}_l(1) \widehat{\mathfrak{m}}_{\ddagger l} (1)
  -
 \frac{1}{|k|^2 } 
 \sum_{l=1}^3 
   k_l \int_1^{T} \tTq \widehat{\mathfrak{m}}_{\ddagger l}\partial_t\left(\tq\overline{\widehat{\mathcal{B}}}_l\right)  dt
-
\frac{1 }{|k|^2}
\sum_{l=1}^3
k_l \int_1^{T}  t^{4\qq} \overline{\widehat{\mathcal{B}}}_l \widehat{\mathfrak{u}}_{\ddagger l} dt.
\end{multline*}
We will estimate each of the terms on the right side.

As in \eqref{upperBIGT.est} and \eqref{upperONET.est}, for $|k|\geq 1$, we have 
\begin{align*}
 \frac{ T^{4\qq}}{|k|^2}
\sum_{l=1}^3 
    \left| k_l\overline{\widehat{\mathcal{B}}}_l(T) \widehat{\mathfrak{m}}_{\ddagger l} (T)
\right|
+ \frac{ 1}{|k|^2}
\sum_{l=1}^3 
\left| k_l\overline{\widehat{\mathcal{B}}}_l(1) \widehat{\mathfrak{m}}_{\ddagger l} (1)\right|
\lesssim
 T^{6\qq}\Vert \hat{f}(T,k)\Vert_{L^2_{\pv}}^2
 +
 \norm{ \hat{f}_1(k)}_{L^2_{\pv}}^2.
\end{align*}
Next from \eqref{local.conservation.FTb} we have 
\begin{align*}
\tTq\partial_t\left(\tq {\widehat{\mathcal{B}}}_j\right)
 +
2\tTq ik_j \left( \widehat{\mathcal{A}}
-\frac{2}{\sqrt{3}}
\widehat{\mathcal{C}}  \right)
+
\frac{2\qq }{t}t^{4\qq}\widehat{\mathcal{B}}_j
&=- \tTq ik \cdot \widehat{\highG}_{j+1}, \quad (1 \leq j \leq 3),
\end{align*}
Hence for any $1 \leq j \leq 3$ we have
\begin{align*}
  \left|  \tTq\partial_t\left(\tq {\widehat{\mathcal{B}}}_j\right) \right|
 \lesssim 
\tTq |k| \left(\left| \widehat{\highG}\right|+\left|\widehat{[\mathcal{A},\mathcal{C}]}\right|\right)
+
2\qq  t^{4\qq -1}\left|{\widehat{\mathcal{B}}}\right|.
\end{align*}
Thus again using \eqref{mTH.est} for any $\qq \leq 1$ and for any small $\eta>0$ we have
\begin{multline*}
    \left|
 \frac{1}{|k|^2 } 
 \sum_{l=1}^3 
   k_l \int_1^{T} \tTq \widehat{\mathfrak{m}}_{\ddagger l}\partial_t\left(\tq\overline{\widehat{\mathcal{B}}}_l\right)  dt
\right|
\\
\leq
C_\eta\int_1^{T} \tTq\norm{\micro \hat{f}(t,k)}_{L^2_{\pv}}^2dt
+
\eta \int^{T}_{1} \tTq \left|\widehat{{\mathcal{B}}}(t,k)\right|^2 dt 
+
\eta \int^{T}_{1} \tTq \left|\widehat{[\mathcal{A},\mathcal{C}]}(t,k)\right|^2 dt. 
\end{multline*}
Then using \eqref{rhs.def.terms.FT}, \eqref{mTH.est}, and \eqref{linear.op.est.grad} for any small $\eta>0$ we have
\begin{align*}
\left|  \frac{1 }{|k|^2}
\sum_{l=1}^3
k_l \int_1^{T}  t^{4\qq} \overline{\widehat{\mathcal{B}}}_l \widehat{\mathfrak{u}}_{\ddagger l} dt\right|
 \leq &~
\int_1^{T}
 t^{4\qq}\left( \tmq\left|\widehat{\highG}(t,k)\right|+\left| \widehat{\mathfrak{l}}\right|+\left| \widehat{\mathfrak{h}}\right|\right) \left|\widehat{{\mathcal{B}}}(t,k)\right| dt
 \\
  \leq &~
  \eta \int^{T}_{1} \tTq  \left|\widehat{{\mathcal{B}}}(t,k)\right|^2 dt
  +C_\eta\int_1^{T} t^{5\qq}\left| \widehat{\mathfrak{h}}\right|^2dt
  \\
&~
+C_\eta\int_1^{T} \tTq\norm{\micro \hat{f}(t,k)}_{L^2_{\pv}}^2dt.
\end{align*}
We collect all of these estimates for $\widehat{\mathcal{B}}(t,k)$ to conclude for any small $\eta>0$ that 
\begin{multline}\label{FT.B.estimate}
\int_1^{T} \tTq|\widehat{\mathcal{B}}(t,k)|^2  dt
\lesssim
 T^{6\qq}\Vert \hat{f}(T,k)\Vert_{L^2_{\pv}}^2+\norm{ \hat{f}_1(k)}_{L^2_{\pv}}^2
 +
\eta \int^{T}_{1} \tTq \left|\widehat{[\mathcal{A},\mathcal{C}]}(t,k)\right|^2 dt
 \\
+C_\eta\int_1^{T} \tTq\norm{\micro \hat{f}(t,k)}_{L^2_{\pv}}^2dt
+C_\eta\int_1^{T} t^{5\qq}\big|\widehat{\mathfrak{h}}\big|^2dt.
\end{multline}
This is our primary estimate for $\widehat{\mathcal{B}}(t,k)$.

To conclude the proof we multiply \eqref{FT.ACt.estimate} by a suitably small constant $\eta_1>0$ and then we add this to \eqref{FT.B.estimate} choosing $\eta>0$ in \eqref{FT.B.estimate} small enough that $\eta-\eta_1>0$.  Combining \eqref{FT.ACt.estimate} with \eqref{FT.B.estimate} we obtain
\begin{multline*}
    \int^{T}_{1} \tTq \left|\widehat{[\mathcal{A},\mathcal{B},\mathcal{C}]}(t,k)\right|^2 dt
    \lesssim
     T^{6\qq}\Vert \hat{f}(T,k)\Vert_{L^2_{\pv}}^2+\norm{ \hat{f}_1(k)}_{L^2_{\pv}}^2
 \\
+\int_1^{T} \tTq\norm{\micro \hat{f}(t,k)}_{L^2_{\pv}}^2dt
+\int_1^{T} t^{5\qq}\left| \widehat{\mathfrak{h}}\right|^2dt.
\end{multline*}
Next we take the square root in the above inequality and use 
\begin{equation}
\label{sqrt.ineq}
\frac{1}{\sqrt{2}} (A+B)\leq \sqrt{A^2+B^2}\leq A+B,
\end{equation}
further take $\sup_{0\le t\le T}$ for the term $t^{6\qq}\Vert \hat{f}( t,k)\Vert_{L^2_{\pv}}$,
and then integrate the resulting inequality with respect to $d\Sigma(k)$ over $\Z^3$ to obtain the proof of the theorem.
\end{proof}

\begin{remark}
    In the proof of Theorem \ref{thmmacroscop} above, we have used \(\qq\leq 1\) to control the term \(t^{4\qq -1}\) on the right hand side of \eqref{usingqleq1}. 
\end{remark}

\section{Main energy estimate for \texorpdfstring{\(f\)}{}}\label{sectionWienernorm}

In this section we prove Proposition \ref{Torus Boltzmann Micro}, that is, our main energy estimate for the distribution function \(f\). This relies on the coercivity estimate for the linearised operator \(\linL\) in Proposition \ref{lowerN}, and the trilinear estimate for the nonlinearity in Lemma \ref{thm1}.

\begin{proposition}[Estimates for the distribution function]\label{Torus Boltzmann Micro}
Consider \(T>1\), and a suitably regular solution \(f\) of \eqref{rBoltz00} on \([1,T]\times \mathbb{T}^3_x\times \mathbb{R}^3_p\). Then the function \(f\) satisfies the following uniform estimate
\begin{multline}
\norm{\tTq f}_{L^1_kL^\infty_TL^2_{\pv}}
+
{\norm{t^{3\qq/2}\micro f}_{L^1_kL^2_TL^2_{\pv}}}
\\
\lesssim 
\Vert {f}_1\Vert_{L^1_kL^2_{\pv}}
+{\Vert{\tTq f}\Vert_{L^1_kL^\infty_TL^2_{\pv}}} {\norm{t^{3\qq/2} f}_{L^1_kL^2_TL^2_{\pv}}}.
\label{est.fRHS}
\end{multline}
\end{proposition}

The estimate above follows by suitably adapting the energy estimate in \cite[Section 3]{1904.12086} for the nonrelativistic Boltzmann equation on \(\T^3_x\).

\begin{proof}
We take the Fourier transform of \eqref{rBoltz00} to obtain
\begin{align}\label{torus fourier transformed}
\left(\partial_t  + i\frac{\pv}{|\tq \pv|}\cdot k 
-\frac{2\qq }{t} p^i \partial_{p^i} \right) \hat{f}(t,k,\pv)
+\linL\hat{f}(t,k,\pv)=\widehat{\Gamma}(\hat{f},\hat{f})(t,k,\pv),
\end{align}
where $\hat{\Gamma}(\hat{f},\hat{h})(k,\pv)$ is defined in \eqref{def.GaF}.  Take the product of  \eqref{torus fourier transformed} with the complex conjugate of $\hat{f}(t,k,\pv)$, denoted $\bar{\hat{f}}$, and then take the real part of the resulting equation to obtain
\begin{align*}
\frac{1}{2}\left(\frac{d}{dt}   
-\frac{2\qq }{t} p^i \partial_{p^i} \right)
\vert \hat{f}(t,k,\pv)\vert^2
+\real{\bar{\hat{f}}\linL\hat{f} }
=\real{\bar{\hat{f}}\hat{\Gamma}(\hat{f},\hat{f})}.
\end{align*}
Here $\real{z}$ is the real part of a complex number $z$.
Integrating the above with respect to $\pv$ we have
\begin{align*}
    \frac{1}{2}\left(\frac{d}{dt}   
+\frac{6\qq }{t}  \right)\Vert \hat{f}\Vert_{L^2_{\pv}}^2(t,k)+ \real{\irnp{\linL\hat{f}}{\hat{f}}}
= \real{\irnp{\hat{\Gamma}(\hat{f},\hat{f})}{\hat{f}}}.
\end{align*}
From \eqref{timeDeriv.pq} this becomes
\begin{align}\label{eqn.complex}
    \frac{1}{2}\frac{d}{dt}  \left( \tQQq
\Vert \hat{f}\Vert_{L^2_{\pv}}^2(t,k)  \right)+ \tQQq\real{\irnp{\linL\hat{f}}{\hat{f}}}= \tQQq\real{\irnp{\hat{\Gamma}(\hat{f},\hat{f})}{\hat{f}}}.
\end{align}
Then integrating in time on the interval $[1,t]$, we have
\begin{align}\label{pro.tpm.p1}
\frac{1}{2} \tQQq \Vert \hat{f}\Vert_{L^2_{\pv}}^2(t,k)+\int^t_1  \tau^{6\qq}\real{\irnp{\linL\hat{f}}{\hat{f}}}d\tau=&~\frac{1}{2}\Vert \hat{f}_1(k,\cdot)\Vert_{L^2_{\pv}}^2
\\ \notag
&~+\int^t_1 \tau^{6\qq}\real{\irnp{\hat{\Gamma}(\hat{f},\hat{f})}{\hat{f}}}
d\tau.
\end{align}
By the coercivity estimate of $\linL$ in Proposition \ref{lowerN}, there is $\de_0>0$ such that
\begin{equation*}
\de_0\Vert \micro \hat{f} \Vert_{L^2_{\pv}}^2(t,k)\le \tTq\real{\irnp{\linL\hat{f}}{\hat{f}}}(t,k).
\end{equation*}
Thus it follows from \eqref{pro.tpm.p1} that
\begin{align}\notag
\frac{1}{2} \tQQq \Vert \hat{f}\Vert_{L^2_{\pv}}^2(t,k)+\de_0\int^t_1  \tau^{3\qq}\Vert \micro \hat{f} \Vert_{L^2_{\pv}}^2(\tau,k)d\tau
\leq &~\frac{1}{2}\Vert \hat{f}_1\Vert_{L^2_{\pv}}^2(k)
\\ \notag
&~+\int^t_1 \tau^{6\qq}\real{\irnp{\hat{\Gamma}(\hat{f},\hat{f})}{\hat{f}}}
d\tau.
\end{align}
Take the square root on both sides, use \eqref{sqrt.ineq}, and then we further have
\begin{multline*}
  \frac{1}{\sqrt{2}} \tTq \Vert \hat{f}\Vert_{L^2_{\pv}}(t,k)
+
\sqrt{\de_0}
\left(\int^t_1  \tau^{3\qq}\Vert \micro \hat{f} \Vert_{L^2_{\pv}}^2(\tau,k)d\tau \right)^{\frac12}
\leq \frac{1}{2}\Vert \hat{f}_1\Vert_{L^2_{\pv}}(k)
\\ 
+
\sqrt{2}\left(\int^t_1 \tau^{6\qq} \left|\real{\irnp{\hat{\Gamma}(\hat{f},\hat{f})}{\hat{f}}}\right|
d\tau \right)^{\frac12}.  
\end{multline*}
We have derived, for all $0\leq t\leq T$ and any $k\in \Z^3$, that 
\begin{align}\notag
  \tTq \Vert \hat{f}\Vert_{L^2_{\pv}}(t,k)
&+
\left(\int^t_1  \tau^{3\qq}\Vert \micro \hat{f} \Vert_{L^2_{\pv}}^2(\tau,k)d\tau \right)^{\frac12}
\\
&\le C_0\left\{ \Vert \hat{f}_1\Vert_{L^2_{\pv}}(k) 
+
\left(\int^t_1 \tau^{6\qq} \left|{\irnp{\hat{\Gamma}(\hat{f},\hat{f})}{\hat{f}}}\right|
d\tau \right)^{\frac12}
\right\},
\label{v-int}
\end{align}
with $C_0=\max\{2,\sqrt{2/\de_0}\}>0$.  By the following well-known identity,
\begin{align*}
\langle  \Gamma(f,f),h\rangle
=
\langle  \Gamma(f,f), \micro h\rangle,
\end{align*}
it follows directly that
\begin{equation}\notag 
\irnp{\hat{\Gamma}(\hat{f},\hat{f})}{\hat{f}}=\irnp{\hat{\Gamma}(\hat{f},\hat{f})}{\micro\hat{f}}.
\end{equation}
Then, taking $\sup_{0\le t\le T}$ on both sides of \eqref{v-int} and integrating the resulting inequality with respect to $d\Sigma(k)$ over $\Z^3$, we have
\begin{align}
&\int_{\Z^3} \sup_{0\le t\le T} \left(\tTq \Vert \hat{f}(t,k)\Vert_{L^2_{\pv}} \right) d\Sigma(k) + 
\int_{\Z^3} \left(\int^T_1  \tau^{3\qq}\Vert \micro f (\tau,k) \Vert_{L^2_{\pv}}^2d\tau \right)^{\frac12}
d\Sigma(k)
\notag\\
&\le C_0\left\{ \Vert \hat{f}_1\Vert_{L^1_kL^2_{\pv}}
+
\int_{\Z^3} 
\left(\int^T_1 \tau^{6\qq} \left|{\irnp{\hat{\Gamma}(\hat{f},\hat{f})}{\micro\hat{f}}}\right|
d\tau \right)^{\frac12}
d\Sigma(k)\right\}.
\label{pro.tpm.p2}
\end{align}
We can estimate the last term on the right-hand side of \eqref{pro.tpm.p2} as
\begin{multline}
\int_{\Z^3} \left(\int^T_1 \tQQq\left|\irnp{\hat{\Gamma}(\hat{f},\hat{f})}{{\micro\hat{f}}}\right|dt\right)^{\frac12}d\Sigma(k)\\
\lesssim \int_{\Z^3}\left(\int^T_1 \tQQq \norm{{\micro\hat{f}}(k)}_{L^2_{\pv}}\int_{\Z^3}\norm{\hat{f}(k-l)}_{L^2_{\pv}}\norm{\hat{f}(l)}_{L^2_{\pv}} d\Sigma(l)dt\right)^{\frac12}d\Sigma(k).
\label{pro.tpm.p3}
\end{multline}
Then \eqref{pro.tpm.p3} is a direct consequence of Lemma \ref{lem: Boltzmann nonlinear} that is stated and proven below.

We now estimate the upper bound in \eqref{pro.tpm.p3}. Applying the Cauchy--Schwarz inequality to $\int_{1}^{T}(\cdot)dt$, we obtain
\begin{align}
&\mathcal{D}\eqdef\int_{\Z^3_k}\left(\int^{T}_{1} \tQQq\norm{\micro\hat{f}(t,k)}_{L^2_{\pv}}\int_{\Z^3_l}\Vert \hat{f}(t,k-l)\Vert_{L^2_{\pv}} \norm{\hat{f}(t,l)}_{L^2_{\pv}}  d\Sigma(l)dt\right)^{\frac12}d\Sigma(k)
\notag
\\
&\leq \int_{\Z^3_k} \left(\int^{T}_{1} \tQQq \tTq\left(\int_{\Z^3_l} \Vert \hat{f}(t,k-l)\Vert_{L^2_{\pv}} \norm{\hat{f}(t,l)}_{L^2_{\pv}}d\Sigma(l)\right)^2dt\right)^{1/4}
\notag
\\
&\qquad\qquad\qquad\times \left(\int^{T}_{1} \tTq\norm{\micro\hat{f}(t,k)}_{L^2_{\pv}}^2 dt\right)^{1/4} d\Sigma(k),
\notag
\end{align}
and then  use Young's inequality with  small constant $\eta>0$ to observe that
\begin{align}
&\mathcal{D}\leq \eta\int_{\Z^3_k}\left(\int^{T}_{1} \tTq\norm{\micro\hat{f}(t,k)}_{L^2_{\pv}}^2 dt\right)^{\frac12}d\Sigma(k)
\label{pro.tpm.p4}
\\ \notag
&\qquad+\frac{1}{4\eta} \int_{\Z^3_k}\left(\int^{T}_{1} \tQQq \tTq\left(\int_{\Z^3_l} \Vert \hat{f}(t,k-l)\Vert_{L^2_{\pv}} \norm{\hat{f}(t,l)}_{L^2_{\pv}}d\Sigma(l)\right)^2dt\right)^{\frac12} d\Sigma(k).
\end{align}
We now recall the following version of Minkowski's inequality 
\begin{equation}\label{minkowski.ineq}
\big\|\|\cdot\|_{L^1_l}\big\|_{L^2_t}\leq \big\|\|\cdot\|_{L^2_t}\big\|_{L^1_l}.
\end{equation}
For the second term on the right-hand side of \eqref{pro.tpm.p4},
we use \eqref{minkowski.ineq} to obtain
\begin{multline*}
\left(\int^{T}_{1} \tQQq \tTq\left(\int_{\Z^3_l} \Vert \hat{f}(t,k-l)\Vert_{L^2_{\pv}} \norm{\hat{f}(t,l)}_{L^2_{\pv}}d\Sigma(l)\right)^2dt\right)^{\frac12}
\\
\leq
\int_{\Z^3_l} \left(\int^{T}_{1} \tQQq \Vert \hat{f}(t,k-l)\Vert_{L^2_{\pv}}^2 \tTq\norm{\hat{f}(t,l)}_{L^2_{\pv}}^2 dt\right)^{\frac12}d\Sigma(l)
\\
\leq 
\int_{\Z^3_k} \int_{\Z^3_l} \sup_{0\le \tau\le T} \left( \tau^{3\qq}\Vert \hat{f}(\tau,k-l) \Vert_{L^2_{\pv}}\right) \left(\int^{T}_{1} \tTq\norm{\hat{f}(t,l)}_{L^2_{\pv}}^2 dt\right)^{\frac12}d\Sigma(l) d\Sigma(k).
\end{multline*}
Then Fubini's theorem and translation invariance imply that
\begin{multline*}
    \int_{\Z^3_k} \int_{\Z^3_l} \sup_{0\le \tau\le T} \left( \tau^{3\qq}\Vert \hat{f}(\tau,k-l) \Vert_{L^2_{\pv}}\right) \left(\int^{T}_{1} \tTq\norm{\hat{f}(t,l)}_{L^2_{\pv}}^2 dt\right)^{\frac12}d\Sigma(l) d\Sigma(k)
\\
= \int_{\Z^3_l}  d\Sigma(l)\left(\int^{T}_{1} \tTq\norm{\hat{f}(t,l)}_{L^2_{\pv}}^2 dt\right)^{\frac12} \int_{\Z^3_k}d\Sigma(k) \sup_{0\le \tau\le T} \left( \tau^{3\qq}\Vert \hat{f}(\tau,k-l) \Vert_{L^2_{\pv}}\right) 
\\
=\Vert \tTq{f}\Vert_{L^1_kL^\infty_TL^2_{\pv}}
\Vert t^{3\qq/2} {f}\Vert_{L^1_kL^2_TL^2_{\pv}}.
\end{multline*}
Applying the estimates above to the second term on the right side of \eqref{pro.tpm.p4} and further using \eqref{pro.tpm.p2} and \eqref{pro.tpm.p3} leads to the desired estimate \eqref{est.fRHS} in the statement of the proposition.
\end{proof}

We now state and prove the following lemma.

\begin{lemma}\label{lem: Boltzmann nonlinear}
For all \(f\), \(h\), and \(\eta\), the following uniform estimate holds,
\begin{align}\label{ineq: Boltzmann nonlinear}
\left| \irnp{\widehat{\Gamma}(\hat{f},\hat{h})(k)}{\hat{\eta}(k)}\right|\lesssim \norm{\hat{\eta}(k)}_{L^2_{\pv}}\int_{\Z^3}\Vert \hat{f}(k-l)\Vert_{L^2_{\pv}} \norm{\hat{h}(l)}_{L^2_{\pv}}  \,d\Sigma(l).
\end{align}
\end{lemma}

Notice that \eqref{ineq: Boltzmann nonlinear} directly implies \eqref{pro.tpm.p3}.

\begin{proof}
By the definition \eqref{def.GaF} of $\hat{\Ga} (\cdot,\cdot)$ and Fubini's theorem, we obtain
\begin{align}\label{fourier.convolution.nonlin}
\irnp{\hat{\Gamma}(\hat{f},\hat{h})(k)}{\hat{\eta}(k)}
=
\int_{\Z^3} \irnp{ \Gamma(\hat{f}(k-l),\hat{h}(l))}{\hat{\eta}(k)}  d\Sigma(l).
\end{align}
This calculation is performed in detail in \cite[Lemma 3.2]{1904.12086}.  In our case, we have
\begin{multline*}
    \irnp{\hat{\Gamma}(\hat{f},\hat{h})(k)}{\hat{\eta}(k)}
\\
=
t^{3\qq}\int_{\mathbb{R}^3}   d\pv \int_{\mathbb{R}^3}   dq 
\int_{\mathbb{S}^{2}}  d\omega
~ v_{\Emptyset}  
 \sqrt{J(\tqq\qv)}
\left\{[\hat{f}(\pv')*\hat{h}(\qv')](k)-[\hat{f}(\pv)*\hat{h}(\qv)](k)\right\}\bar{\hat{\eta}}(\pv,k) 
\\
=
t^{3\qq}\int_{\mathbb{R}^3} ~  d\pv \int_{\mathbb{R}^3} ~  dq 
\int_{\mathbb{S}^{2}} ~ d\omega
~ v_{\Emptyset} ~  
 \sqrt{J(\tqq\qv)}
\\
\times\int_{\Z^3} \left\{\hat{f}(k-l,\pv')\hat{h}(l,\qv')-\hat{f}(k-l,\pv)\hat{h}(l,\qv)\right\} \bar{\hat{\eta}}(\pv,k) d\Sigma(l) 
\\
=
\int_{\Z^3} \irnp{\Gamma(\hat{f}(k-l),\hat{h}(l))}{\hat{\eta}(k)}  d\Sigma(l).
\end{multline*}
We conclude
\begin{equation}
\label{lem.ptri.p1}
\left\vert \irnp{\hat{\Gamma}(\hat{f},\hat{h})(k)}{\hat{\eta}(k)}\right\vert\le \int_{\Z^3} \left\vert \irnp{\Gamma(\hat{f}(k-l),\hat{h}(l))}{\hat{\eta}(k)}\right\vert d\Sigma(l).
\end{equation}
Then, by Lemma \ref{thm1} we have 
\begin{equation}\label{ad.vwte}
\left|\irnp{\Gamma(\hat{f}(k-l),\hat{h}(l))}{\hat{\eta}(k)} \right| 
\lesssim \Vert \hat{f}(k-l)\Vert_{L^2_{\pv}} \norm{\hat{h}(l)}_{L^2_{\pv}}  \norm{\hat{\eta}(k)}_{L^2_{\pv}}.
\end{equation}
With \eqref{ad.vwte} above, the desired estimate \eqref{ineq: Boltzmann nonlinear} follows from \eqref{lem.ptri.p1}.  \end{proof}

\section{ Global stability of the Maxwell--J\"uttner equilibrium}\label{sectproofsmainresults}

In this section, we will prove Theorem \ref{Torus Existence}. In Section \ref{subsecexistuniq}, we first show the local well-posedness result Theorem \ref{Torus Existence} is based on. In Section \ref{subsglbalexistence}, we show global existence and the uniform decay estimate \eqref{subsexpnonexp} in Theorem \ref{Torus Existence}. In Section \ref{subfastdecay}, we prove the faster decay estimates in the case when \(\qq\in (0,1/3]\). In Section \ref{subsexpnonexp}, we prove exponential decay in the nonexpanding case \(\qq=0\). In Section \ref{rapid.decay.solution}, we prove time-decay for weighted energy norms. Finally, in Section \ref{sec.spatial.regularity.solution}, we prove propagation of spatial regularity.

\subsection{Local existence and uniqueness around the equilibrium}\label{subsecexistuniq}

We prove here local existence and uniqueness of solutions for the massless Boltzmann equation on FLRW spacetime \eqref{rBoltz00} around the Maxwell--J\"uttner equilibrium.

\begin{theorem}[Local existence and uniqueness around the equilibrium]\label{local.thm}
Let \(\qq\in [0,1]\). Under the assumptions of Theorem \ref{Torus Existence}, there is a small $\epsilon_0>0$, a $T_0=T_0(\epsilon_0)>0$ and some  $C_0=C_0(T_0)>0$ such that if 
\begin{equation}\label{initial.small.local}
\| f_{1}\|_{L_{k}^1L^2_{\pv}}\leq\epsilon_0,
\end{equation}
then the massless Boltzmann equation on FLRW spacetime around the equilibrium, \eqref{rBoltz00},
 admits a unique local-in-time solution
$$
f(t,x,\pv),~~ \ 0\leq t\leq T_0<\infty, ~~\ x\in  \T^3,~~\ \pv\in \R^3,
$$
which satisfies the  uniform estimate
\begin{equation}\label{loces}
\| \tTq f\|_{L^1_{{k}}L^\infty_{T_0}L^2_{\pv}}+\| t^{3\qq/2} f\|_{L^1_{{k}}L^2_{T_0}L^2_{\pv}}
\leq C_{0}
\|f_1\|_{L_{{k}}^1L^2_{\pv}}.
\end{equation}
If additionally 
$F_1(x,\pv)=J(\pv)+J^{\frac{1}{2}}(\pv)f_1(x,\pv)\geq0$,
then the solution satisfies
$$
F(t,{x},\pv)=J(t^{2\qq}p) +\sqrt{J(t^{2\qq}p) }f(t,{x},\pv)\geq0.
$$
\end{theorem}

To prove Theorem \ref{local.thm}, we start from the following linear inhomogeneous problem
\begin{equation}\label{lc.lb}
\begin{split}
        \partial_t h+ \frac{\pv^i}{|\tq \pv|}\partial_{x^i} h
    -\frac{2\qq }{t} p^i \partial_{p^i} h
    + t^{-3\qq} \nu_0 h
    -\Gamma(h,\hh)
= &~
\compK (\hh),
\\ 
h(t=1,x,\pv)
= &~
h_1(x,\pv),
\end{split}
\end{equation}
where $\hh=\hh(t,x,\pv)$ is a given function.  The solvability of \eqref{lc.lb} is guaranteed by the following lemma.

\begin{lemma}\label{lb.loc}
Let \(h_{1}\in L^1_{k}L^2_{\pv}\) be an initial datum. We suppose that for some $\epsilon_0>0$, $T_0>0$ and $T\in (0,T_0]$ the given function ${\hh}$ satisfies
\begin{equation}
\label{lb.loc.a2}
\| \tTq {\hh}\|_{ L^1_{k}L^\infty_{T}L^2_{\pv}}+\|t^{3\qq/2}{\hh}\|_{L^1_{k}L^2_{T}L^2_{\pv}}\leq \epsilon_0.
\end{equation}
Then, for $\epsilon_0>0$ in \eqref{lb.loc.a2} small enough, there exists $T_0>0$ and $C_1=C_1(T_0)>0$ such that for any $T\in (0,T_0]$ the initial value problem \eqref{lc.lb} admits a unique weak solution $h(t,x,\pv)$ for $0\leq t\leq T$, $x\in  \T^3$ and $\pv\in \R^3$ satisfying
\begin{align}\notag 
    \| \tTq h\|_{L^1_{k}L^\infty_{T}L^2_{\pv}}+\| t^{3\qq/2}  h\|_{ L^1_{k}L^2_{T}L^2_{\pv}}
\leq C_0\left(\|h_{1}\|_{L^1_{k}L^2_{\pv}}+\sqrt{T}\|{\hh}\|_{ L^1_{k}L^2_{T}L^2_{\pv}}\right).
\end{align}
\end{lemma}

Regarding the local-in-time existence and uniqueness of solutions to the linear equation \eqref{lc.lb}, the proof of Lemma \ref{lb.loc} is completely standard and omitted.

\begin{proof}[Proof of Theorem \ref{local.thm}.]  We will first build a sequence of approximate solutions denoted by $(f^n(t,x,\pv))_{n=0}^\infty$ for the problem
\eqref{rBoltz00} with initial data \eqref{idff}.  For $n=0,1,2,\ldots$, we use the following iterative scheme 
\begin{align}\notag
        \partial_t f^{n+1}+ \frac{\pv^i}{|\tq \pv|}\partial_{x^i} f^{n+1}
    -\frac{2\qq }{t} p^i \partial_{p^i} f^{n+1}
    + t^{-3\qq} \nu_0 f^{n+1}
= &~
\Gamma(f^{n+1},f^{n})
+
\compK (f^{n}),
\\  \label{approx.solution}
f^{0}(t,x,\pv)
= 
f_1(x,\pv), \qquad 
f^{n+1}(t=1,x,\pv)
= &~
f_1(x,\pv).
\end{align}
Then, by using Lemma \ref{lb.loc}, and the a priori estimates in \secref{sectcollisionop}, it is a standard  procedure to apply an induction argument to show that there are $\epsilon_0>0$ and $T_0>0$ such that if \eqref{initial.small.local} holds then the approximation sequence $(f^n(t,x,\pv))_{n=0}^\infty$ is well-defined in the Banach space $L^1_{k}L^\infty_{T}L^2_{\pv}\cap L^1_{k}L^2_{T}L^2_{\pv}$ for any $1<T\leq T_0$.    

 We will prove that for a short time, $T>1$, the iterated solutions to \eqref{approx.solution} satisfy the following uniform estimate
    \begin{equation}\label{induction hypo}
        \sup_{n\ge 0}
        \left( \| \tTq f^{n}\|_{L^1_{{k}}L^\infty_{T}L^2_{\pv}}+\| t^{3\qq/2} f^{n}\|_{L^1_{{k}}L^2_{T}L^2_{\pv}}\right)
        \le 4C_1\|f_1\|_{L^1_{{k}}L^2_{\pv}},
\end{equation}
where $C_1>0$ is a uniform constant.

Fix an integer $m\ge 1$ and suppose that for some small $T-1> 0$, that
\eqref{induction hypo} holds for 
all $n\in\{0,1,2,3, \ldots,m\}$ for some constant $C_1>0$ that is uniform in $m$.  Then it is sufficient to prove that \eqref{induction hypo} holds for $n=m+1$ by induction.  We mention that \eqref{induction hypo} holds for $n=0$ with $f^0=f_1$ as defined in \eqref{approx.solution}.

Then, for the iterative problem \eqref{approx.solution}, following the proof of Proposition \ref{Torus Boltzmann Micro}, 
\begin{multline*}
    \norm{\tTq f^{n+1}}_{L^1_kL^\infty_TL^2_{\pv}}
+
{\norm{t^{3\qq/2} f^{n+1}}_{L^1_kL^2_TL^2_{\pv}}}
\\
\leq
C_1\Vert {f}_1\Vert_{L^1_kL^2_{\pv}}
+
C_2{\Vert{\tTq f^{n+1}}\Vert_{L^1_kL^\infty_TL^2_{\pv}}} {\norm{t^{3\qq/2} f^{n}}_{L^1_kL^2_TL^2_{\pv}}}
\\
+
C_2
\int_{\Z^3_k} \left(\int^T_1 \tau^{6\qq} \left|{\irnp{\compK(\hat{f}^{n})}{\hat{f}^{n+1}}}(\tau, k)\right|
d\tau \right)^{\frac12} d\Sigma(k).
\end{multline*}
Above $C_2>0$ and the constant $C_1>0$ in \eqref{induction hypo} is given by
\begin{equation*}
    C_1\eqdef 2\max\big\{2,\sqrt{2/\nu_0}\big\}>0.
\end{equation*}
For the last term, from Proposition \ref{noKestimate}, we have 
\begin{equation*}
     \tau^{6\qq} \left|\irnp{\compK (\widehat{f}^{n})}{\widehat{f}^{n+1}}\right|
\lesssim
\tau^{3\qq}
\norm{\widehat{f}^{n}}_{L^2_{\pv}}(\tau)
\norm{\widehat{f}^{n+1} }_{L^2_{\pv}}(\tau).
\end{equation*}
We conclude that 
\begin{equation*}
        \int^{T}_1 \tau^{6\qq} 
\left|\irnp{\compK (\widehat{f}^{n})}{\widehat{f}^{n+1}}(\tau, k)\right|
d\tau
\lesssim
(T-1) \norm{ \widehat{f}^{n}(k) }_{L^\infty_{T}L^2_{\pv}}
\norm{t^{3\qq} \widehat{f}^{n+1}(k)}_{L^\infty_{T}L^2_{\pv}}.
\end{equation*}
Thus for any small $\delta>0$ we have 
\begin{multline*}
\int_{\Z^3_k} \left(\int^T_1 \tau^{6\qq} \left|{\irnp{\compK(\hat{f}^{n})}{\hat{f}^{n+1}}}(\tau, k)\right|
d\tau \right)^{\frac12} d\Sigma(k)
\\
\leq
C (T-1)^{\frac{1}{2}} \norm{ {f}^{n} }_{L^1_k L^\infty_{T}L^2_{\pv}}^{\frac{1}{2}}
\norm{t^{3\qq} {f}^{n+1}}_{L^1_k L^\infty_{T}L^2_{\pv}}^{\frac{1}{2}}
\\
\leq
 C_\delta (T-1) \norm{ {f}^{n} }_{L^1_k L^\infty_{T}L^2_{\pv}}
 +
 \delta
\norm{t^{3\qq} {f}^{n+1}}_{L^1_k L^\infty_{T}L^2_{\pv}}.
\end{multline*}
Thus for $\delta>0$ small enough, and some $C_3>0$, we have
\begin{multline*}
    \norm{\tTq f^{n+1}}_{L^1_kL^\infty_TL^2_{\pv}}
+
{\norm{t^{3\qq/2} f^{n+1}}_{L^1_kL^2_TL^2_{\pv}}}
\leq
2C_1\Vert {f}_1\Vert_{L^1_kL^2_{\pv}}
\\
+
C_3{\Vert{\tTq f^{n+1}}\Vert_{L^1_kL^\infty_TL^2_{\pv}}} {\norm{t^{3\qq/2} f^{n}}_{L^1_kL^2_TL^2_{\pv}}}
+
C_3
(T-1) \norm{ {f}^{n} }_{L^1_k L^\infty_{T}L^2_{\pv}}.
\end{multline*}
Then by the induction assumption, \eqref{induction hypo} for $n=m$, we have 
\begin{multline*}
    \norm{\tTq f^{n+1}}_{L^1_kL^\infty_TL^2_{\pv}}
+
{\norm{t^{3\qq/2} f^{n+1}}_{L^1_kL^2_TL^2_{\pv}}}
\leq 
2C_1\Vert {f}_1\Vert_{L^1_kL^2_{\pv}}
\\
+
2C_1 C_3 
{\Vert{\tTq f^{n+1}}\Vert_{L^1_kL^\infty_TL^2_{\pv}}} 
\|f_1\|_{L^1_{{k}}L^2_{\pv}}
+
2C_1C_3 
(T-1) 
\|f_1\|_{L^1_{{k}}L^2_{\pv}}.
\end{multline*}
Thus, we obtain 
\begin{equation*}
        \norm{\tTq f^{n+1}}_{L^1_kL^\infty_TL^2_{\pv}}
+
{\norm{t^{3\qq/2} f^{n+1}}_{L^1_kL^2_TL^2_{\pv}}}
\leq
C_4\Vert {f}_1\Vert_{L^1_kL^2_{\pv}},
\end{equation*}
where
\begin{equation*}
    C_4 
    =
  \frac{2C_1 + 2C_1C_3 
(T-1) }{1-2C_1 C_3  \|f_1\|_{L^1_{{k}}L^2_{\pv}}}.
\end{equation*}
Then for $T-1>0$ and $\|f_1\|_{L^1_{{k}}L^2_{\pv}}$ sufficiently small, we observe that $C_4  \leq 4 C_1$.  This proves \eqref{induction hypo} by induction.  By a similar standard procedure we can also show that $(f^n(t,x,\pv))_{n=0}^\infty$ is a Cauchy sequence in this space.  Then the limit function $f(t,x,\pv)$  is a weak solution of  \eqref{rBoltz00} satisfying the estimate \eqref{loces}.

We now prove the uniqueness of the local-in-time solution.  Suppose alternatively that we have two solutions $f$ and $h$ to \eqref{rBoltz00}, both satisfying the local-in-time estimate \eqref{loces} for $1 \leq t \leq T_0$ with initial data $f_1$ and $h_1$, respectively.  Then the difference of the two solutions $f-h$ satisfies
\begin{multline*}
               \left( \partial_t + \frac{\pv^i}{|\tq \pv|}\partial_{x^i} 
    -\frac{2\qq }{t} p^i \partial_{p^i} \right) \left( f-h\right)
    +
    \linL \left( f-h\right)
=
\Gamma (f,f)-\Gamma (h,h)
\\
=
\Gamma (f,f-h)+\Gamma (f-h,h).
\end{multline*}
We can again apply the same estimates as in the proof of Proposition \ref{Torus Boltzmann Micro} to obtain
\begin{multline}\notag
\norm{\tTq (f-h)}_{L^1_kL^\infty_{T_0}L^2_{\pv}}
+
{\norm{t^{3\qq/2}\micro (f-h)}_{L^1_kL^2_{T_0}L^2_{\pv}}}
\\
\lesssim 
\Vert {f}_1-{h}_1\Vert_{L^1_kL^2_{\pv}}
+
{\Vert{\tTq (f-h)}\Vert_{L^1_kL^\infty_{T_0}L^2_{\pv}}}
\left( {\norm{t^{3\qq/2} f}_{L^1_kL^2_{T_0}L^2_{\pv}}} +{\norm{t^{3\qq/2} h}_{L^1_kL^2_{T_0}L^2_{\pv}}}\right).
\end{multline}
Since $f$ and $h$ satisfy \eqref{loces} with $\epsilon_0>0$ sufficiently small we obtain
\begin{equation*}
    \norm{\tTq (f-h)}_{L^1_kL^\infty_TL^2_{\pv}}
+
{\norm{t^{3\qq/2}\micro (f-h)}_{L^1_kL^2_TL^2_{\pv}}}
\lesssim 
\Vert {f}_1-{h}_1\Vert_{L^1_kL^2_{\pv}}.
\end{equation*}
The uniqueness of the local-in-time solution follows directly.

We now explain how to apply the  proof of positivity, as in \cite{Ukai-86}, to the massless Boltzmann equation on FLRW in \eqref{rBoltz.eqn0}.  For this, we use the  alternative approximating formula
\begin{equation}
    \partial_t F^{n+1}+\frac{\pv^i}{|t^\qq\pv|} \partial_{x^i} F^{n+1}
    - \frac{2 \qq}{t} p^i \partial_{p^i} F^{n+1} 
    + R(F^{n})F^{n+1}
    =Q_+(F^{n},F^{n}),  
\label{rBoltz.eqn0.approx}
\end{equation}
for $n\in\{0,1, 2, \ldots\}$
with the initial conditions 
\begin{equation*}
   \left. F^{n+1}\right|_{t=1} =  F_1(x,\pv)= J(\pv) + \sqrt{J(\pv)} f_1(x,\pv) \geq 0.
\end{equation*}
We also take
\begin{equation}\label{first.step}
 F^{0}(t,x,\pv)\eqdef J(\tqq\pv) \geq 0, \quad t \geq 1.
\end{equation}
In \eqref{rBoltz.eqn0.approx} we have used the standard decomposition of 
$
Q
=
Q_+ - Q_-
$
into gain and loss terms where from \eqref{collisionCMc} we have 
\begin{equation}\notag
\begin{split}
Q_+(F_1,F_2)
= &~ 
\tTq\int_{\mathbb{R}^3}  {dq}  ~
\int_{\mathbb{S}^{2}} d\omega 
~ v_{\Emptyset}   ~ F_1(p^{\prime })F_2(q^{\prime}),
\\
Q_-(F_1,F_2)
= &~ 
\tTq\int_{\mathbb{R}^3}  {dq}  ~
\int_{\mathbb{S}^{2}} d\omega 
~ v_{\Emptyset} ~ F_1(p)F_2(q) = R(F_2)F_1(p),
\end{split}
\end{equation}
where
\begin{align*}
    R(F_2) \eqdef Q_-(1,F_2)
    =
    \tTq\int_{\mathbb{R}^3}  {dq}  ~
\int_{\mathbb{S}^{2}} d\omega 
~ v_{\Emptyset} ~ F_2(q).
\end{align*}
We will also use the notation $Q_+^{n} \eqdef Q_+(F^{n},F^{n})$.

The characteristics of \eqref{rBoltz.eqn0} and \eqref{massless.rBoltz} for $s \geq 1$ are given by 
\begin{equation}\notag
\begin{split}
\frac{d {X}^j}{ds}
= &~ 
s^{-\qq}\frac{{P}^j}{| {P}|}, \qquad {X}^j(s=1,x,\pv)=x^j,
\\
\frac{d {P}^j}{ds}
= &~ 
-\frac{2\qq }{s} P^j, \qquad {P}^j(s=1,\pv)=\pv^j.
\end{split}
\end{equation}
For $0 \leq \qq < 1$ we define
\begin{equation}\label{tiq.weight}
    \tiq(t) \eqdef \frac{t^{1-\qq}-1}{1-\qq}, \qquad \tiq(1)=0.
\end{equation}
For $\qq = 1$ we define 
\begin{equation}\label{weightflowq1}
    \mathcal{S}_1(t) \eqdef 
    \log(t), \qquad \mathcal{S}_1(1)=0.
\end{equation}
These characteristics are solved as
\begin{equation} \label{characterisitics.vacuum}
\begin{split}
{X}^j(t,x,\pv)
= &~ 
x^j+\tiq(t)\frac{\pv^j}{ |\pv|},
\\
{P}^j(t,\pv)
= &~ 
{t^{-2\qq }}p^j.
\end{split}
\end{equation}
We will use \eqref{characterisitics.vacuum} to express solutions to \eqref{rBoltz.eqn0.approx} as an ordinary differential equation.

We first consider the linear model of \eqref{massless.rBoltz} with a source term ${\sourceS}$ as
\begin{equation}\label{massless.rBoltz.collisionless}
    \partial_t F+ \frac{\pv^i}{|\tq \pv|}\partial_{x^i} F
    -\frac{2\qq }{t} p^i \partial_{p^i} F = {\sourceS}(t,x,\pv), \quad 
F(t=1, x, p) = F_1(x,p).
\end{equation}
Then with \eqref{characterisitics.vacuum} solutions to \eqref{massless.rBoltz.collisionless} satisfy
\begin{align*}
    \frac{d}{dt } \left(F(t,{X}(t,x,\pv),{P}(t,\pv)) \right)
    = {\sourceS}(t,{X}(t,x,\pv),{P}(t,\pv)).
\end{align*}
Integrating over $[1,t]$ we obtain
 \begin{align*}
     F(t,{X}(t),{P}(t)) 
     = &~ 
     F_1({X}(1),{P}(1)) + \int_1^t d\tau~ {\sourceS}(\tau,{X}(\tau),{P}(\tau)).
 \end{align*}
 Equivalently,
\begin{align*}
    F\left(t,x+\tiq(t)\frac{\pv}{ |\pv|},{t^{-2\qq }} \pv \right) 
    = &~ 
    F_1(x,p)
    \\
    &~ 
    + \int_1^t d\tau~ {\sourceS}
    \left(\tau,x+\tiq(\tau)\frac{\pv}{ |\pv|},{\tau^{-2\qq }} \pv \right). 
\end{align*}
Next, setting $\qv = {t^{-2\qq }} \pv$ and $y=x+\tiq(t)\frac{\pv}{ |\pv|}=x+\tiq(t)\frac{\qv}{ |\qv|}$, we equivalently have
\begin{align*}
    F\left(t,y,\qv\right) 
    = &~ 
    F_1\left(t,y-\tiq(t)\frac{\qv}{ |\qv|}, \tqq\qv \right)
    \\
    &~ 
    + \int_1^t d\tau~ {\sourceS}
    \left(\tau,y+\left(\tiq(\tau)-\tiq(t) \right)\frac{\qv}{ |\qv|},{(t/\tau)^{2\qq }} \qv \right). 
\end{align*}
Thus, for $t \geq \tau \geq 1$, we define the backward characteristics of \eqref{characterisitics.vacuum} by 
\begin{equation} \label{characterisitics.vacuum.backward}
\begin{split}
\widetilde{X}^j(t,\tau,x,\pv)
= &~ 
x^j+\left(\tiq(\tau)-\tiq(t) \right)\frac{\pv^j}{ |\pv|},
\\
\widetilde{P}^j(t,\tau,\pv)
= &~ 
{(t/\tau)^{2\qq }}\pv^j.
\end{split}
\end{equation}
Now, by using \eqref{characterisitics.vacuum.backward}, we define the operator $U(t,\tau)$ for $t \geq \tau \geq 1$ by 
\begin{equation*}
    (U(t,\tau)G)(s,x,\pv) \eqdef G(s,\widetilde{X}(t,\tau,x,\pv),\widetilde{P}(t,\tau,\pv)), \quad s \ge 1.
\end{equation*}
Then we can write the solution to \eqref{massless.rBoltz.collisionless} as
\begin{align*}
    F(t,x,\pv) 
    =  
    (U(t,1)F_1)(x,\pv) 
    + 
    \int_1^t d\tau~ 
    (U(t,\tau){\sourceS})(\tau,x,\pv).
\end{align*}
Next we will apply this calculation to \eqref{rBoltz.eqn0.approx}.

We use the source ${\sourceS} = Q_+(F^{n},F^{n})=Q_+^{n}$ in \eqref{massless.rBoltz.collisionless} to write \eqref{rBoltz.eqn0.approx} as  
\begin{align*}
    \frac{d}{dt } \left(F^{n+1}(t,{X}(t),{P}(t)) \right)
    + R(F^{n})F^{n+1}(t,{X}(t),{P}(t))
    = {Q_+^{n}}(t,{X}(t),{P}(t)).
\end{align*}
Further with \eqref{characterisitics.vacuum} we define
\begin{align*}
       \Phi^{n}(t,s, x,\pv) \eqdef &~ \int_s^t d\tau~
    R(F^{n})(\tau,{X}(\tau),{P}(\tau)), \quad \Phi^{n}(t) = \Phi^{n}(t,1, x,\pv).
\end{align*}
Then we can write \eqref{rBoltz.eqn0.approx} as 
\begin{align*}
    \frac{d}{dt } \left(e^{\Phi^{n}(t)}F^{n+1}(t,{X}(t),{P}(t)) \right)
    = e^{\Phi^{n}(t)}{Q_+^{n}}(t,{X}(t),{P}(t)).
\end{align*}We integrate over $[1,t]$ to obtain
\begin{align*}
    F^{n+1}\left(t,x+\tiq(t)\frac{\pv}{ |\pv|},{t^{-2\qq }} \pv \right) 
    = &~ 
   e^{-\Phi^{n}(t,1, x,\pv)}  F_1(x,p)
    \\
    &~ 
    + \int_1^t d\tau~ e^{-\Phi^{n}(t,\tau, x,\pv)} {Q_+^{n}}
    \left(\tau,x+\tiq(\tau)\frac{\pv}{ |\pv|},{\tau^{-2\qq }} \pv \right). 
\end{align*}
Above for $t\geq \tau \geq 1$, we have used that
\begin{equation*}
    \Phi^{n}(\tau)-\Phi^{n}(t) 
    =
     \Phi^{n}(\tau,1, x,\pv)-\Phi^{n}(t,1, x,\pv)
     =
     -\Phi^{n}(t,\tau, x,\pv).
\end{equation*}
Further with \eqref{characterisitics.vacuum.backward} we define 
\begin{align*}
       \widetilde{\Phi}^{n}(t,s, x,\pv) \eqdef &~ \int_s^t d\tau~
    R(F^{n})(\tau,\widetilde{X}(t,\tau,x,\pv),\widetilde{P}(t,\tau,\pv)),
       \quad \widetilde{\Phi}^{n}(t) = \widetilde{\Phi}^{n}(t,1, x,\pv). 
\end{align*}
Then using the backward characteristics as in \eqref{characterisitics.vacuum.backward}  we can write \eqref{rBoltz.eqn0.approx} as 
\begin{multline}\label{positivity.duhamel}
            F^{n+1}(t,x,\pv) 
    =  
   e^{-\widetilde{\Phi}^{n}(t)} (U(t,1)F_1)(x,\pv) 
   \\
    + 
    \int_1^t d\tau~ e^{-\widetilde{\Phi}^{n}(t,\tau)}
    (U(t,\tau){Q_+^{n}})(\tau,x,\pv).
\end{multline}
Now we use \eqref{positivity.duhamel} to give the proof of the positivity of a solution to \eqref{rBoltz.eqn0.approx}.

Indeed given $F^{n} \ge 0$, then clearly $Q_+^{n} = Q_+(F^{n},F^{n})\geq 0$.  Then \eqref{positivity.duhamel} and $F_1(x,\pv)\geq 0$ immediately implies that $F^{n+1}(t,x,\pv) \geq 0$. From \eqref{first.step} we have that $ F^{0}(t,x,\pv) \geq 0$.  Thus, by induction  $F^{n+1}(t,x,\pv)  \ge 0$, for all $n\ge 0$ if $F_1\ge 0$.  

Next we will show that in the limit as $n\to\infty$ that 
\begin{equation}\label{non.neg.limit}
    F^{n+1}(t,x,\pv) \to F(t,x,p)= J(\tqq \pv) + \sqrt{J(\tqq \pv)} f(t,x,p)\ge 0,
\end{equation}
and that this is the same local-in-time solution we previously obtained. We prove convergence in the space $L^1_{{k}}L^\infty_{T_0}L^2_\pv$. However, due to the embedding
$$
L^1_{{k}}L^\infty_{T_0}L^2_\pv\subset L^\infty([1,T_0]\times\T^3_x;L^2_\pv)
$$
we also have convergence in $L^\infty((0,T_0)\times\T^3;L^2_\pv)$. Thus we have the almost everywhere non-negativity in the limit \eqref{non.neg.limit}.  To prove the convergence, for 
\begin{equation*}
    F^{n+1}(t,x,\pv) = J(\tqq \pv) + \sqrt{J(\tqq\pv)} f^{n+1}(t,x,\pv),
\end{equation*}
we can rewrite \eqref{rBoltz.eqn0.approx} as
\begin{multline}\label{iternation.non.neg}
       \left( \partial_t +\frac{\pv^i}{|t^\qq\pv|} \partial_{x^i} 
    - \frac{2 \qq}{t} p^i \partial_{p^i} + t^{-3\qq} \nu_0\right) f^{n+1} 
    \\
    =
    \Gamma_+(f^{n},f^{n})
    -
    \Gamma_-(f^{n+1},f^{n})
    +
\compK (f^{n}).  
\end{multline}
Here, $\compK$ is defined in \eqref{compactK} and then instead of \eqref{gamma0} we use 
\begin{align*}
\Gamma_+ (h_1,h_2)
&\eqdef
t^{3\qq}\int_{\mathbb{R}^3} ~  dq 
\int_{\mathbb{S}^{2}} ~ d\omega
~ v_{\Emptyset} ~  
 \sqrt{J(\tqq\qv)}
 h_1(p^{\prime })h_2(q^{\prime}),
\\
\Gamma_- (h_1,h_2)
&\eqdef
t^{3\qq}\int_{\mathbb{R}^3} ~  dq 
\int_{\mathbb{S}^{2}} ~ d\omega
~ v_{\Emptyset} ~  
 \sqrt{J(\tqq\qv)}
h_1(p)h_2(q).
\end{align*}
Similar to \eqref{approx.solution}, 
the iteration scheme \eqref{iternation.non.neg} satisfies the uniform estimate in \eqref{induction hypo}  and it is  a Cauchy sequence with the same proofs.  Thus the limit function of \eqref{iternation.non.neg} is  a weak solution of  \eqref{rBoltz00} satisfying  \eqref{loces}.  Since we have uniqueness, then the two solutions coincide.  The proof of Theorem \ref{local.thm} is complete. 
\end{proof}

\subsection{Global existence near the equilibrium}\label{subsglbalexistence}

We will now prove the global-in-time uniform decay estimate \eqref{globalestim} from Theorem \ref{Torus Existence} with $l=0$. 

\begin{theorem}
    Let \(\qq\in [0,1]\). Under the assumptions of Theorem \ref{Torus Existence}, there is a unique global solution \(f\) of \eqref{rBoltz00}. Moreover, the following estimate holds,
    \begin{align}\label{uniform.global}
\forall T\geq 1,\qquad     {\Vert{\tTq f}\Vert_{L^1_kL^\infty_TL^2_{\pv}}}
+
{\norm{t^{3\qq/2} f}_{L^1_kL^2_TL^2_{\pv}}}
\lesssim 
\norm{ {f}_1}_{L^1_kL^2_{\pv}}.
\end{align}
\end{theorem}

\begin{proof}
To this end, we first estimate the last term on the right side of \eqref{abcest.torus1}.  From \eqref{fourier.convolution.nonlin}, we have 
\begin{align}\notag
\irnp{\hat{\Gamma}(\hat{f},\hat{f})(k)}{\tTq{\mathfrak{e}_{\ell}(\tqq \pv)}}
=
\int_{\Z^3} \irnp{ \Gamma(\hat{f}(k-l),\hat{f}(l))}{\tTq{\mathfrak{e}_{\ell}(\tqq \pv)}}  d\Sigma(l).
\end{align}
Then from Lemma \ref{L8.7} we have the bound 
\begin{align*}
    \left\vert  \irnp{\hat{\Gamma}(\hat{f},\hat{f})(t,k)}{ \tTq{\mathfrak{e}_{\ell}(\tqq \pv)}} \right\vert^2
    \lesssim
   \left( \int_{\Z^3} \| \hat{f}(k-l)\|_{L^2_{\pv}} \| \hat{f}(l)\|_{L^2_{\pv}}  d\Sigma(l)\right)^2.
\end{align*}
Then exactly as in \eqref{minkowski.ineq} and the estimates below it, we obtain that 
\begin{multline*}
 \int_{\mathbb{Z}^3_k} \left( \int^{T}_{1} t^{5\qq}\left\vert  \irnp{\hat{\Gamma}(\hat{f},\hat{f})(t,k)}{ \tTq{\mathfrak{e}_{\ell}(\tqq \pv)}} \right\vert^2 dt \right)^{\frac12} d\Sigma (k)
 \\
    \lesssim
    \norm{\tTq{f}}_{L^1_kL^\infty_TL^2_{\pv}} \norm{t^{3\qq/2}{f}}_{L^1_kL^2_TL^2_{\pv}}.
\end{multline*}
In fact, a lower time weight is possible in the above estimate.
We place this estimate into \eqref{abcest.torus1} to obtain
\begin{align}\label{ABC.est.n}
\norm{t^{3\qq/2} [\mathcal{A},\mathcal{B},\mathcal{C}]}_{L^1_k L^2_T} 
&\lesssim
\Vert \tTq f \Vert_{L^1_k L^\infty_T L^2_{\pv}} +
\Vert t^{3\qq/2}\micro f \Vert_{L^1_k L^2_T L^2_{\pv}} + \Vert f_1 \Vert_{L^1_k L^2_{\pv}} \\ \notag
&\qquad+ 
    \norm{\tTq{f}}_{L^1_kL^\infty_TL^2_{\pv}} \norm{t^{3\qq/2}{f}}_{L^1_kL^2_TL^2_{\pv}}.
\end{align}
We multiply a small constant to the estimate of $\norm{t^{3\qq/2} [\mathcal{A},\mathcal{B},\mathcal{C}]}_{L^1_k L^2_T} $ in \eqref{ABC.est.n}  and then add that to \eqref{est.fRHS} to achieve that
\begin{align*}
    {\Vert{\tTq f}\Vert_{L^1_kL^\infty_TL^2_{\pv}}}
+
{\norm{t^{3\qq/2} f}_{L^1_kL^2_TL^2_{\pv}}}
\lesssim 
\norm{ {f}_1}_{L^1_kL^2_{\pv}}
+{\norm{\tTq f}_{L^1_kL^\infty_TL^2_{\pv}}} {\norm{t^{3\qq/2} f}_{L^1_kL^2_TL^2_{\pv}}}.
\end{align*}
Thus, if we start with a sufficiently small $\norm{ {f}_1}_{L^1_kL^2_{\pv}}$, then ${\norm{\tTq f}_{L^1_kL^\infty_TL^2_{\pv}}}$ will remain small for a short period of time, and we obtain \eqref{uniform.global}. Then, the above estimate \eqref{uniform.global} can be propagated for long times. This provides our global uniform decay estimate for any $0 \leq \qq \leq 1$. Then, we obtain the global-in-time existence and uniqueness of the solution using a standard continuity argument.
\end{proof}

\subsection{Faster time-decay rates for \texorpdfstring{$0 < \qq \leq \frac{1}{3}$}{}} \label{subfastdecay}

Let us now prove the faster decay estimate \eqref{decayestimate} in Theorem \ref{Torus Existence} with the weights \eqref{weight.decay}$_1$ and \eqref{weight.decay}$_2$ when $m=0$.

\begin{proposition}
    Let $0 < \qq \leq \frac{1}{3}$ and fix $0<\tc<1$. Under the assumptions of Theorem \ref{Torus Existence}, there holds \begin{align}\label{estimated.nonlin.second}
          \forall T\geq 1,\qquad   {\Vert{\tTq \tw(t) f}\Vert_{L^1_kL^\infty_TL^2_{\pv}}}
+
{\norm{t^{3\qq/2} \tw(t) f}_{L^1_kL^2_TL^2_{\pv}}}
\lesssim
\norm{ {f}_1}_{L^1_kL^2_{\pv}}.
\end{align}
\end{proposition}

\begin{proof}
 First, we use \eqref{L0} and Proposition \ref{LEM:nuASYMP} to write \eqref{eqn.complex} as
\begin{multline*}
       \frac{1}{2}\frac{d}{dt}  \left( \tQQq
\Vert \hat{f}\Vert_{L^2_{\pv}}^2(t,k)  \right)
+
\tc\frac{\nu_0}{t^{3\qq}}\left( \tQQq
\Vert \hat{f}\Vert_{L^2_{\pv}}^2(t,k)  \right)
+
(1-\tc)\frac{\nu_0}{t^{3\qq}}\left( \tQQq
\Vert \hat{f}\Vert_{L^2_{\pv}}^2(t,k)  \right)
\\
= \tQQq\real{\irnp{\hat{\Gamma}(\hat{f},\hat{f})}{\hat{f}}}
-
\tQQq\real{\irnp{\compK\hat{f}}{\hat{f}}}. 
\end{multline*}
Now we recall the time weight \eqref{weight.decay} for the case of $0 < \qq \leq \frac{1}{3}$ as
\begin{equation}\notag
    \tw(t) =
    \begin{cases}
        \exp\left(\tc\frac{\nu_0 t^{1-3\qq}}{1-3\qq}\right),\quad &\text{if} \,\,0 < \qq<\frac{1}{3},\\
        t^{\tc\nu_0},\qquad &\text{if} \,\,\qq=\frac{1}{3}.
    \end{cases}
\end{equation}
Then similar to \eqref{ODE.faster.decay} we have 
\begin{multline*}
       \frac{1}{2}\frac{d}{dt}  \left( \tQQq \tw(t)^2
\Vert \hat{f}\Vert_{L^2_{\pv}}^2(t,k)  \right)
+
(1-\tc)\frac{\nu_0}{t^{3\qq}}\left( \tQQq \tw(t)^2
\Vert \hat{f}\Vert_{L^2_{\pv}}^2(t,k)  \right)
\\
= \tQQq \tw(t)^2\left\{\real{\irnp{\hat{\Gamma}(\hat{f},\hat{f})}{\hat{f}}}
-
\real{\irnp{\compK\hat{f}}{\hat{f}}} \right\}. 
\end{multline*}
Next integrating from $[1,t]$ as in \eqref{pro.tpm.p1} we have
\begin{multline*}
    \frac{1}{2} \tQQq \tw(t)^2\Vert \hat{f}\Vert_{L^2_{\pv}}^2(t,k)
    +
    (1-\tc)\nu_0\int^t_1 \tau^{3\qq} \tw(\tau)^2
\Vert \hat{f}(\tau,k)\Vert_{L^2_{\pv}}^2  d\tau
\\ 
    =
    \frac{1}{2}\tw(1)\Vert \hat{f}_1(k)\Vert_{L^2_{\pv}}^2
+\int^t_1 \tau^{6\qq}
\tw(\tau)^2\left(\real{\irnp{\hat{\Gamma}(\hat{f},\hat{f})}{\hat{f}}}
-
\real{\irnp{\compK\hat{f}}{\hat{f}}} \right)
d\tau.
\end{multline*}
Then following the arguments in \eqref{sqrt.ineq} and \eqref{v-int} then instead of \eqref{pro.tpm.p2} we obtain for some  uniform constant $C_1>0$ that
\begin{multline*}
        {\Vert{\tTq \tw(t) f}\Vert_{L^1_kL^\infty_TL^2_{\pv}}}
+
{\norm{t^{3\qq/2} \tw(t) f}_{L^1_kL^2_TL^2_{\pv}}}
\leq C_1
\norm{ {f}_1}_{L^1_kL^2_{\pv}}
\\
+ C_1
\int_{\Z^3} 
\left(\int^T_1 \tau^{6\qq} \tw(\tau)^2 \left|\irnp{\compK\hat{f}}{\hat{f}}\right|
d\tau \right)^{\frac12}
d\Sigma(k)
\\
+ C_1
\int_{\Z^3} 
\left(\int^T_1 \tau^{6\qq} \tw(\tau)^2 \left|{\irnp{\hat{\Gamma}(\hat{f},\hat{f})}{\hat{f}}}\right|
d\tau \right)^{\frac12}
d\Sigma(k).
\end{multline*}
We next estimate the terms involving $\irnp{\compK\hat{f}}{\hat{f}}$ and ${\irnp{\hat{\Gamma}(\hat{f},\hat{f})}{\hat{f}}}$.

We first estimate ${\irnp{\hat{\Gamma}(\hat{f},\hat{f})}{\hat{f}}}$ exactly as in \eqref{pro.tpm.p3}.  Then similarly to \eqref{pro.tpm.p4} we apply Cauchy--Schwarz's inequality to $\int_{1}^{T}(\cdot)dt$ and use Young's inequality with  small constant $\eta>0$ to obtain
\begin{align}
&\int_{\Z^3_k}\left(\int^{T}_{1} \tQQq \tw(t)^2\norm{\hat{f}(t,k)}_{L^2_{\pv}}\int_{\Z^3_l}\Vert \hat{f}(t,k-l)\Vert_{L^2_{\pv}} \norm{\hat{f}(t,l)}_{L^2_{\pv}}  d\Sigma(l)dt\right)^{\frac12}d\Sigma(k)
\notag
\\
&\leq \int_{\Z^3_k} \left(\int^{T}_{1} \tQQq \tTq \tw(t)^2\left(\int_{\Z^3_l} \Vert \hat{f}(t,k-l)\Vert_{L^2_{\pv}} \norm{\hat{f}(t,l)}_{L^2_{\pv}}d\Sigma(l)\right)^2dt\right)^{1/4}
\notag
\\
&\qquad\qquad\qquad\times \left(\int^{T}_{1} \tTq \tw(t)^2\norm{\hat{f}(t,k)}_{L^2_{\pv}}^2 dt\right)^{1/4} d\Sigma(k)
\notag
\\
&\leq \eta\int_{\Z^3_k}\left(\int^{T}_{1} \tTq \tw(t)^2\norm{\hat{f}(t,k)}_{L^2_{\pv}}^2 dt\right)^{\frac12}d\Sigma(k)
\label{pro.tpm.pD}
\\ \notag
&\qquad+\frac{1}{4\eta} \int_{\Z^3_k}\left(\int^{T}_{1} \tQQq \tTq \tw(t)^2\left(\int_{\Z^3_l} \Vert \hat{f}(t,k-l)\Vert_{L^2_{\pv}} \norm{\hat{f}(t,l)}_{L^2_{\pv}}d\Sigma(l)\right)^2dt\right)^{\frac12} d\Sigma(k).
\end{align}
For the second term on the right-hand side of \eqref{pro.tpm.pD},
we use \eqref{minkowski.ineq} to obtain
\begin{multline*}
\left(\int^{T}_{1} \tQQq \tTq \tw(t)^2\left(\int_{\Z^3_l} \Vert \hat{f}(t,k-l)\Vert_{L^2_{\pv}} \norm{\hat{f}(t,l)}_{L^2_{\pv}}d\Sigma(l)\right)^2dt\right)^{\frac12}
\\
\leq
\int_{\Z^3_l} \left(\int^{T}_{1} \tQQq \Vert \hat{f}(t,k-l)\Vert_{L^2_{\pv}}^2 \tTq \tw(t)^2\norm{\hat{f}(t,l)}_{L^2_{\pv}}^2 dt\right)^{\frac12}d\Sigma(l)
\\
\leq 
\int_{\Z^3_k} \int_{\Z^3_l} \sup_{0\le \tau\le T} \left( \tau^{3\qq}\Vert \hat{f}(\tau,k-l) \Vert_{L^2_{\pv}}\right) \left(\int^{T}_{1} \tTq\tw(t)^2\norm{\hat{f}(t,l)}_{L^2_{\pv}}^2 dt\right)^{\frac12}d\Sigma(l) d\Sigma(k).
\end{multline*}
Then Fubini's theorem and translation invariance imply that
\begin{multline*}
    \int_{\Z^3_k} \int_{\Z^3_l} \sup_{0\le \tau\le T} \left( \tau^{3\qq}\Vert \hat{f}(\tau,k-l) \Vert_{L^2_{\pv}}\right) \left(\int^{T}_{1} \tTq \tw(t)^2\norm{\hat{f}(t,l)}_{L^2_{\pv}}^2 dt\right)^{\frac12}d\Sigma(l) d\Sigma(k)
\\
=\Vert \tTq{f}\Vert_{L^1_kL^\infty_TL^2_{\pv}}
\Vert t^{3\qq/2}\tw(t) {f}\Vert_{L^1_kL^2_TL^2_{\pv}}.
\end{multline*}
Then the estimate above with \eqref{pro.tpm.pD} proves the following uniform estimate
\begin{multline}\notag 
    \int_{\Z^3} 
\left(\int^T_1 \tau^{6\qq} \tw(\tau)^2 \left|{\irnp{\hat{\Gamma}(\hat{f},\hat{f})}{\hat{f}}}\right|
d\tau \right)^{\frac12}
d\Sigma(k)
\\
\lesssim
\eta\Vert t^{3\qq/2}\tw(t) {f}\Vert_{L^1_kL^2_TL^2_{\pv}}
+
\frac{1}{4\eta}
\Vert \tTq{f}\Vert_{L^1_kL^\infty_TL^2_{\pv}}
\Vert t^{3\qq/2}\tw(t) {f}\Vert_{L^1_kL^2_TL^2_{\pv}}.
\end{multline}
This will be our main estimate for ${\irnp{\hat{\Gamma}(\hat{f},\hat{f})}{\hat{f}}}$.

Then first choosing $\eta>0$ sufficiently small and second choosing $\Vert \tTq{f_1}\Vert_{L^1_kL^\infty_TL^2_{\pv}}$ sufficiently small so that $\Vert \tTq{f}\Vert_{L^1_kL^\infty_TL^2_{\pv}}$ remains sufficiently small, we obtain for some different uniform constant $C_1>0$ that 
\begin{multline}\label{estimated.nonlin}
        {\Vert{\tTq \tw(t) f}\Vert_{L^1_kL^\infty_TL^2_{\pv}}}
+
{\norm{t^{3\qq/2} \tw(t) f}_{L^1_kL^2_TL^2_{\pv}}}
\leq C_1
\norm{ {f}_1}_{L^1_kL^2_{\pv}}
\\
+ C_1
\int_{\Z^3} 
\left(\int^T_1 \tau^{6\qq} \tw(\tau)^2 \left|\irnp{\compK\hat{f}}{\hat{f}}\right|
d\tau \right)^{\frac12}
d\Sigma(k).
\end{multline}
We next estimate the term including $\irnp{\compK\hat{f}}{\hat{f}}$.   As in Proposition \ref{noKestimate} we split $\compK = \compK_c + \compK_s$.  For the $\compK_s$ part, from Proposition \ref{noKestimate}, for any small $\eta>0$ we have the estimate 
\begin{multline*}
    \int_{\Z^3} 
\left(\int^T_1 \tau^{6\qq} \tw(\tau)^2 \left|\irnp{\compK_s\hat{f}}{\hat{f}}\right|
d\tau \right)^{\frac12}
d\Sigma(k)
\\
\leq \eta
    \int_{\Z^3} 
\left(\int^T_1 \tau^{3\qq} \tw(\tau)^2  \|  \hat{f}(\tau,k )\|_{L^2_p}^2
d\tau \right)^{\frac12}
d\Sigma(k)
=
\eta {\norm{t^{3\qq/2} \tw(t) f}_{L^1_kL^2_TL^2_{\pv}}}.
\end{multline*}
Then, for the $\compK_c$ part, we use \eqref{compact.improved.est} with $\spr =\frac13$ to obtain
\begin{align}\notag 
    \tTq |\langle  \compK_c h_1, h_2\rangle |
\lesssim t^{-\qq} \|\ind_{\le R} h_1 \|_{L^2_p} \| \ind_{\le R} h_2 \|_{L^2_p}.
\end{align}
Then we have 
\begin{multline*}
    \int_{\Z^3} 
\left(\int^T_1 \tau^{6\qq} \tw(\tau)^2 \left|\irnp{\compK_c\hat{f}}{\hat{f}}\right|
d\tau \right)^{\frac12}
d\Sigma(k)
\\
\leq C_{\eta}
    \int_{\Z^3} 
\left(\int^T_1 \tau^{3\qq-\qq} \tw(\tau)^2  \|  \hat{f}(\tau,k )\|_{L^2_p}^2
d\tau \right)^{\frac12}
d\Sigma(k).
\end{multline*}
Next we suppose that $T>0$ is a large time and we notice that the constant $C_{\eta}>0$ above is independent of $T$.  Then  fix a small $\delta>0$ and choose a time $T_1 < T$ such that we have
\begin{align*}
    C_{\eta} \tau^{-\qq}< \delta, \quad \forall \tau \geq T_1.
\end{align*}
On the other hand region we bound 
\begin{align*}
   \tau^{-\qq/2}  \tw(\tau)
    \leq   \tw(T_1)\eqdef C_\delta, \quad \forall 1\leq \tau < T_1.
\end{align*}
Collecting these estimates we obtain for any small $\delta>0$ that 
\begin{multline*}
    \int_{\Z^3} 
\left(\int^T_1 \tau^{6\qq} \tw(\tau)^2 \left|\irnp{\compK_c\hat{f}}{\hat{f}}\right|
d\tau \right)^{\frac12}
d\Sigma(k)
\\
\leq 
\delta {\norm{t^{3\qq/2} \tw(t) f}_{L^1_kL^2_TL^2_{\pv}}}
+
C_\delta {\norm{t^{3\qq/2}  f}_{L^1_kL^2_TL^2_{\pv}}}.
\end{multline*}
The large constant $C_\delta>0$ is uniform for $T>T_1$.  We collect these estimates to obtain
\begin{multline*}
    \int_{\Z^3} 
\left(\int^T_1 \tau^{6\qq} \tw(\tau)^2 \left|\irnp{\compK\hat{f}}{\hat{f}}\right|
d\tau \right)^{\frac12}
d\Sigma(k)
\\
\leq 
\eta {\norm{t^{3\qq/2} \tw(t) f}_{L^1_kL^2_TL^2_{\pv}}}
+
\delta {\norm{t^{3\qq/2} \tw(t) f}_{L^1_kL^2_TL^2_{\pv}}}
+
C_\delta {\norm{t^{3\qq/2}  f}_{L^1_kL^2_TL^2_{\pv}}}.
\end{multline*}
We plug this estimate into \eqref{estimated.nonlin}  choosing $\eta, \delta>0$ sufficiently small to obtain
\begin{align}\notag
            {\Vert{\tTq \tw(t) f}\Vert_{L^1_kL^\infty_TL^2_{\pv}}}
+
{\norm{t^{3\qq/2} \tw(t) f}_{L^1_kL^2_TL^2_{\pv}}}
\leq C_2
\norm{ {f}_1}_{L^1_kL^2_{\pv}}
+ C_2
{\norm{t^{3\qq/2}  f}_{L^1_kL^2_TL^2_{\pv}}}.
\end{align}
Here $C_2>0$ is some large uniform constant.  Lastly we apply the inequality \eqref{uniform.global} to the above to obtain the uniform decay estimate \eqref{estimated.nonlin.second}.
This proves the faster decay rates represented by the weight \eqref{weight.decay}$_1$ when $m=0$. The case of the weight \eqref{weight.decay}$_2$ when $m=0$ is proven by the same method.
\end{proof}

\subsection{Exponential large time decay for \texorpdfstring{$\qq =0$}{}} \label{subsexpnonexp}

We now prove the large time decay estimates \eqref{decayestimate} in the non-expanding case when $\qq =0$.

\begin{proposition}
    Let \(\qq=0\). Under the assumptions of Theorem \ref{Torus Existence}, there is a small uniform constant \(\lambda>0\) such that
    \begin{align}
    \forall T \geq 1,\qquad \norm{e^{\lambda t} f}_{L^1_kL^\infty_TL^2_{\pv}}
+
{\norm{e^{\lambda t} f}_{L^1_kL^2_TL^2_{\pv}}}
\lesssim 
\Vert {f}_1\Vert_{L^1_kL^2_{\pv}}.\label{exponentialdecaynorm}
\end{align}
\end{proposition}

\begin{proof}
For a small constant $\lambda>0$ we consider $\hat{h}(t,k,\pv) = e^{\lambda t} \hat{f}(t,k,\pv)$ with ${h}_1=e^{\lambda}{f}_1$.  Then we can re-write \eqref{torus fourier transformed} as
\begin{align}\notag 
\left(\partial_t  + i\frac{\pv}{|\tq \pv|}\cdot k 
-\frac{2\qq }{t} p^i \partial_{p^i} \right) \hat{h}
+\linL\hat{h}=e^{-\lambda t}\hat{\Gamma}(\hat{h},\hat{h})+ \lambda \hat{h}.
\end{align}
Then following the proof of \eqref{est.fRHS} we obtain
\begin{align*}
    \norm{h}_{L^1_kL^\infty_TL^2_{\pv}}
+
{\norm{\micro h}_{L^1_kL^2_TL^2_{\pv}}}
\lesssim 
\Vert {f}_1\Vert_{L^1_kL^2_{\pv}}
+{\Vert{ h}\Vert_{L^1_kL^\infty_TL^2_{\pv}}} {\norm{ h}_{L^1_kL^2_TL^2_{\pv}}}
+\lambda {\norm{ h}_{L^1_kL^2_TL^2_{\pv}}}.
\end{align*}
Furthermore, directly following the proof of \eqref{abcest.torus1} and \eqref{ABC.est.n}, for a sufficiently small constant $\lambda>0$ when $\qq =0$ we have 
\begin{align}\notag
\norm{ e^{\lambda t} [\mathcal{A},\mathcal{B},\mathcal{C}]}_{L^1_k L^2_T} 
&\lesssim
\Vert h \Vert_{L^1_k L^\infty_T L^2_{\pv}} +
\Vert \micro h \Vert_{L^1_k L^2_T L^2_{\pv}} + \Vert h_1 \Vert_{L^1_k L^2_{\pv}} \\ \notag
&\qquad+ 
    \norm{h}_{L^1_kL^\infty_TL^2_{\pv}} \norm{h}_{L^1_kL^2_TL^2_{\pv}}.
\end{align}
Collecting these estimates for $\Vert {f}_1\Vert_{L^1_kL^2_{\pv}}$ and $\lambda>0$ sufficiently small we obtain the decay estimate \eqref{exponentialdecaynorm}. This proves the exponential time decay when $\qq=0$.
\end{proof}

\subsection{Time-decay of weighted norms}\label{rapid.decay.solution}

Let us prove the decay estimates for weighted norms \eqref{decayestimate} using the weight $w_m(\tqq\pv)$ defined in \eqref{weight.fcn.def}.

\begin{proposition}\label{prop:weightenergy}
    Let \(\qq\in [0,1]\) and \(m\geq 0\). Under the assumptions of Theorem \ref{Torus Existence}, there holds \begin{align}
\forall t\geq 1, \qquad\Vert |t^{2\qq}\pv|^{m/2}{f}(t)\Vert_{L^1_kL^2_{\pv}}
\lesssim  t^{-3\qq}\Vert w_m{f}\Vert_{L^1_kL^2_{\pv}}(1).\label{decayweighted}
\end{align}
\end{proposition}

\begin{proof}
We multiply \eqref{torus fourier transformed} by $w_m(\tqq\pv)^2$ in \eqref{weight.fcn.def} and use the invariance in \eqref{nullRHStp} to obtain
\begin{multline}\notag
   \left(\partial_t  + i\frac{\pv}{|\tq \pv|}\cdot k 
-\frac{2\qq }{t} p^i \partial_{p^i} \right) \left( w_m(\tqq\pv)^2 \hat{f}(t,k,\pv)\right)
\\
+w_m(\tqq\pv)^2\linL\hat{f}(t,k,\pv)=w_m(\tqq\pv)^2\hat{\Gamma}(\hat{f},\hat{f})(t,k,\pv). 
\end{multline}
Then following the proof of Proposition \ref{Torus Boltzmann Micro} we obtain instead of \eqref{est.fRHS} that
\begin{multline}
\norm{\tTq w_mf}_{L^1_kL^\infty_TL^2_{\pv}}
+
\int_{\Z^3}\left(\int^T_1  \tau^{6\qq} w_m(\tau^{2\qq}\pv)^2\real{\irnp{\linL\hat{f}}{\hat{f}}}(t,k)d\tau\right)^{\frac12} d\Sigma(k)
\\
\lesssim 
\Vert w_m {f}\Vert_{L^1_kL^2_{\pv}}(1)
+{\Vert{\tTq w_mf}\Vert_{L^1_kL^\infty_TL^2_{\pv}}} {\norm{t^{3\qq/2} w_m f}_{L^1_kL^2_TL^2_{\pv}}}.
\notag
\end{multline}
We apply now Corollary \ref{2.10} to obtain
\begin{multline}\notag
\norm{\tTq w_mf}_{L^1_kL^\infty_TL^2_{\pv}}
+
{\norm{t^{3\qq/2} w_m f}_{L^1_kL^2_TL^2_{\pv}}}
\\
\lesssim 
\Vert w_m {f}\Vert_{L^1_kL^2_{\pv}}(1)
+{\Vert{\tTq w_mf}\Vert_{L^1_kL^\infty_TL^2_{\pv}}} {\norm{t^{3\qq/2} w_m f}_{L^1_kL^2_TL^2_{\pv}}}
+
\norm{t^{3\qq/2}f}_{L^1_kL^2_TL^2_{\pv}}.
\end{multline}
Next we use \eqref{uniform.global} to obtain $\norm{t^{3\qq/2}f}_{L^1_kL^2_TL^2_{\pv}} \lesssim 
\norm{ {f}_1}_{L^1_kL^2_{\pv}} \lesssim \Vert w_m {f}\Vert_{L^1_kL^2_{\pv}}(1)$.  Then if we further start with a  sufficiently small $\Vert w_m {f}\Vert_{L^1_kL^2_{\pv}}(1)$ norm, then the norm ${\Vert{\tTq w_mf}\Vert_{L^1_kL^\infty_TL^2_{\pv}}}$ will remain small for a short period of time, so that we obtain
\begin{equation}
    \norm{\tTq w_mf}_{L^1_kL^\infty_TL^2_{\pv}}
+
{\norm{t^{3\qq/2} w_m f}_{L^1_kL^2_TL^2_{\pv}}}
\lesssim 
\Vert w_m {f}\Vert_{L^1_kL^2_{\pv}}(1).
\label{torus.ft.weight}
\end{equation}
Then we obtain from \eqref{torus.ft.weight} the rapid decay rates stated in \eqref{decayweighted}. 
\end{proof}

\begin{remark}
The weights $w_{m}(\tqq \pv)$ can be directly added to the proofs of the faster time decay rates in the previous subsections to obtain the estimates in \eqref{decayestimate}. 
\end{remark}

\subsection{Propagation of spatial regularity}\label{sec.spatial.regularity.solution}

We next prove the propagation of spatial regularity for solutions to \eqref{rBoltz00}--\eqref{idff}. We study the Boltzmann equation in Fourier space as in \eqref{torus fourier transformed}. And for $l>0$, we multiply \eqref{torus fourier transformed} by $\langle k \rangle^{l}$, and prove analogous global uniform bounds to propagate this frequency weight and thereby show the propagation of spatial regularity. 

\begin{proposition}\label{prop:spatreg}
       Let \(\qq\in [0,1]\). Under the assumptions of Theorem \ref{Torus Existence}, there holds 
\begin{align}\label{weight.propagate}
\forall T\geq 1, \qquad    {\Vert{\tTq f}\Vert_{L^1_{k,l}L^\infty_TL^2_{\pv}}} 
+ 
\|t^{3\qq/2} f\|_{L^1_{k,l}L^2_TL^2_p}
\lesssim 
\Vert  {f}_1\Vert_{L^1_{k,l}L^2_{\pv}}.
\end{align}
\end{proposition}

The following proof follows closely the argument in \cite{1904.12086}.

\begin{proof}
Following the proofs of Lemma \ref{lem: Boltzmann nonlinear} and Theorem \ref{thmmacroscop}, one can prove that
\begin{multline*}
        \int_{\Z^3}\bigg(\int_{1}^{T} \big|\irnp{\hat{\Gamma}(\hat{f},\hat{g})}{\langle  k \rangle^{2l}\,\hat{h}}\big| dt \bigg)^{\frac{1}{2}}d\Sigma(k)
        \leq C_{\eta} \|f\|_{L^1_{k,l}L^{\infty}_TL^2_p}\|g\|_{L^1_{k,l}L^{2}_TL^2_p}
        \\
+C_{\eta} \|f\|_{L^1_{k,l}L^{2}_TL^2_p}\|g\|_{L^1_{k,l}L^{\infty}_TL^2_p}+\eta\|h\|_{L^1_{k,l}L^{2}_TL^2_p},
\end{multline*}
and also using the argument from \eqref{ABC.est.n} we have
\begin{multline*}
\norm{t^{3\qq/2} [\mathcal{A},\mathcal{B},\mathcal{C}]}_{L^1_{k,l} L^2_T} \lesssim
\Vert t^{3\qq/2}\micro f \Vert_{L^1_{k,l} L^2_T L^2_{\pv}} +\Vert \tTq f \Vert_{L^1_{k,l} L^\infty_T L^2_{\pv}} 
\\ 
+ \Vert f_1 \Vert_{L^1_{k,l} L^2_{\pv}} 
+ 
\norm{\tTq{f}}_{L^1_{k,l}L^\infty_TL^2_{\pv}} \norm{t^{3\qq/2}{f}}_{L^1_{k,l}L^2_TL^2_{\pv}}.
\end{multline*}
Next, taking the \(L^2_p\) inner product of \eqref{torus fourier transformed} with the complex conjugate of \(\langle k \rangle^{2l}\hat{f}\) and then taking the real part, we obtain, similarly to \eqref{eqn.complex}, that
\begin{multline}\notag
        \frac{1}{2}\frac{d}{dt}  \left( \tQQq \langle k \rangle^{2l}
\Vert \hat{f}\Vert_{L^2_{\pv}}^2(t,k)  \right)+ \tQQq\real{\irnp{\linL\hat{f}}{\langle k \rangle^{2l}\hat{f}}}
\\
= \tQQq\real{\irnp{\hat{\Gamma}(\hat{f},\hat{f})}{\langle k \rangle^{2l}\hat{f}}}.
\end{multline}
Then following exactly the proof of Proposition \ref{Torus Boltzmann Micro}, for any small \(\eta>0\), we have 
\begin{multline*}
{\Vert{\tTq f}\Vert_{L^1_{k,l}L^\infty_TL^2_{\pv}}} 
+ 
\|t^{3\qq/2}\micro f\|_{L^1_{k,l}L^2_TL^2_p}
\\
\lesssim 
\Vert  {f}_1\Vert_{L^1_{k,l}L^2_{\pv}}
+ 
\int_{\Z^3} \left(\int^T_1 \tTq\big|\irnp{\hat{\Gamma}(\hat{f},\hat{f})}{\langle k \rangle^{2l}\micro\hat{f}}\big|\,dt\right)^{\frac12}d\Sigma(k)
\\ 
\lesssim 
\Vert  {f}_1\Vert_{L^1_{k,l}L^2_{\pv}}
+
C_{\eta}
{\Vert{\tTq f}\Vert_{L^1_{k,l}L^\infty_TL^2_{\pv}}} 
\Vert t^{3\qq/2} {f}\Vert_{L^1_{k,l}L^2_TL^2_{\pv}}
+
\eta
\|t^{3\qq/2}\micro f\|_{L^1_{k,l}L^2_TL^2_p}.
\end{multline*}
For $\Vert  {f}_1\Vert_{L^1_{k,l}L^2_{\pv}}$ sufficiently small, after a short time ${\Vert{\tTq f}\Vert_{L^1_{k,l}L^\infty_TL^2_{\pv}}}$ remains small.  Collecting the estimates above, we obtain for $\eta>0$ sufficiently small that the estimate \eqref{weight.propagate} holds. This establishes the propagation of spatial regularity for any $l>0$ and for $\Vert {f}_1\Vert_{L^1_{k,l}L^2_{\pv}}$ sufficiently small.
\end{proof}

\section{Global stability of the vacuum solution}\label{sec.sol.nearby.vacuum}

In this section, we will prove Theorem \ref{thmvacuum}. In Section \ref{subsecexistuniqvac}, we first show the local well-posedness result Theorem \ref{thmvacuum} is based on. In Section \ref{subsglbalexistvac}, we show global existence and the uniform decay estimate \eqref{apriorivacuumestimate0} in Theorem \ref{thmvacuum}. 

\subsection{Local existence and uniqueness around the vacuum}\label{subsecexistuniqvac}

We state here local existence and uniqueness of solutions for the massless Boltzmann equation on FLRW spacetime \eqref{rBoltz00} around vacuum.

\begin{theorem}[Local existence and uniqueness near vacuum]\label{local.thm2}
Let \(\qq\in [0,1]\). Under the assumptions of Theorem \ref{thmvacuum}, there is a small $\epsilon_0>0$, a $T_0=T_0(\epsilon_0)>0$ and some  $C_0=C_0(T_0)>0$ such that if 
\begin{equation}\notag 
\| f_{1}\|_{L_{k}^1L^2_{\pv}}\leq\epsilon_0,
\end{equation}
then the massless Boltzmann equation on FLRW spacetime around vacuum, \eqref{perturb.massless.rBoltz0},
 admits a unique local-in-time solution
$$
f(t,x,\pv),~~ \ 0\leq t\leq T_0<\infty, ~~\ x\in  \T^3,~~\ \pv\in \R^3,
$$
which satisfies the  uniform estimate
\begin{equation}\notag 
\| \tTq f\|_{L^1_{{k}}L^\infty_{T_0}L^2_{\pv}} 
\leq C_{0}
\|f_1\|_{L_{{k}}^1L^2_{\pv}}.
\end{equation}
If additionally 
$F_1(x,\pv)= J^{\frac{1}{2}}(\pv)f_1(x,\pv)\geq0$,
then the solution satisfies
$$
F(t,{x},\pv)= \sqrt{J(t^{2\qq}p) }f(t,{x},\pv)\geq0.
$$
\end{theorem}

The previous theorem is obtained by the same arguments performed in the proof of Theorem \ref{local.thm} concerning local existence and uniqueness of solutions for the massless Boltzmann around the Maxwell--J\"uttner equilibrium.

\subsection{Global existence near the vacuum solution}\label{subsglbalexistvac}

We will now prove the global-in-time uniform decay estimate \eqref{apriorivacuumestimate0} from Theorem \ref{thmvacuum}.

\begin{theorem}
    Let \(\qq\in (\frac{1}{3},1]\). Under the assumptions of Theorem \ref{thmvacuum}, there is a unique global solution \(f\) of \eqref{perturb.massless.rBoltz0}. Moreover, the following estimate holds,
\begin{equation}\notag 
\forall T\geq 1,\qquad {\norm{t^{3\qq} f}_{L^1_kL^\infty_TL^2_{\pv}}}
\lesssim 
\norm{f_1}_{L^1_kL^2_{\pv}}.  
\end{equation}
\end{theorem}

\begin{proof}
Now take the Fourier transform of \eqref{perturb.massless.rBoltz0} to obtain
\begin{align}\label{vac.torus.fourier.transformed}
\left(\partial_t  + i\frac{\pv}{|\tq \pv|}\cdot k 
-\frac{2\qq }{t} p^i \partial_{p^i} \right) \hat{f}(t,k,\pv)
=\hat{\Gamma}(\hat{f},\hat{f})(t,k,\pv).
\end{align}
Next take  the product of  \eqref{vac.torus.fourier.transformed} with the complex conjugate of $\hat{f}(t,k,\pv)$, denoted $\bar{\hat{f}}$, and then take the real part of the resulting equation to obtain
\begin{align*}
\frac{1}{2}\left(\frac{d}{dt}   
-\frac{2\qq }{t} p^i \partial_{p^i} \right)
\vert \hat{f}(t,k,\pv)\vert^2
=\real{\bar{\hat{f}}\hat{\Gamma}(\hat{f},\hat{f})}.
\end{align*}
Next integrate the above with respect to $\pv$ to obtain
\begin{align*}
    \frac{1}{2}\left(\frac{d}{dt}   
+\frac{6\qq }{t}  \right)\Vert \hat{f}\Vert_{L^2_{\pv}}^2(t,k)
= \real{\irnp{\hat{\Gamma}(\hat{f},\hat{f})}{\hat{f}}}.
\end{align*}
From \eqref{timeDeriv.pq} this becomes
\begin{align*}
    \frac{1}{2}\frac{d}{dt}  \left( \tQQq
\Vert \hat{f}\Vert_{L^2_{\pv}}^2(t,k)  \right)= \tQQq\real{\irnp{\hat{\Gamma}(\hat{f},\hat{f})}{\hat{f}}}.
\end{align*}
We integrate in time over $[1,t]$ to obtain
\begin{align}\notag
\frac{1}{2} \tQQq \Vert \hat{f}\Vert_{L^2_{\pv}}^2(t,k)
\leq &~\frac{1}{2}\Vert \hat{f}_1\Vert_{L^2_{\pv}}^2(k)
+\int^t_1 \tau^{6\qq}\real{\irnp{\hat{\Gamma}(\hat{f},\hat{f})}{\hat{f}}}
d\tau.
\end{align}
Take the square root on both sides and use \eqref{sqrt.ineq} to obtain
\begin{align*}
 \tTq \Vert \hat{f}\Vert_{L^2_{\pv}}(t,k)
\leq 
2\left\{ \Vert \hat{f}_1\Vert_{L^2_{\pv}}(k) 
+
\left(\int^t_1 \tau^{6\qq} \left|{\irnp{\hat{\Gamma}(\hat{f},\hat{f})}{\hat{f}}}\right|
d\tau \right)^{\frac12}
\right\}.  
\end{align*}
We take  $\sup_{0\le t\le T}$ on both sides and then integrate the result with respect to $d\Sigma(k)$ over $\Z^3$ to obtain
\begin{align*}
{\norm{t^{3\qq} f}_{L^1_kL^\infty_TL^2_{\pv}}}
\leq 
2\left\{ \norm{f_1}_{L^1_kL^2_{\pv}}
+
\int_{\Z^3} \left(\int^T_1 \tau^{6\qq} \left|{\irnp{\hat{\Gamma}(\hat{f},\hat{f})}{\hat{f}}}(\tau,k)\right|
d\tau \right)^{\frac12} d\Sigma(k)
\right\}.  
\end{align*}
We now estimate the non-linear term in the upper bound.

From \eqref{ineq: Boltzmann nonlinear} we have 
\begin{align}\notag
\left| \irnp{\hat{\Gamma}(\hat{f},\hat{f})(k)}{\hat{f}(k)}\right|\lesssim \norm{\hat{f}(k)}_{L^2_{\pv}}\int_{\Z^3}\Vert \hat{f}(k-l)\Vert_{L^2_{\pv}} \norm{\hat{f}(l)}_{L^2_{\pv}}  \,d\Sigma(l).
\end{align}
We thus conclude that 
\begin{multline*}
   \tau^{6\qq}\left| \irnp{\hat{\Gamma}(\hat{f},\hat{f})(\tau,k)}{\hat{f}(\tau,k)}\right|
   \\
   \lesssim 
   \tau^{-3\qq}{\norm{t^{3\qq} \hat{f}(k)}_{L^\infty_TL^2_{\pv}}}
\int_{\Z^3}  {\norm{t^{3\qq} \hat{f}(k-l)}_{L^\infty_TL^2_{\pv}}}{\norm{t^{3\qq} \hat{f}(l)}_{L^\infty_TL^2_{\pv}}} \,d\Sigma(l). 
\end{multline*}
Then for $\qq > \frac{1}{3}$ we have the following uniform estimate
\begin{multline*}
\left(\int^T_1 \tau^{6\qq} \left|{\irnp{\hat{\Gamma}(\hat{f},\hat{f})}{\hat{f}}}(\tau,k)\right|
d\tau \right)^{\frac12}
\\
\lesssim 
{\norm{t^{3\qq} \hat{f}(k)}^{\frac12}_{L^\infty_TL^2_{\pv}}}
\left(\int_{\Z^3}  {\norm{t^{3\qq} \hat{f}(k-l)}_{L^\infty_TL^2_{\pv}}}{\norm{t^{3\qq} \hat{f}(l)}_{L^\infty_TL^2_{\pv}}} \,d\Sigma(l) \right)^{\frac12}.
\end{multline*}
This holds because if $\qq > \frac{1}{3}$ then we have 
\begin{align*}
  \left|   \int^T_1 \tau^{-3\qq} d\tau \right| \leq  \frac{1}{3\qq-1}\left|  T^{1-3\qq}-1 \right| \leq \frac{2}{3\qq-1}. 
\end{align*}
Now we apply the Cauchy--Schwarz inequality with a small constant $\eta>0$ to obtain 
\begin{multline*}
\int_{\Z^3} \left(\int^T_1 \tau^{6\qq} \left|{\irnp{\hat{\Gamma}(\hat{f},\hat{f})}{\hat{f}}}(\tau,k)\right|
d\tau \right)^{\frac12} \,d\Sigma(k) 
\leq C {\norm{t^{3\qq} {f}}^{\frac32}_{L^1_k L^\infty_TL^2_{\pv}}}
\\
\leq
\eta {\norm{t^{3\qq} {f}}_{L^1_k L^\infty_TL^2_{\pv}}}
+
C_\eta {\norm{t^{3\qq} {f}}^{2}_{L^1_k L^\infty_TL^2_{\pv}}}.
\end{multline*}
We conclude for some constant $C_0>0$ that we have 
\begin{align*}
{\norm{t^{3\qq} f}_{L^1_kL^\infty_TL^2_{\pv}}}
\leq 
C_0\left\{ \norm{f_1}_{L^1_kL^2_{\pv}}
+
{\norm{t^{3\qq} {f}}^{2}_{L^1_k L^\infty_TL^2_{\pv}}}
\right\}.  
\end{align*}
Thus for $\qq > \frac{1}{3}$ and for $\norm{f_1}_{L^1_kL^2_{\pv}}$ chosen sufficiently small, then ${\norm{t^{3\qq} f}_{L^1_kL^\infty_TL^2_{\pv}}}$ will remain small for a short time.   Thus we obtain the following uniform estimate 
\begin{align}\notag 
{\norm{t^{3\qq} f}_{L^1_kL^\infty_TL^2_{\pv}}}
\leq 
2C_0 \norm{f_1}_{L^1_kL^2_{\pv}}.  
\end{align}
This provides the global uniform decay estimate for any $\qq \in (\frac{1}{3},1]$. Finally, global-in-time existence and uniqueness of the solution follows by a standard continuity argument.
\end{proof}

\begin{remark}
    Let \(\qq\in [0,1]\), and recall the weight \(\tiq(t)\) defined in \eqref{tiq.weight}--\eqref{weightflowq1}. Then, we have the following invariant quantity for the massless Boltzmann equation,
\begin{multline}\label{collisionalINVARIANCE}
    |\pv'| \left|x+\tiq(\tau)\left(\frac{\pv}{ |\pv|}-\frac{\pv'}{ |\pv'|}\right)\right|^2 
    +
|\qv'| \left|x+\tiq(\tau)\left(\frac{\pv}{ |\pv|}-\frac{\qv'}{ |\qv'|}\right)\right|^2 
  \\
  =
|\pv| |x|^2
+
|\qv| \left|x+\tiq(\tau)\left(\frac{\pv}{ |\pv|}-\frac{\qv}{ |\qv|}\right)\right|^2. 
\end{multline}
Then, in the whole space $\R^3$, the invariant \eqref{collisionalINVARIANCE} can be used to quantify the dispersion and prove a global-in-time a priori estimate for near vacuum initial data in the style of \cite{MR0760333} and many subsequent works.  In this direction, we refer to \cite{GL1996}.  
\end{remark}

\appendix

\section{Conservation laws and H-theorem}\label{appconservlaws}

In this appendix we show the conservation in time of mass, energy, and momentum, and also the H-theorem.  We first define the symmetrised collision operator \[Q^*(F,G):=\frac{t^{3\qq}}{2}\int_{\mathbb{R}^3}\int_{\mathbb{S}^2} v_{\Emptyset}\sigma \Big(F(p')G(q')+F(q')G(p')-
F(p)G(q)-F(q)G(p)\Big) d\omega dq.\] Note that the collision term satisfies \(Q(F,F)=Q^*(F,F)\).

\begin{lemma}
For any sufficiently regular functions \(\varphi(p)\), \(F(p)\), and \(G(p)\) decaying at infinity, we have 
\begin{multline*}
       2\int_{\R^3_p} Q^*(F,G)\varphi dp
       \\
=t^{3\qq}\int_{\mathbb{R}^3}\int_{\mathbb{R}^3}\int_{\mathbb{S}^2} v_{\Emptyset}\sigma\big(F(p')G(q')+F(q')G(p')\big)\begin{pmatrix}
        \varphi(p) \\ \varphi (q) \\ -\varphi (p')\\ -\varphi (q')
    \end{pmatrix} d\omega dqdp
    \\
-t^{3\qq}\int_{\mathbb{R}^3}\int_{\mathbb{R}^3}\int_{\mathbb{S}^2} v_{\Emptyset}\sigma\big(F(p)G(q)+F(q)G(p)\big)\begin{pmatrix}
        \varphi(p) \\ \varphi (q) \\ -\varphi (p')\\ -\varphi (q')
    \end{pmatrix} d\omega dqdp.
\end{multline*}
\end{lemma}

\begin{proof}
The first equation follows by definition. The second follows by using the change of variables \((p,q,\omega)\mapsto (p,q,-\omega)\). For this, we note that \(\varrho\sqrt{s}=\varrho^2=g(p-q,p-q)\), that under this change of variables \((p',q')\mapsto (q',p')\), and the symmetry \(\varrho(p,q)=\varrho(q,p)\). Next, in the first equation we use the change of variables \((p,q)\mapsto (p',q')\), so
\begin{align*}
2\int_{\R^3_p} &Q^*(F,G)\varphi (p)dp=t^{3\qq}\int_{\mathbb{R}^3}\int_{\mathbb{R}^3}\int_{\mathbb{S}^2} \frac{\varrho^2 \sigma(\varrho,\theta)}{|t^{\mathfrak{q}}p'||t^{\mathfrak{q}}q'|}
\Big(F(p')G(q')+F(q')G(p')\\
&\qquad
-F(p(p',q'))G(q(p',q'))-F(q(p',q'))G(p(p',q'))\Big)\varphi(p(p',q'))d\omega dq'dp',
\end{align*}
where we used the Jacobian determinant computed in the previous lemma. Moreover, we note that \(\varrho\) and \(\theta\) are collision-invariants. We invert the equations to get explicitly \(p=p(p',q')\) and \(q=q(p',q')\). Now, we rename \((p',q')\) to \((p,q)\) to get 
\begin{align*}
2\int_{\R^3_p} Q^*(F,G)\varphi (p)dp&=t^{3\qq}\int_{\mathbb{R}^3}\int_{\mathbb{R}^3}\int_{\mathbb{S}^2} \frac{\varrho^2 \sigma(\varrho,\theta)}{|t^{\mathfrak{q}}p||t^{\mathfrak{q}}q|}
\Big(F(p)G(q)+F(q)G(p)\\
&\qquad\qquad\qquad\quad
-F(p')G(q')-F(q')G(p')\Big)\varphi(q')d\omega dqdp.
\end{align*}
For the last equation, we use the change of variables \((p,q,\omega)\mapsto(p,q,-\omega)\). Note that under this change \((p',q')\mapsto (q',p')\).
\end{proof}

\begin{corollary}\label{corconservationlaws}
There holds \(\int Q(F,F)dp=\int p^iQ(F,F)dp=\int |p|Q(F,F)dp=0\) for \(i\in \{1,2,3\}\). In particular, the mass, momentum, and energy are conserved quantities for solutions of the massless Boltzmann equation on FLRW spacetime.
\end{corollary} 

\begin{proof}
Adding the 4 expressions in the previous lemma, we have 
\begin{align*}
\int_{\R^3_p} Q(F,F)\varphi (p)dp&=\frac{t^{3\qq}}{4}\int_{\mathbb{R}^3}\int_{\mathbb{R}^3}\int_{\mathbb{S}^2} \frac{\varrho^2 \sigma(\varrho,\theta)}{|t^{\mathfrak{q}}p||t^{\mathfrak{q}}q|}
\Big(F(p')F(q')-F(p)F(q)\Big)\\
&\qquad \qquad\qquad \qquad \qquad \cdot \big(\varphi(p)+\varphi(q)-\varphi(p')-\varphi(q')\big)d\omega dqdp.
\end{align*}
Therefore, if \(\varphi(p')+\varphi(q')=\varphi(p)+\varphi(q)\), then \(\int  Q(F,F)\varphi (p)dp=0.\) In particular, one can take \(\varphi(p)=1\), \(\varphi(p)=p^i\), and \(\varphi(p)=|p|\).
\end{proof}

\begin{theorem}[H-theorem for massless Boltzmann on FLRW]\label{prophthm}
For any regular solution of the massless Boltzmann equation on FLRW spacetime, we have
\[\frac{d}{dt}\bigg(\int_{\mathbb{T}^3_x}\int_{ \mathbb{R}^3_p} (- t^{6\mathfrak{q}}F\log F)(t,x,p) dpdx\bigg)\geq 0.\]
\end{theorem}

\begin{proof}
Adding the 4 expressions in the lemma above, and letting \(\phi=1+\log F\),
\begin{align*}
&\int_{\R^3_p} Q(F,F)(1+\log F)dp=
\frac{t^{3\qq}}{4}\int_{\mathbb{R}^3}\int_{\mathbb{R}^3}\int_{\mathbb{S}^2} \frac{\varrho^2 \sigma(\varrho,\theta)}{|t^{\mathfrak{q}}p||t^{\mathfrak{q}}q|}
\Big(F(p')F(q')-F(p)F(q)\Big)\\
&\qquad\qquad\qquad \qquad\qquad\qquad \cdot\big(\log F(p)+\log F(q)-\log F(p')-\log F(q')\big)d\omega dqdp\\
&\quad=\frac{t^{3\qq}}{4}\int_{\mathbb{R}^3}\int_{\mathbb{R}^3}\int_{\mathbb{S}^2} \frac{\varrho^2 \sigma(\varrho,\theta)}{|t^{\mathfrak{q}}p||t^{\mathfrak{q}}q|}
F(p')F(q') (1-\mu ) \log (\mu) d\omega dqdp,
\end{align*}
where in the last line we wrote \(\mu=F(p)F(q)F(p')^{-1}F(q')^{-1}.\) And since \((1-\mu)\log \mu\leq 0\) for all \(\mu>0\), we conclude that \(\int Q(F,F)\log F dp \leq 0.\) Finally, we get
\begin{multline*}
    \frac{d}{dt}\bigg(t^{6\mathfrak{q}}\int_{\mathbb{T}^3_x\times \mathbb{R}^3_p} F\log F dxdp\bigg)
    \\
    =6\mathfrak{q}t^{6\mathfrak{q}-1}\int_{\mathbb{T}^3_x\times \mathbb{R}^3_p} F\log F dxdp+t^{6\mathfrak{q}}\int_{\mathbb{T}^3_x\times \mathbb{R}^3_p} \partial_t F(1+\log F) dxdp
    \\
=6\mathfrak{q}t^{6\mathfrak{q}-1}\int_{\mathbb{T}^3_x\times \mathbb{R}^3_p} (F\log F) dxdp
\\
-t^{6\mathfrak{q}}\int_{\mathbb{T}^3_x\times \mathbb{R}^3_p} \Big[\frac{p^i}{|t^{\qq}p|} \partial_{x^i} F -\frac{2\qq}{t}p^i \partial_{p^i} F-Q(F,F)\Big](1+\log F) dxdp
\\
=t^{6\mathfrak{q}}\int_{\mathbb{T}^3_x\times \mathbb{R}^3_p} Q(F,F)\log F dxdp\leq 0.
\end{multline*}
This completes the proof.
\end{proof}

\section{Determination of the Maxwell--J\"uttner parameters}\label{MJparamet}

In this section, the Maxwell--J\"uttner parameters are identified for fixed values of the initial mass, energy, and momentum.

\begin{theorem}[Determination of Maxwell--J\"uttner parameters from conserved quantities]\label{thm:mjparamet}
    Given a sufficiently regular initial distribution \(F_1\colon \mathbb{T}^3_x\times \mathbb{R}^3_p\to [0,+\infty)\), there are parameters \(a\in \mathbb{R}\), \(c\in [0,+\infty)\), and \(b\in \mathbb{R}^3\) such that \(c>|b|\), for which 
\(
J_{a,b,c} (p)
=(8\pi)^{-1}e^{a+b \cdot  p-c|p|} 
\)
satisfies 
\begin{multline*}
  0=\int_{\mathbb{T}^3_x}\int_{ \mathbb{R}^3_p} (F_1-J_{a,b,c})dpdx= \int_{\mathbb{T}^3_x}\int_{ \mathbb{R}^3_p} p^i(F_1-J_{a,b,c})dpdx\\
= \int_{\mathbb{T}^3_x}\int_{ \mathbb{R}^3_p} |p|(F_1-J_{a,b,c})dpdx.   
\end{multline*}
\end{theorem}
   
The proof of Theorem \ref{thm:mjparamet} follows from the propositions obtained below.

\begin{proposition}[Mass, momentum and energy of the Maxwell--J\"uttner equilibria]
    Let \(a\), \(c\in \mathbb{R}\), and \(b\in \mathbb{R}^3\) such that \(c>|b|\).  The mass, momentum, and energy of the Maxwell--J\"uttner equilibrium \(J_{a,b,c}\) are
    \begin{align*}
    t^{6\qq}\int_{\mathbb{T}^3_x}\int_{ \mathbb{R}^3_p} J_{a,b,c} (t^{{2\qq}}p)dpdx=\frac{e^a c}{(c^2-|b|^2)^2},\\ t^{6\qq}\int_{\mathbb{T}^3_x}\int_{ \mathbb{R}^3_p} t^{2\qq}p^iJ_{a,b,c} (t^{{2\qq}}p)dpdx=\frac{4e^acb^i}{(c^2-|b|^2)^3},\\ t^{6\qq}\int_{\mathbb{T}^3_x}\int_{ \mathbb{R}^3_p} |t^{2\qq}p|J_{a,b,c} (t^{{2\qq}}p)dpdx=e^a\frac{(3c^2+|b|^2)}{(c^2-|b|^2)^3}.
    \end{align*}
\end{proposition}

\begin{proof}
    First, we compute the mass of \(J_{a,b,c} (t^{{2\qq}}p)\) by making the change of variables \(q=t^{2\qq}p\) to get
    \begin{align*}
            t^{6\qq}\int_{\mathbb{T}^3_x}\int_{ \mathbb{R}^3_p} J_{a,b,c} (t^{{2\qq}}p)dpdx
            &=\frac{t^{6q}e^a}{8\pi}\int_{\mathbb{R}^3_p} e^{b\cdot (t^{2\qq}p)-c|t^{2\qq}p|}dp
            \\
            &=\frac{e^a}{8\pi}\int_{\mathbb{R}^3_q} e^{b\cdot q-c|q|}dq=\frac{e^a c}{(c^2-|b|^2)^2},
    \end{align*}
 where the last integral was computed using spherical coordinates. The momentum and energy of \(J_{a,b,c} (t^{{2\qq}}p)\) are computed in a similar way.
\end{proof}

Then, we need to solve for \(a\), \(c\in \mathbb{R}\), and \(b\in \mathbb{R}^3\) in the following system of equations
\begin{align}
    A=\frac{e^a c}{(c^2-|b|^2)^2},\qquad B^i=\frac{4e^acb^i}{(c^2-|b|^2)^3},\qquad C=e^a\frac{(3c^2+|b|^2)}{(c^2-|b|^2)^3}.\label{eqnsabc}
\end{align}

\begin{proposition}
    Let \(A\), \(C\in (0,+\infty)\), and \(B\in [0,+\infty)^3\) such that \(C\geq |B|\). The nonlinear system of equations \eqref{eqnsabc} has a unique solution \(a\in \mathbb{R}\), \(c\in (0,+\infty)\), and \(b\in \mathbb{R}^3\) such that \(c>|b|\). In fact, we have
    \begin{align*}
        e^a&=\frac{864A^4}{(C+\Delta)(4(C+\Delta)^2-9|B|^2)},\quad\qquad b^i=\frac{36AB^i}{4(C+\Delta)^2-9|B|^2},\\
        c&=\frac{24A(C+\Delta)}{4(C+\Delta)^2-9|B|^2},\qquad \qquad \qquad \quad\Delta=\sqrt{C^2-\frac{3}{4}|B|^2}.
    \end{align*}
\end{proposition}

\begin{proof}
    By the condition \(C\geq |B|\) there holds \(\Delta\geq \frac{1}{4}|B|^2\geq 0\). Moreover, we get from the equations \eqref{eqnsabc} that \[4(C+\Delta)^2-9|B|^2=\frac{144e^ac^2}{(c^2-|b|^2)^5}>0.\] Finally, the result can be shown by direct computations.
\end{proof}

\providecommand{\bysame}{\leavevmode\hbox to3em{\hrulefill}\thinspace}
\providecommand{\href}[2]{#2}

\end{document}